\documentclass[11pt]{article}
 \usepackage{a4wide}
\usepackage{amssymb}
\usepackage{stmaryrd}

\usepackage[Postscript=dvips]{diagrams}
\usepackage{QED}

\newtheorem{theorem}{Theorem}[section]
\newtheorem{lemma}[theorem]{Lemma}

\newtheorem{proposition}[theorem]{Proposition}

\newtheorem{corollary}[theorem]{Corollary}

\newcommand{\even}[1]{\mbox{$#1^{\tt even}$}}
\newcommand{\E}{{\cal E}}
\newcommand{\G}{{\cal G}}
\newcommand{\M}{{\cal M}}
\newcommand{\D}{{\cal D}}
\newcommand{\Siep}{{\Sigma}}
\newcommand{\Ip}[1]{\lesssim_{#1}}
\newcommand{\Mat}[4]{\left( \begin{array}{cc}
#1 & #2 \\
#3 & #4
\end{array} \right)}
\newcommand{\odd}[1]{\mbox{$#1^{\tt odd}$}}

\newcommand{\rat}{\triangledown}
\newcommand{\ident}[1]{{\tt id}_{#1}}
\newcommand{\Fst}{{\tt fst}}
\newcommand{\Snd}{{\tt snd}}
\newcommand{\Mod}{{\cal M}}

\newcommand{\Nat}{{\bf Nat}}

\newcommand{\plays}[1]{M^{\circledast}_{#1}}

\newcommand{\ITEM}[1]{\begin{itemize} #1 \end{itemize}}
\newcommand{\SET}[1]{\{ #1 \}}
\newcommand{\Fr}{\rightarrow}
\newcommand{\Sqleq}{\sqsubseteq}
\newcommand{\Rest}{{\upharpoonright}}
\newcommand{\Eqdef}{\succeq}
\newcommand{\La}{\lambda}
\newcommand{\Pfr}{\rightharpoonup}
\newcommand{\Deq}{\approx}
\newcommand{\Over}[1]{\overline{ #1 }}

\newcommand{\Sum}{\sum}
\newcommand{\Hat}[1]{\widehat{ #1}}
\newcommand{\Incl}{\subseteq}
\newcommand{\B}[1]{{\tt #1 }}
\newcommand{\THEN}{\; \Longrightarrow \;}
\newcommand{\IFF}{\;\mbox{ iff }\;}
\newcommand{\ofcourse}{\mbox{$!$}}

\newcommand{\with}{\mbox{$\&$}}

\newcommand{\tensor}{\mbox{$\otimes$}}
\newcommand{\linimpl}{\mbox{$\multimap$}}
\newcommand{\Games}{\mbox{$\cal{G}$}}
\newcommand{\KG}{K_{\ofcourse }({\cal G})}
\newcommand{\tunit}{I}

\newcommand{\pcfc}{\mbox{$\mbox{PCF}c$}}

\newcommand{\finevalT}[2]{\mbox{{\bf FET}($#1,#2$)}}

\newcommand{\twotrans}[2]{\mbox{$\frac{\mbox{$#1$}}{\mbox{$#2$}}$}}

\newcommand{\lpo}{\mbox{$\sqsubseteq$}}

\newcommand{\evalT}[2]{\mbox{{\bf ET}($#1,#2$)}}

\newcommand{\varempty}{\varnothing}
\newcommand{\lang}{\langle}
\newcommand{\rang}{\rangle}
\newcommand{\Sx}{\sqsubseteq}
\newlength{\sqpreordheight}
\newlength{\sqpreorddepth}
\newcommand{\Subeq}%
   {\mathbin{\raisebox{-1.02ex}[\sqpreordheight][\sqpreorddepth]%
                            {$\stackrel{\textstyle \sqsubset}{\approx}$}}}

\begin{document}









\newcommand{\affil}[1]{\footnote{#1}}

\title{Full Abstraction for PCF\footnote{This
research was supported by grants from UK SERC and ESPRIT Basic
Research Action 6811 ``CLICS II''. Radha Jagadeesan was supported
in part by grants from NSF and  ONR. Pasquale Malacaria was
supported in part by the HCM fellowship n. ERBCHBICT940947. }}



%

\author{Samson.Abramsky\footnote{University of Oxford, email: samson.abramsky@cs.oxac.uk}
\, and Radha
Jagadeesan\footnote{DePaul University
Chicago, email: RJagadeesan@cs.depaul.edu} 
\, and Pasquale
Malacaria\footnote{Queen Mary University of London, email: p.malacaria@qmul.ac.uk}
}







\maketitle

\abstract{An intensional model for the programming language PCF is
described, in which the types of PCF are interpreted by games, and
the terms by certain ``history-free'' strategies. This model is
shown to capture definability in PCF. More precisely, every
compact strategy in the model is definable in a certain simple
extension of PCF. We then introduce an intrinsic preorder on
strategies, and show that it satisfies some striking properties,
such that the intrinsic preorder on function types coincides with
the pointwise preorder. We then obtain an order-extensional fully
abstract model of PCF by quotienting the intensional model by the
intrinsic preorder. This is the first syntax-independent
description of the fully abstract model for PCF. (Hyland and Ong
have obtained very similar results by a somewhat different route,
independently and at the same time).

We then consider the effective version of our model, and prove a
Universality Theorem: every element of the effective extensional
model is definable in PCF. Equivalently, every recursive strategy
is definable up to observational equivalence. }






\section{Introduction}

The Full Abstraction Problem for  PCF \cite{PlotkinGD:lcfcpl,MilnerR:fulamt,BerryG:fulasl,CurienPL:seqful} is one
of the longest-standing problems in the semantics of programming languages.
There is quite widespread agreement that it is one of the most difficult;
there is much less agreement as to what exactly the problem is, or more
particularly as to the precise criteria for a solution.
The usual formulation is that one wants a ``semantic characterization''
of the fully abstract model (by which we mean the inequationally
fully abstract order-extensional model, which Milner proved to be uniquely
specified up to isomorphism by these properties \cite{MilnerR:fulamt}).
The problem is to understand what should be meant by a 
``semantic characterization''.

Our view is that the essential content of the problem, what makes it
important, is that it calls for a semantic characterization of
{\em sequential, functional computation at higher types}.
The phrase ``sequential functional computation'' deserves careful
consideration. On the one hand, sequentiality refers to a computational
process extended over time, not a mere function; on the other hand,
we want to capture just those sequential computations in which the
different parts or ``modules'' interact with each other in a purely
functional fashion.

There have, to our knowledge, been just four models of PCF put forward
as embodying some semantic analysis. Three are domain-theoretic:
the ``standard model'' based on Scott-continuous functions \cite{PlotkinGD:lcfcpl};
Berry's bidomains model based on stable functions \cite{BerryG:modcom};
and the Bucciarelli-Ehrhard model based on strongly stable functions \cite{EhrhardT:extemb}.
The fourth is the Berry-Curien model based on sequential algorithms
\cite{BerryG:seqacd}.\footnote{Cartwright and Felleisen's model without error
values turns out to be equivalent to the sequential algorithms model
\cite{FelleisenM:obsseq,CurienPL:obssac}.
The main result in \cite{FelleisenM:obsseq,CurienPL:obssac} is that
the sequential algorithms model with errors is fully abstract for
SPCF, an extension of PCF with a {\tt catch} construct and errors. This
is a fine result, but SPCF has a rather different flavour to PCF, and
arguably is no longer purely functional in character.}
Of these, we can say that the standard model gives a good account
of functional computation at higher types, but fails to capture
sequentiality, while the sequential algorithms model gives a good
analysis of sequential computation, but fails to capture functional
behaviour. In each case, the failure can calibrated in terms of
{\em definability}: the standard model includes parallel functions;
the sequential algorithms model includes algorithms which compute
``functionals'' which are sensitive to non-functional aspects of the
behaviour of their arguments. The bidomains model also contains
non-sequential functions; while the strongly stable model, in the
light of a recent result by Ehrhard \cite{EhrhardT:prosas}, can be seen as the
``extensional collapse'' of the sequential algorithms model.
In short, all these models are unsatisfactory because they contain ``junk''.
On the other side of the coin, we have Milner's result that an
order-extensional
model is fully-abstract iff all its compact elements are definable.

\subsection*{Intensional Full Abstraction}
This suggests that the key step towards solving the Full Abstraction problem
for PCF is to capture PCF definability.
This motivates the following definition. A model ${\cal M}$
(not necessarily extensional) is {\em intensionally fully abstract}
if it is algebraic, and all its compact elements are definable in PCF.
In support of this terminology, we have the fact that the fully abstract
model can be obtained from an intensionally fully abstract model ${\cal M}$
in the following canonical fashion. Firstly, define a logical relation
on ${\cal M}$ induced by the ordering on the ground types
(which are assumed {\em standard}, i.e. isomorphic to the usual
flat domains of natural numbers and booleans). Because of the
definability properties of ${\cal M}$, this relation is a preorder at
all types.
In particular, it is {\em reflexive} at all types. This says that
all elements of the model have extensional (functional) behaviour---there
is no junk.

We can now apply Theorem 7.2.2 of \cite{StoughtonA:fulamp} to conclude
that ${\cal M}$ 
can be collapsed by a continuous homomorphism to the fully abstract model.
In short, the fully abstract model is the extensional collapse of
any intensionally fully abstract model. Moreover, note that the collapsing
map is  a homomorphism, and in particular preserves application. This
contrasts sharply with ``collapses'' of the standard model to obtain the
fully abstract model, as in the work of Mulmuley \cite{MulmuleyK:fulase} and
Stoughton and Jung \cite{JungA:stufam}, which are only homomorphic on the
``inductively reachable'' subalgebra.

Thus we propose that a reasonable factorization of the full abstraction
problem is to look for  a semantic presentation of an intensionally
fully abstract model, which embodies a semantic analysis of sequential
functional computation. 
The construction of such a model is our first main result; it is described
in Sections~2 and~3.

We have explained how the (order-extensional, inequationally) fully abstract
model can be obtained from any intensionally fully abstract model
by means of a general construction, described in \cite{StoughtonA:fulamp}. 
However, this description of the fully
abstract model leaves something to be desired. Firstly, just because the
construction in \cite{StoughtonA:fulamp} is very general, it is
unlikely to yield any useful information
about the fully abstract model. Secondly, it is not entirely syntax-free:
it refers to the {\sl type structure} of PCF.

What would the ideal form of description of the fully abstract model be?
We suggest that it should comprise the specification of a cartesian closed
category whose objects are certain cpo's, given together with certain
additional ``intensional'' structure, to be used to characterize sequentiality;
and whose morphisms are continuous functions between these cpo's---not
{\sl all} continuous functions, of course, but only the sequential ones,
as determined by the intensional structure. The interpretation of PCF
generated from this category should then be the fully abstract model.
Most of the attempts at solving the full abstraction problem of which we
are aware, including Berry's bidomains, Curien's bicds, and Bucciarelli
and Erhard's strongly stable functions, clearly fall within this general
scheme. (Thus for example the intensional structure in bidomains is the
stable ordering; for domains with coherence it is the coherence.)

In Section~4, we will explain how the category of games described in Section~2 
does indeed give rise to a category of sequential domains in exactly this
sense. This yields the first syntax-independent description of the
fully abstract model for PCF.  

A still more stringent requirement on a description of the fully abstract model
is that it should yield effective methods for deciding observation
equivalence on terms. For example, consider ``Finitary PCF'', i.e. PCF
based on the booleans rather than the natural numbers. The interpretation
of each type of Finitary PCF in the fully abstract model is a finite poset.
A natural question is whether these finite posets can be
effectively
presented. Suppose that we have a category of sequential domains as described
in the previous paragraph, yielding a fully abstract model of PCF.
If the ``intensional structure'' part of the interpretation of each type
could itself be specified in a finite, effective fashion, then such a model
would immediately yield a positive solution to this problem.
Because of its intensional character, our model does not meet this requirement:
there are infinitely many strategies at each functional type of Finitary PCF.
The same point occurs in one form or another with all the currently known
descriptions of the fully abstract model for PCF.
A remarkable result by Ralph Loader \cite{Loa96} shows that this is in fact 
inevitable.
Loader proved that observation equivalence for Finitary PCF is undecidable.
This shows that an intensional description of the fully abstract model is
the best that we can hope to do.
\subsection*{Related Work}
The results in the present paper were obtained in June 1993 (the results on
Intensional Full Abstraction in Section 3) and September 1993 
(the results on the intrinsic
preorder and (extensional) Full Abstraction in Section 4).
They were announced on various electronic mailing lists in June and September
1993. An extended  abstract of the present paper appeared in the Proceedings
of the Second Symposium on Theoretical Aspects of Computer Science, which
was held in Sendai in April 1994 \cite{AbramskyS:fulap}.

Independently, and essentially simultaneously, Martin Hyland and Luke Ong
gave a different model construction, also based on games and strategies,
which led to the same model of PCF, and essentially the same results on 
Intensional
Full Abstraction. Following our work on the intrinsic preorder, they
showed that similar results held for their model. What is interesting
is that such similar results have been obtained by somewhat different routes.
Hyland and Ong's approach is based on dialogue games and innocent strategies,
in the tradition of Lorentzen's dialogue interpretations of logical proofs
\cite{LorenzenP:logua,LorenzenP:eindk}, and the work by Kleene and Gandy on the
semantics of higher-type 
recursion theory \cite{GandyRO:diabso}, while our approach is closer
to process semantics  
and the Geometry of Interaction
\cite{AbramskyS:gamfcm,MalacariaP:frogi}. Further work is needed to 
understand more fully the relationship between the two approaches.

Independently, Hanno Nickau obtained essentially the same model and results
as Hyland and Ong \cite{NickauH:hersf}. A very different description
of the fully abstract 
model for PCF was obtained by Peter O'Hearn and Jon Riecke, using Kripke logical
relations \cite{OHearnPW:krilrp}. This construction is very
interesting, and probably 
of quite general applicability, but does not appear to us to embody a
specific semantic analysis of sequentiality.

Since the results described in this paper were obtained, there has
been significant further progress in the use of game semantics to give
fully abstract models for programming languages. 
These results all build on the concepts, methods and results developed
in the present paper, and that of Hyland and Ong. For an expository account of 
some of these results, and some references,
see \cite{AM97}; there is an overview in \cite{Abr97}.
The main results of the present paper are recast in an abstract,
axiomatic form in \cite{Abr00}.
There have also been some
significant applications of game semantics, notably \cite{MH99,GM00}.

\section{The Model}
We shall refer to
\cite{AbramskyS:gamfcm} for general background and motivation on game semantics.

We begin by fixing some notation.  If $X$ is a set, we write
$X^{\star}$ for the set of finite sequences (words, strings) on $X$.
We shall use $s$, $t$, $u$, $v$  and primed and subscripted variants
of these to denote sequences, and $a$, $b$, $c$, $d$, $m$, $n$ and
variants to denote
elements of these sequences. Concatenation of sequences will be indicated
by juxtaposition, and we will not distinguish notationally between an element
and the corresponding unit sequence. Thus {\sl e.g.} $as$ denotes a sequence
with first element $a$  and tail $s$.
If $f: X
\rightarrow Y$, then $f^{\star}: X^{\star} \rightarrow Y^{\star}$ is
the unique monoid homomorphism extending $f$.  We write $|s|$ for the
length of a finite sequence, and $s_i$ for the $i$'th element of $s$,
$1\leq i\leq |s|$.
Given a set $S$ of sequences, we write
\even{S} for the subset of even length sequences and \odd{S} for the
subset of odd length sequences.  If $Y \subseteq X$ and $s \in
X^{\star}$, we write $s {\upharpoonright} Y$ for the result of
deleting all occurrences of symbols not in $Y$ from $s$.
We write $ s
\sqsubseteq t$ if $s$ is a prefix of $t$, {\em i.e.} for some $u$, $s u
= t$.  We always consider sequences under this prefix ordering and use
order-theoretic notions~\cite{DaveyBA:Intlor} without further comment.

Given a family of sets $\{ X_i \}_{i \in I}$ we write $\sum_{i \in I} X_i$
for their disjoint union (coproduct); we fix
\[ \sum_{i \in I} X_i  =  \{ (i, x) \mid i \in I, x \in X_i \} \]
as a canonical concrete representation. In particular, we write
$X_{1} + X_{2}$ for $\sum_{i \in \{ 1,2 \}} X_i$.
If $s \in (\sum_{i \in I} X_{i})^{\star}$ and $i \in I$,
we define $s \Rest i \in X_{i}$ inductively by:
\[ \begin{array}{lcl}
\epsilon \Rest i & = & \epsilon \\
((j, a)s) \Rest i & = & \left\{ \begin{array}{ll}
a(s \Rest i), & i = j \\
s \Rest i, & i \neq j.
\end{array} \right.
\end{array} \]
We use $\Fst$ and $\Snd$ as notation for first and second projection functions.
Note that with $s$ as above, $\Fst^{\star}(s)$ is a sequence of indices $i_{1} \cdots
i_{k} \in I^{\star}$ tracking which components of the disjoint union the
successive elements of $s$ are in.

We will also need some notation for manipulating partial functions.
We write $f : X \Pfr Y$ if $f$ is a partial function from the set $X$ to
the set $Y$; and $fx \Eqdef y$ for ``$fx$ is defined and equal to $y$''.
If $f : X \Pfr Y$ is an injective partial function, we write $f^{\ast} : Y \Pfr
X$  for the converse, which is also an injective partial function.
(NB: the reader should beware of confusing $f^{\star}$ with $f^{\ast}$.
In practice, this should not be a problem.)
If $f, g : X \Pfr Y$ are partial functions with disjoint domains
of definition, then we write $f \vee g : X \Pfr Y$  for the
partial function obtained by taking  the union of (the graphs of) $f$ and $g$.
We write $0_{X}$ for the everywhere-undefined partial function on $X$ and
sometimes ${\tt id}_{X}$, sometimes $1_{X}$ for the identity function on $X$.
We shall omit subscripts whenever we think
we can get away with it.

\subsection{Games}
The games we consider are between Player and Opponent.  A {\em play} or {\em
run} of the game consists of an alternating sequence of moves, which may be
finite or infinite.  Our plays are always with Opponent to move first.

A game is a structure $A = (M_A,\lambda_A,P_A,{\Deq}_A)$, where
\begin{itemize}
\item $M_A$ is the set of moves.
\item $\lambda_A: M_A \rightarrow \{ P,O \} \times \SET{Q,A}$ is the
  labelling function.

  The labelling function indicates if a move is by Player (P) or
  Opponent (O), and if a move is a question (Q) or an answer (A). The
  idea is that questions correspond to requests for data, while
  answers correspond to data
  ({\em e.g.} integer or boolean values). In a higher-order
  context, where arguments may be functions which may themselves be
  applied to arguments, all four combinations of Player/Opponent with
  Question/Answer are possible. $\lambda_A$ can be decomposed into
  two functions $\lambda_{A}^{PO}:M_A\rightarrow\{P,O\}$ and
  $\lambda_{A}^{QA}:M_A\rightarrow\{Q,A\}$.

  We write
 \[ \begin{array}{rcl}
 \{ P, O \} \times \{ Q, A \} & = & \{ PQ, PA, OQ, OA \} \\
   \langle \lambda_{A}^{PO},\lambda_{A}^{QA}
  \rangle& = & \lambda_A , \\
  M_A^P & = &  \lambda_A^{-1}(\{P \}\times\SET{Q,A}), \\
  M_A^O & = &
  \lambda_A^{-1}(\{O \}\times\SET{Q,A}), \\
  M_A^Q& = &
  \lambda_A^{-1}(\{P,O \}\times\SET{Q}), \\
  M_A^A & = &
  \lambda_A^{-1}(\{P,O \}\times\SET{A})
  \end{array}
  \]
etc., and define
\[  \overline{P} = O, \;\; \overline{O} = P, \]
\[  \overline{\La_{A}^{PO}}(a)= \overline{\La_{A}^{PO}(a)}, \;\;
  \Over{\La_{A}}= \langle \Over{\La_{A}^{PO}},\La_{A}^{QA} \rangle . \]

\item Let $M_A^{\circledast}$ be the set of all finite sequences $s$ of
  moves satisfying:
\[ \begin{array}{ll}
{\bf (p1)} &  s = at \THEN\ a \in M_A^O \\
{\bf (p2)} &  (\forall i: 1 \leq i
< |s| )\; [ \lambda_{A}^{PO}(s_{i+1}) = \overline{\lambda_{A}^{PO}(s_{i})} ] \\
{\bf (p3)} & (\forall t\Sqleq s)\;(|t\Rest M_{A}^A| \leq |t\Rest M_{A}^Q| ).
\end{array} \]
Then $P_A$, the set of valid {\em positions} of the game, is a
non-empty prefix closed subset of $M_A^{\circledast}$.

The conditions {\bf (p1)}--{\bf (p3)} can be thought of as global rules
applying to all games. {\bf (p1)} says that Opponent moves first, and
{\bf (p2)}  that Opponent and Player alternate. {\bf (p3)} is known as
the {\em bracketing condition}, and can be nicely visualised as follows.
Write each
question in a play as a left parenthesis ``('', and each answer as a
right parenthesis ``)''.  Then the string must be well-formed in the
usual sense, so that each answer is associated with a unique previous
question---the most recently asked, as yet unanswered question.  In
particular, note that a question by Player must be answered by
Opponent, and vice versa.

\item
$\Deq_{A}$ is an equivalence relation on $P_{A}$ satisfying
\[ \begin{array}{ll}
\mbox{{\bf (e1)}} &  s \Deq_{A} t \THEN\ \lambda_{A}^{\star}(s) = \lambda_{A}^{\star}(t) \\
\mbox{{\bf (e2)}} &  s \Deq_{A} t, s' \sqsubseteq s, t' \sqsubseteq t,
|s'| = |t'|
\THEN  s' \Deq_{A} t' \\
\mbox{{\bf (e3)}} &  s \Deq_{A} t, sa \in P_{A} \THEN \exists b. \, sa \Deq_{A} tb .
\end{array} \]
Note in particular that {\bf (e1)} implies that if $s \Deq_{A} t$, then $|s| = |t|$.
\end{itemize}
For example, the game for ${\bf Nat}$ has one possible opening move
$\ast$ (request for data), with $\lambda_{\Nat}(\ast ) = OQ
$; and for each $n \in \omega$, a possible response $\underline{n}$ with
$\lambda_{\Nat}(\underline{n}) = PA $.
${\Deq}_{\Nat}$ is the identity relation on $P_{\Nat}$.
The game for ${\bf Bool}$ is defined similarly.

\subsection{Strategies}
A {\em strategy} for Player in $A$ is a non-empty subset
$\sigma \subseteq \even{P_A}$ such that
$\Over{\sigma} = \sigma \cup \B{dom}(\sigma)$
is prefix-closed, where
\[ \B{dom}(\sigma) = \SET{sa \in P_{A}^{\B{odd}} \mid  \exists b. \
  sab \in  \sigma} . \]
We will be interested in a restricted class of strategies, the history-free
(or history independent, or history insensitive) ones.
A strategy $\sigma$ is {\em history-free} if it satisfies
\begin{itemize}
\item $sab, tac \in \sigma \THEN b=c$
\item $sab, t\in\sigma, ta\in P_A \THEN tab
  \in \sigma$ (equivalently, $ta\in \B{dom}(\sigma)$).
\end{itemize}
Henceforth, ``strategy'' will always by default mean ``history-free strategy''.

Given any strategy $\sigma$, we can define $
\B{fun}(\sigma):M_{A}^O \Pfr M_{A}^P$ by
\[ \B{fun}(\sigma)(a)\Eqdef b \IFF (\exists s)\;[ sab
\in\sigma ] . \]
Conversely, given $f:M_{A}^O\Pfr M_{A}^P$ we can define
$\B{traces}(f)\Incl(M_{
A}^{\circledast})^{\B{even}}$ inductively by:
\[ \B{traces}(f)=\{ \epsilon \} \cup\SET{sab \mid s\in
  \B{traces}(f), sa\in P_A, f(a) \Eqdef b} . \]
We say that $f$ induces the strategy
$\sigma_f=\B{traces}(f)$, if $\B{traces}(f)\Incl P_A$. Note that if
$\tau$ is a strategy, we have
\[ \B{fun}(\sigma_f) \Incl f, \;\;\;\;\;\;
\sigma_{\B{fun}(\tau)} = \tau \]
so there is always a least partial function on moves canonically inducing a
(history-free) strategy.

\begin{proposition}
If $f : M_{A}^{O} \Pfr M_A^{P}$ is any partial function, then
$\B{traces}(f) \subseteq \plays{A}$.
\end{proposition}

\begin{proof} Certainly any $s \in \B{traces}(f)$ satisfies ``$O$
moves first'' and the alternation condition. We show that it
satisfies the bracketing condition by induction on $|s|$. If $s =
tab$, then since $ta \in P_{A}$ and $|ta|$ is odd, the number of
questions in $ta$ must exceed the number of answers; hence $s$
satisfies the bracketing condition.
\end{proof}

The equivalence relation on positions extends to a relation on
strategies, which we shall write as $\Subeq$. \\
$\sigma\Subeq\tau$ iff:
\begin{equation}
\label{sequiv1}
sab\in\sigma, s'\in \tau, sa
  \Deq s'a' \THEN \exists b'.\ [s'a'b'\in\tau \wedge sab\Deq s'a'b'].
\end{equation}

By abuse of notation we write the symmetric closure of this relation
as $\Deq$:
$$ \sigma\Deq \tau \ \IFF \ \sigma\Subeq\tau\wedge \tau\Subeq\sigma.$$

Interpreting the equivalence on positions as factoring out coding conventions,
$\sigma \Deq \tau$ expresses the fact that $\sigma$ and $\tau$
are the same modulo
coding conventions.  $\sigma \Deq \sigma$ expresses a
``representation independence'' property of strategies.

\begin{proposition} [Properties of $\Subeq$] \hfill

$\Subeq$ is a partial preorder relation (i.e. transitive) on
  strategies.
Hence $\Deq$ is a partial equivalence relation (i.e. symmetric and transitive).

\end{proposition}

\begin{proof}  Suppose $\sigma \Subeq \tau$ and $\tau \Subeq
\upsilon$, and $s \in \sigma$, $u \in \upsilon$, $sab \in\sigma $
and $sa \Deq ua''$. By induction on $|sa|$ using the definition of
$\sigma \Subeq \tau$ and (e3), there is $ta'b' \in \tau $ with
$sab \Deq ta'b'$. But then $ta' \Deq ua''$, and since $\tau \Subeq
\upsilon$, $ua''b''\in\upsilon$ with $ta'b'\Deq ua''b''$ and hence
$sab\Deq ta'b'\Deq ua''b''$ as required.
\end{proof}

>From now on, we are only interested in those history-free
strategies $\sigma$ such that $\sigma \Deq \sigma$.We write ${\tt
Str}(A)$ for the set of such strategies over $A$. If $\sigma$ is
such a strategy for a game $A$, we shall write $\sigma : A$. We
write $\hat{A}$ for the set of partial equivalence classes of
strategies on  $A$, which we think of as the set of ``points'' of
$A$. We write $[\sigma ] = \{ \tau \mid \sigma \Deq \tau \}$ when
$\sigma \Deq \sigma$.

\subsection{Multiplicatives}
\paragraph{Tensor}

The game $A \tensor B$ is defined as follows.   We call the games $A$
and $B$ the {\sl component} games.

\begin{itemize}
\item $M_{A \tensor B} = M_A + M_B$, the disjoint union of the two move
sets.
\item $\lambda_{A \tensor B} = [\lambda_A,\lambda_B]$, the source tupling.
\item $P_{A \tensor B}$ is the set of all $s \in \plays{A \tensor B}$
such that:
\begin{enumerate}
\item {\em Projection condition:} The restriction to the moves in $M_A$ (resp. $M_B$) is in
$P_A$ (resp. $P_B$).
\item {\em Stack discipline:} Every answer in $s$ must be in the same component game as the
  corresponding question.
\end{enumerate}

\item $s \Deq_{A \tensor B} t \IFF  s \Rest A \Deq_{A} t \Rest_{A} \;
\wedge \; s \Rest B \Deq_{B} t \Rest B \; \wedge \; \Fst^{\star}(s) =
\Fst^{\star}(t).$
\end{itemize}
We omit the easy proof that $\Deq_{A \tensor B}$ satisfies
{\bf (e1)}--{\bf (e3)}.
Note that, if the equivalence relations $\Deq_{A}$ and $\Deq_{B}$ are
the identities on $P_{A}$ and $P_{B}$ respectively, then $\Deq_{A \tensor B}$
is the identity on $P_{A \tensor B}$.

The tensor unit is given by
\[ \tunit = (\varnothing,\varnothing,\{\epsilon \}, \{ (\epsilon , \epsilon )
\} ). \]

\paragraph{Linear Implication}
The game $A \linimpl B$ is defined as follows.   We call the games $A$
and $B$ the {\sl component} games.

\begin{itemize}
\item $M_{A \linimpl B} = M_A + M_B$, the disjoint union of the two move
sets.
\item $\lambda_{A \linimpl B} = [\overline{\lambda_{A}},\lambda_B]$.
\item $P_{A \linimpl B}$ is the set of all $s \in \plays{A \linimpl B}$
such that:
\begin{enumerate}
\item {\em Projection condition:} The restriction to the moves in $M_A$ (resp. $M_B$) is in
$P_A$ (resp. $P_B$).
\item {\em Stack discipline:} Every answer in $s$ must be in the same component game as the
  corresponding question.
\end{enumerate}

\item $s \Deq_{A \linimpl B} t \IFF  s \Rest A \Deq_{A} t \Rest_{A} \;
\wedge \; s \Rest B \Deq_{B} t \Rest B \; \wedge \; \Fst^{\star}(s) =
\Fst^{\star}(t).$
\end{itemize}
Note that, by {\bf (p1)}, the first move in any position in
$P_{A \linimpl B}$ must be in $B$.

We refer to the condition requiring answers to be given in the same
components as the corresponding questions as the {\sl stack discipline}.
It ensures that computations must evolve in a properly nested fashion.
This abstracts out a key structural feature of functional computation, and plays
an important r\^{o}le in our results.

\begin{proposition}[Switching Condition]
If a pair of successive moves in a position in $A \tensor B$ are in different
components, (i.e. one was in $A$ and the other in $B$),
then the second move was by Opponent (i.e. it was Opponent who
switched components). If two successive moves in $A \linimpl B$ are
in different components, the second move was by Player (i.e. it was Player
who switched components).
\end{proposition}

\begin{proof} Each position in $A \tensor B$ can be classified as in
one of four ``states'': $(O, O)$, i.e. an even number of moves
played in both components, so Opponent to move in both; $(P, O)$,
meaning an odd number of moves played in the first component, so
Player to move there, and an even number of moves played in the
second component, so Opponent to play there; $(O, P)$; and $(P,
P)$. Initially, we are in state $(O, O)$. After Opponent moves, we
are in $(P, O)$ or $(O, P)$, and Player can only move in the same
component that Opponent has just moved in. After Player's move, we
are back in the state $(O, O)$. A simple induction shows that this
analysis holds throughout any valid play, so that we can never in
fact reach a state $(P, P)$, and Player must always play in the
same component as the preceding move by Opponent. A similar
analysis applies to $A \linimpl B$; in this case the initial state
is $(P, O)$, after Opponent's move we are in $(P, P)$, and after
Player's response we are in $(O, P)$ or $(P, O)$. $\;\;\; $
\end{proof}

Note that, by comparison with \cite{AbramskyS:gamfcm}, the Switching Condition is
a consequence of our definition of the multiplicatives rather than
having to be built into it.
This is because of our global condition {\bf (p1)}, which corresponds to
restricting our attention to ``Intuitionistic'' rather than ``Classical'' games.
Note also that the unreachable state $(P, P)$ in $A \tensor B$ is precisely
the problematic one in  the analysis of Blass' game semantics in \cite{AbramskyS:gamfcm}.

\subsection{The Category of Games}
We build a category \Games:
\begin{eqnarray*}
\mbox{Objects} &:& \mbox{Games} \\
\mbox{Morphisms} &:& [\sigma]:A\Fr B \; \mbox{is a partial equivalence
  class} \; [\sigma]\in\Hat{A\linimpl B}
\end{eqnarray*}
We shall write $\sigma : A \Fr B$ to mean that $\sigma$ is a strategy in
$A \linimpl B$ satisfying $\sigma \Deq \sigma$.

There are in general two ways of defining a (history-free) strategy or operation
on strategies: in terms of the representation of strategies as sets of
positions, or via the partial function on moves inducing the strategy.
Some notation
will be useful in describing these partial functions.
Note that the type of the function $f$ inducing a strategy in $A \linimpl B$
is
\[ f : M_{A}^{P} + M_{B}^{O} \Pfr M_{A}^{O} + M_{B}^{P}. \]
Such a function can be written as a matrix
\[ f = \Mat{f_{1,1}}{f_{1,2}}{f_{2,1}}{f_{2,2}} \]
where
\[ f_{1,1} : M_{A}^{P} \Pfr M_{A}^{O} \;\;\;\; \;\; f_{1,2} : M_{B}^{O} \Pfr
M_{A}^{O} \]
\[ f_{2,1} : M_{A}^{P} \Pfr M^{P}_{B} \;\;\;\;\;\; f_{2,2} : M_{B}^{O} \Pfr
M_{B}^{P} . \]
For example, the twist map
\[ M_{A}^{P} + M_{A}^{O} \cong M_{A}^{O} + M_{A}^{P} \]
corresponds to the matrix
\[ \Mat{0}{\ident{M_{A}^{O}}}{\ident{M_{A}^{P}}}{0} \]
where $0$ is the everywhere-undefined partial function.
(Compare the interpretation
of axiom links in \cite{GirardJY:geoi1i}.) The strategy induced by this function is
the copy-cat strategy as defined in \cite{AbramskyS:gamfcm}.
As a set of positions, this strategy is defined by:

\[ \ident{A} = \{ s \in P^{\tt even}_{A \linimpl A} \mid \; s
{\upharpoonright} 1 = s {\upharpoonright} 2 \} . \]
In process terms, this is a bi-directional one place
buffer~\cite{AbramskyS:prop}.
These copy-cat strategies are the identity morphisms in \Games .

\paragraph{Composition}
The composition of (history-free) strategies can similarly
be defined either
in terms of the set representation, or via the underlying functions on moves
inducing the strategies.  We begin with the set representation.
Given $\sigma: A
\rightarrow B, \;
\tau: B \rightarrow C$, we define
\[ \begin{array}{ccl}
\sigma \| \tau & = & \{ s \in (M_A + M_B + M_C )^{\star} \mid
s \Rest A, B \in \Over{\sigma}, \; s \Rest B, C \in \Over{\tau} \} \\
\sigma; \tau & = & \{ s {\upharpoonright} A,C \mid \; s \in \sigma \| \tau
 \}^{\tt even} .
\end{array} \]
This definition bears a close resemblance to that of ``parallel composition
plus hiding'' in the trace semantics of CSP \cite{HoareCAR:comsp};
see \cite{AbramskyS:gamfcm} for an extended
discussion of the analogies between game semantics and concurrency
semantics, and \cite{AbramskyS:prop} for other aspects.

We now describe composition
in terms of the functions inducing strategies.  Say we have $\sigma_f:
A \rightarrow B, \; \sigma_g: B\rightarrow C$.  We want to find $h$ such
that $\sigma_f; \sigma_g = \sigma_h$.  We shall compute $h$ by
the ``execution formula''~\cite{GirardJY:towgi,GirardJY:geoi1i,GirardJY:geoi2d}.  Before giving the
formal definition, let us explain the idea, which is rather simple.
We want to hook the strategies up so that Player's moves in $B$ under
$\sigma$ get turned into Opponent's moves in $B$ for $\tau$,
and vice versa.  Consider the following picture:
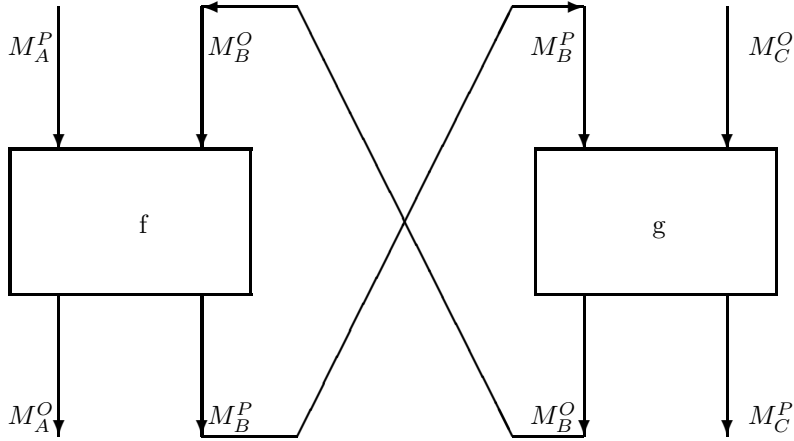
\begin{center}
\setlength{\unitlength}{0.0125in}
\begin{picture}(361,195)(135,550)
\thicklines
\put(350,560){\line( 1, 0){ 30}}
\put(220,560){\line( 1, 0){ 40}}
\put(350,740){\vector( 1, 0){ 30}}
\put(260,740){\vector(-1, 0){ 40}}
\put(260,560){\line( 1, 2){ 90}}
\put(260,740){\line( 1,-2){ 90}}
\put(440,620){\vector( 0,-1){ 60}}
\put(440,740){\vector( 0,-1){ 60}}
\put(380,620){\vector( 0,-1){ 60}}
\put(380,740){\vector( 0,-1){ 60}}
\put(220,620){\vector( 0,-1){ 60}}
\put(160,620){\vector( 0,-1){ 60}}
\put(220,740){\vector( 0,-1){ 60}}
\put(160,740){\vector( 0,-1){ 60}}
\put(360,620){\framebox(100,60){}}
\put(140,620){\framebox(100,60){}}
\put (445,565) {\makebox(0,0) [lb] {\raisebox{0pt}[0pt][0pt]{ $M^P_C$}}}
\put (445,720) {\makebox(0,0) [lb] {\raisebox{0pt}[0pt][0pt]{ $M^O_C$}}}
\put (360,565) {\makebox(0,0) [lb] {\raisebox{0pt}[0pt][0pt]{
$\hspace{-2mm}M^O_B$}}}
\put (360,720) {\makebox(0,0) [lb] {\raisebox{0pt}[0pt][0pt]{
$\hspace{-2mm}M^P_B$}}}
\put (225,565) {\makebox(0,0) [lb] {\raisebox{0pt}[0pt][0pt]{
$\hspace{-2mm}M^P_B$}}}
\put (225,720) {\makebox(0,0) [lb] {\raisebox{0pt}[0pt][0pt]{
$\hspace{-2mm}M^O_B$}}}
\put (135,565) {\makebox(0,0) [lb] {\raisebox{0pt}[0pt][0pt]{ $M^O_A$}}}
\put (135,720) {\makebox(0,0) [lb] {\raisebox{0pt}[0pt][0pt]{ $M^P_A$}}}
\put (405,645) {\makebox(0,0) [lb] {\raisebox{0pt}[0pt][0pt]{ g}}}
\put (190,645) {\makebox(0,0) [lb] {\raisebox{0pt}[0pt][0pt]{ f}}}
\end{picture}

\end{center}

Assume that the Opponent starts in $C$.  There are two
possible cases:
\begin{itemize}
\item  The move is mapped by $g$ to a response in $C$: In this case,
this is the response of the function $h$.
\item The move is mapped by $g$ to a response in $B$.  In this
  case, this response is interpreted as a move of the Opponent in
  $B$ and fed as input to $f$.  In turn, if $f$ responds in
  $A$, this is the response of the function $h$.  Otherwise, if $f$
  responds in $B$, this is fed back to $g$.  In this way, we
  get an internal dialogue between the strategies $f$ and $g$.
\end{itemize}

It remains to give a formula for computing $h$ according to these
ideas.  This is the execution formula:
\[ h = {\tt EX}(f, g) =  \bigvee_{k \in \omega} m_k . \]
The join in the definition of $h$ can be interpreted concretely as union of
graphs.  It is well-defined because it is being applied to a family of
partial functions with pairwise disjoint domains of definition.  The
functions $m_k: M_A^P + M_C^O \rightharpoonup M_A^O + M_C^P $ are defined by
\[ m_k = \pi^{\star} \circ (( f + g) \circ \mu)^{k}  \circ (
f + g) \circ \pi . \]
The idea is that $m_k$ is the function which, when
defined, feeds an input from $M_A^P$ or $M^O_C$ exactly $k$ times around
the channels of the internal feedback loop and then exits from $M^O_A$ or
$M^P_C$.  The retraction
\[ \pi: M_A
+ M_C \lhd M_A +M_B+M_B+M_C : \pi^{\star} \]
is defined by
\[
\pi^{\star} = [{\tt inl},0,0,{\tt inr}] \;\;\;\;\;
\pi = [{\tt in}_1,{\tt in}_4]
\]
and the ``message exchange'' function
$\mu: M_A^O +M_B^P+M_B^O+M_C^P \rightharpoonup M_A^P +M_B^O+M_B^P+M_C^O$
is defined by
\[ \mu = 0 + [{\tt inr},{\tt inl}] + 0. \]
Here, $0$ is the everywhere undefined partial function.

The fact that this definition of composition coincides with that given
previously in terms
of sets of positions is proved in \cite[Proposition~3]{AbramskyS:gamfcm}.

\begin{proposition}
Composition is monotone with respect to $\Subeq$:
  \[ \sigma , \sigma' :A\Fr B, \; \tau , \tau' :B\Fr C, \; \sigma\Subeq\sigma', \
  \tau\Subeq\tau' \THEN\ \sigma;\tau\Subeq\sigma';\tau' . \]
\end{proposition}

\begin{proof}  We follow the analysis of composition given in the
proof of Proposition~1 of \cite{AbramskyS:gamfcm}. Suppose $\sigma
\Subeq \sigma'$, $\tau \Subeq \tau'$, $ca \in \sigma ; \tau$ and
$c \Deq c'$. Then $ca = u \Rest A, C$ for uniquely determined $u =
cb_{1}\cdots b_{k}a$ such that $u \Rest A, B \in \sigma$, $u \Rest
B, C \in \tau$. We must have $c \in M_{C}$. Since $\tau \Subeq
\tau'$, $c'b_{1}' \in \tau'$ for unique $b_{1}'$, and $cb_{1} \Deq
c'b_{1}'$. Now $b_1 \in \B{dom}(\sigma )$ and $\sigma \Subeq
\sigma'$ implies that $b_{1}'b_{2}' \in \sigma'$ for unique
$b_{2}'$, and $b_{1}b_{2} \Deq b_{1}'b_{2}'$. Continuing in this
way, we obtain a uniquely determined sequence $u' = c'b_{1}'\cdots
b_{k}'a'$ such that $u' \Rest A, B \in \sigma'$, $u' \Rest B, C
\in \tau'$, and $ca \Deq c'a'$, as required. This argument is
extended to general strings $s \in \sigma ; \tau$ by an induction
on $|s|$. $\;\;\; $
\end{proof}

We say that a string $s \in (M_{A_{1}} + \ldots + M_{A_{n}})^{\star}$ is
{\em well-formed}
if it satisfies the bracketing condition and the stack discipline; and
{\em balanced} if it is well-formed, and the number of questions in $s$
equals the number of answers.
Note that these properties depend only on the string $\bar{s}$ obtained from $s$
by replacing each question in $A_{1}, \ldots , A_{n}$ by
$(_{1}, \ldots , (_{n}$ respectively, and each answer in $A_{1}, \ldots ,
A_{n}$ by $)_{1}, \cdots , )_{n}$ respectively.

\begin{lemma}
The balanced and well-formed strings in
$(M_{A_{1}} + \cdots + M_{A_{n}})^{\star}$
are generated by the following context-free grammar:
\[ \mbox{\sc bal} \;\;  ::=  \;\;  \epsilon \; | \; \mbox{\sc bal} \; \mbox{\sc bal}
\; | \; (_{i} \,  \mbox{\sc bal}\,  )_{i} \;\; (i = 1, \ldots , n) \]
\[ \mbox{\sc wf} \;\; ::= \;\; \epsilon \; | \; \mbox{\sc bal} \; \mbox{\sc wf} \; | \;
(_{i} \, \mbox{\sc wf} \;\; (i=1, \ldots , n) . \]
(More precisely, $s$ is well-formed (balanced) iff $\bar{s}$ is derivable from
{\sc wf} ({\sc bal}) in the above grammar.)
\end{lemma}

\begin{proof} It is easy to see that the terminal strings derivable
from {\sc bal} are exactly the balanced ones, and that strings
derivable from {\sc wf} are well-formed. Now suppose that $s$ is
well-formed. We show by induction on $|s|$ that $s$ is derivable
from {\sc wf}. If $s$ is non-empty, it must begin with a question,
$s = (_{i}t$. If this question is not answered in $s$, then $t$ is
well-formed, and by induction hypothesis $t$ is derivable from
{\sc wf}, hence $s$ is derivable via the production $\mbox{\sc wf}
\rightarrow (_{i} \mbox{\sc wf}$. If this question is answered, so
$s = (_{i} u )_{i} v$, then $(_{i} u )_{i}$ is balanced, and hence
derivable from {\sc bal}, and $v$ is well-formed, and so by
induction hypothesis derivable from {\sc wf}. Then $s$ is
derivable from {\sc wf} via the production $\mbox{\sc wf}
\rightarrow\mbox{\sc bal} \; \mbox{\sc wf}$. $\;\;\; $
\end{proof}

\begin{lemma}[Projection Lemma]
If $s \in (M_{A_{1}} + \cdots + M_{A_{n}})^{\star}$ is well-formed (balanced),
then
so is $s \Rest A_{i_1}, \ldots , A_{i_k}$ for any subsequence
$A_{i_1}, \ldots , A_{i_k}$ of $A_{1}, \ldots , A_{n}$.
\end{lemma}

\begin{proof} We use the characterization of well-formed and
balanced strings from the previous lemma, and argue by induction
on the size of the derivation of $s$ from {\sc wf} or {\sc bal}.
Suppose $s$ is well-formed. If $s$ is empty, the result is
immediate. If $s$ is derivable via $\mbox{\sc wf} \rightarrow
\mbox{\sc bal} \; \mbox{\sc wf}$, so $s = tu$ where $t$ is
balanced and $u$ is well-formed, then we can apply the induction
hypothesis to $t$ and $u$. Similarly when $s = (_{i} t$ where $t$
is well-formed, we can apply the induction hypothesis to $t$. The
argument when $s$ is balanced is similar. $\;\;\; $
\end{proof}

\begin{lemma}[Parity Lemma]
If $s \in \sigma \| \tau $ is such that $s = tmun$, where
$m$, $n$ are moves in the ``visible'' components $A$ and $C$, then:
\begin{itemize}
\item if $m$, $n$ are in the {\em same} component, then $|u \Rest B|$ is even.
\item if $m$, $n$ are in {\em different} components, then $|u \Rest B|$ is odd.
\end{itemize}
\end{lemma}

\begin{proof} Firstly, we consider the case where all moves in $u$
are in $B$. Suppose for example that $m$ and $n$ are both in $A$.
Then the first move in $u$ is by $\sigma$, while the last move is
by $\tau$, since it must have been $\sigma$ which returned to $A$.
Thus $|u|$ is even. Similarly if $m$ and $n$ are both in $C$. Now
suppose that $m$ is in $A$ while $n$ is in $C$. Then the first and
last moves in $u$ were both by $\sigma$, so $|u|$ is odd; and
similarly if $m$ is in $C$ and $n$ is in $A$.

Now we consider the general case, and argue by induction on $|u|$.
Suppose $m$ and $n$ are both in $A$. Let $u = u_{1}m_{1}u_{2}$,
where all moves in $u_{1}$ are in $B$. Suppose firstly that $m_1$
is in $A$; then $|u_{1}|$ is even, and by induction hypothesis
$|u_{2} \Rest B|$ is even, so $|u \Rest B|$ is even. If $m_1$ is
in $C$, then $|u_{1}|$ is odd, and by induction hypothesis $|u_{2}
\Rest B|$ is odd, so $|u \Rest B|$ is even. The other cases are
handled similarly. $\;\;\; $
\end{proof}

\begin{proposition}
If $\sigma : A \Fr B$ and $\tau : B \Fr C$, then $\sigma ; \tau$ satisfies
the bracketing condition and the stack discipline.
\end{proposition}

\begin{proof} By the Projection Lemma, it suffices to verify that
every $s \in \sigma \| \tau$ is well-formed. We argue by induction
on $|s|$. The basis is trivial. Suppose $s = tm$. If $m$ is a
question, it cannot destroy well-formedness. If $m$ is an answer
with no matching question, then by induction hypothesis $t$ is
balanced. Suppose $m$ is in $A$ or $B$; then by the Projection
Lemma, $t \Rest A, B$ is balanced, so $m$ has no matching question
in $s \Rest A,B = (t \Rest A, B)m$, contradicting $s \Rest A, B
\in \sigma$. A similar argument applies when $m$ is in $B$ or $C$.

So we need only consider $s = umvn$ where $m$, $n$ are a matching question-answer pair.
It remains to show  that $m$ and $n$ must be in the same component.
Suppose firstly that $m$ and $n$ both occur in $A$ or $B$. Note that $v$ is
balanced, and then by the Projection Lemma, so is $v \Rest A,B$.
So $m$ and $n$ will be paired in $s \Rest A, B \in \sigma$, and hence they
must be in the same component. Similarly when $m$ and $n$ are both in $B$ or
$C$.

The final case to be considered is when $m$ and $n$ both occur in
$A$ or $C$. Since $v$ is balanced, by the Projection Lemma so is
$v \Rest B$. It follows that $|v \Rest B|$ is even, so by the
Parity Lemma, $m$ and $n$ must be in the same component. $\;\;\;
$
\end{proof}

Combining Propositions~2.4.2 and~2.4.6 with Proposition~2 from \cite{AbramskyS:gamfcm}, we obtain:

\begin{proposition}
\Games\ is a category.
\end{proposition}

\subsection{$\Games$ as an autonomous category}
We have already defined the object part of the tensor product $A
\tensor B$, linear implication $A \linimpl B$, and the tensor unit $I$.
The action of tensor on morphisms is defined as follows.  If $\sigma_f: A
\rightarrow B, \; \sigma_g: A' \rightarrow B'$, then
$\sigma_f \tensor \sigma_g: A\tensor A' \rightarrow B \tensor B'$  is
induced by the partial function
\[ \begin{array}{l}
(M_A^{P} + M_{A'}^{P}) + (M_B^{O} + M_{B'}^{O}) \\
\cong (M_A^{P} + M_{B}^{O}) + (
M_{A'}^{P}+ M_{B'}^{O})
\stackrel{f+g}{\Pfr} (M_A^{O} + M_{B}^{P})+ ( M_{A'}^{O}+ M_{B'}^{P}) \\
\cong (M_A^{O} + M_{A'}^{O}) + (M_B^{P} + M_{B'}^{P}) .
\end{array} \]
The natural isomorphisms for associativity, commutativity and unit of the
tensor product:
\[ \begin{array}{rcrcl}
{\tt assoc}_{A,B,C} & : & (A \tensor B) \tensor C & \cong & A \tensor (B \tensor C) \\
{\tt symm}_{A,B} & : & A \tensor B & \cong & B \tensor A \\
{\tt unit}_{A} & : & A \tensor I & \cong & A
\end{array} \]
are induced by the evident bijections on the sets of moves:
\[ ((M^{P}_A + M^{P}_B) + M^{P}_C) + (M^{O}_A + (M^{O}_B + M^{O}_C))  \cong
((M^{O}_A
+ M^{O}_B) + M^{O}_C) + (M^{P}_A + (M^{P}_B + M^{P}_C)) \]
\[ (M^{P}_{A} + M_{B}^{P}) + (M^{O}_{B} + M^{O}_{A})  \cong
(M_{A}^{O} + M^{O}_{B}) + (M_{B}^{P} + M^{P}_{A}) \]
\[ (M_{A}^{P} + \varnothing ) + M^{O}_{A}  \cong (M^{O}_{A} + \varnothing ) +
M_{A}^{P} .
\]
The application morphism ${\tt App}_{A,B}: (A \linimpl B) \tensor A
\rightarrow B$ is induced by
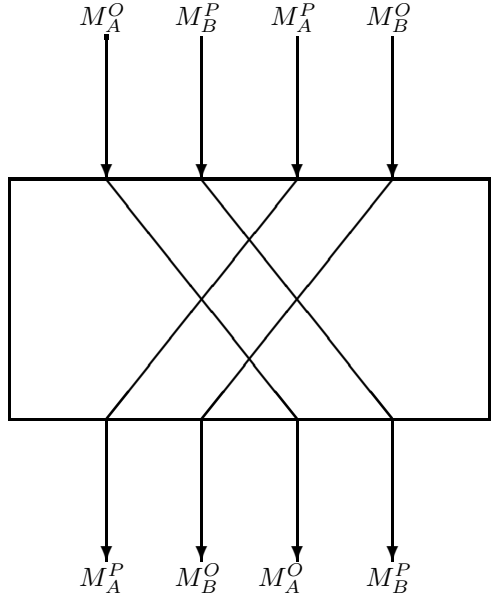
\[ ((M^{O}_A + M^{P}_B) + M^{P}_A) + M^{O}_B \cong ((M^{P}_A + M^{O}_B) + M^{O}_A) + M^{P}_B .
\]

\setlength{\unitlength}{0.0125in}
\begin{picture}(200,263)(40,540)
\thicklines
\put(300,620){\vector( 0,-1){ 60}}
\put(260,620){\vector( 0,-1){ 60}}
\put(220,620){\vector( 0,-1){ 60}}
\put(180,620){\vector( 0,-1){ 60}}
\put(300,780){\vector( 0,-1){ 60}}
\put(260,780){\vector( 0,-1){ 60}}
\put(220,780){\vector( 0,-1){ 60}}
\put(180,780){\vector( 0,-1){ 60}}
\put(180,780){\framebox(0,0){}}
\put(140,620){\framebox(200,100){}}
\put(300,720){\line(-4,-5){ 80}}
\put(220,720){\line( 4,-5){ 80}}
\put(260,720){\line(-4,-5){ 80}}
\put(180,720){\line( 4,-5){ 80}}
\put (205,550) {\makebox(0,0) [lb] {\raisebox{0pt}[0pt][0pt]{ $M^O_B$}}}
\put (285,785) {\makebox(0,0) [lb] {\raisebox{0pt}[0pt][0pt]{ $M^O_B$}}}
\put (165,550) {\makebox(0,0) [lb] {\raisebox{0pt}[0pt][0pt]{ $M^P_A$}}}
\put (245,785) {\makebox(0,0) [lb] {\raisebox{0pt}[0pt][0pt]{ $M^P_A$}}}
\put (285,550) {\makebox(0,0) [lb] {\raisebox{0pt}[0pt][0pt]{ $M^P_B$}}}
\put (205,785) {\makebox(0,0) [lb] {\raisebox{0pt}[0pt][0pt]{ $M^P_B$}}}
\put (240,550) {\makebox(0,0) [lb] {\raisebox{0pt}[0pt][0pt]{ $M^O_A$}}}
\put (165,785) {\makebox(0,0) [lb] {\raisebox{0pt}[0pt][0pt]{ $M^O_A$}}}
\end{picture}

This ``message switching'' function can be understood in algorithmic terms
as follows.  A demand for output from the application at $M^{O}_B$ is
switched to the function part of the input, $A \linimpl B$; a demand by the
function input for information about its input at $M_A^{O}$ is forwarded to
the input port $A$;  a reply with this information about the input at
$M^{P}_A$ is sent back to the function; an answer from the function to the
original demand for output at $M_B^{P}$ is sent back to the output port $B$.
Thus, this strategy does indeed correspond to a protocol for linear
function application---linear in that the ``state'' of the input changes as
we interact with it, and there are no other copies available
allowing us to backtrack.

As for currying, given
$\sigma_f: A\tensor B \rightarrow C$,
$\Lambda(\sigma_f) : A \rightarrow (B \linimpl C)$ is induced by
\[ M^{P}_A + (M^{P}_B + M^{O}_C) \cong (M^{P}_A + M^{P}_B) + M^{O}_C
 \stackrel{f}{\Pfr} (M^{O}_A + M^{O}_B) + M^{P}_C \cong M^{O}_A + (M^{O}_B
+ M^{P}_C) . \]

For discussion of these definitions, and most of the verification that they
work as claimed, we refer to Section~3.5 of \cite{AbramskyS:gamfcm}.

\begin{proposition}
\begin{enumerate}
\item If $\sigma \Deq \sigma'$ and $\tau \Deq \tau'$ then $\sigma \tensor \tau
\Deq \sigma' \tensor \tau'$.

\item $\sigma \tensor \tau$ satisfies the stack discipline.
\end{enumerate}
\end{proposition}
\begin{proposition}
\Games\ is an autonomous category.
\end{proposition}

\subsection{Products}
The game $A \with B$ is defined as follows.
\begin{eqnarray*}
M_{A \with B} &=& M_A + M_B  \\
\lambda_{A \with B} &=& [\lambda_A,\lambda_B] \\
P_{A \with B}&=&  P_A + P_B  \\
\Deq_{A \with B}&=& \Deq_{A} + \Deq_{B} .
\end{eqnarray*}
The projections
\[ A \stackrel{\tt fst}{\longleftarrow} A \with B
\stackrel{\tt snd}{\longrightarrow}
B \]
are induced by the partial injective maps
\[ (M_{A}^{P} + M_{B}^{P}) + M_{A}^{O} \Pfr (M_{A}^{O} + M_{B}^{O}) + M_{A}^{P}
\]
\[ (M_{A}^{P} + M_{B}^{P}) + M_{B}^{O} \Pfr (M_{A}^{O} + M_{B}^{O}) +
M_{B}^{P}
\]
which are undefined on $M_{B}^{P}$ and $M_{A}^{P}$ respectively.
Pairing {\sl cannot} be defined in general on history-free strategies in
${\cal G}$;
however, it can be defined on the co-Kleisli category for the
comonad $\ofcourse $, as we will see.

\subsection{Exponentials}
Our treatment of the exponentials is based on \cite{AbramskyS:gamexp}.
The game $\ofcourse A$ is defined as the ``infinite symmetric tensor power''
of $A$.
The symmetry is built in via the equivalence relation on positions.

\begin{itemize}
\item $M_{\ofcourse A} =\omega\times M_A=\Sum_{i\in\omega} M_A $, the
  disjoint union of countably many copies of the moves of $A$.  So, moves
  of $\ofcourse A$ have the form $(i, m )$, where $i$ is a natural
  number, called the index, and $m$ is a move of $A$.
\item Labelling is by source tupling:
$$ \lambda_{\ofcourse A}(i, a) =  \lambda_A (a) . $$
\item We write $s\Rest i$ to indicate the restriction to moves with
  index $i$.  $P_{\ofcourse A}$ is the set of all
  $s \in \plays{\ofcourse A}$
such that:
\begin{enumerate}
\item {\em Projection condition:} $(\forall i) \; [s\Rest i \in P_A]$.
\item {\em Stack discipline:} Every answer in $s$ is in the same index as the
corresponding question.
\end{enumerate}
\item Let $S(\omega )$ be the set of permutations on $\omega$.
\[ s \Deq_{\ofcourse A} t \; \Longleftrightarrow \; ( \exists \pi \in
S(\omega ) )
[(\forall i \in \omega . \, s \Rest i \Deq_{A} t \Rest \pi (i) ) \;
\wedge \; (\pi \, \circ \, {\tt fst})^{\ast}(s) = {\tt fst}^{\ast}(t) ] .
\]

\end{itemize}

\paragraph{Dereliction}
For each game $A$ and $i \in \omega$, we define a strategy
\[ {\tt der}_{A}^{i} : \ofcourse A \rightarrow A \]
induced by the partial function $h_{i}$:
\[ \begin{array}{lcl}
h_{i}(j, a) & = & \left\{ \begin{array}{ll}
a, & i=j \\
\mbox{undefined}, & i \neq j
\end{array} \right. \\
h_{i}(a) & = & (i, a) .
\end{array} \]
In matrix form
\[ h_{i} =  \Mat{0}{{\tt in}_{i}}{{\tt in}_{i}^{\ast}}{0} . \]

\begin{proposition}
\begin{enumerate}
\item For all $i$, $j$:
\[ {\tt der}_{A}^{i} \Deq {\tt der}_{A}^{j}. \]
\item \, ${\tt der}_{A}^{i}$ satisfies the stack discipline.
\end{enumerate}
\end{proposition}

By virtue of this Proposition, we henceforth write ${\tt der}_{A}$, meaning
${\tt der}_{A}^{i}$ for arbitrary choice of $i$.

\paragraph{Promotion}
A {\em pairing function} is an injective map
\[ p : \omega \times \omega \rightarrowtail \omega . \]
Given $\sigma_{f} : \ofcourse A \rightarrow B$ and a pairing function $p$, we define
$\sigma^{\dagger}_{p} : \ofcourse A \rightarrow \ofcourse B$ as the strategy induced by the
partial function $f^{\dagger}_{p}$ defined by:
\[ \begin{array}{lcl}
f_{p}^{\dagger}(p(i,j),a) & = & \left\{
\begin{array}{ll}
(p(i,j'), a'), & f(j,a) = (j',a') \\
(i,b), & f(j,a) = b
\end{array} \right.
\\
f_{p}^{\dagger}(i,b) & = & \left\{
\begin{array}{ll}
(p(i,j),a), & f(b) = (j,a) \\
(i,b'), & f(b) = b' .
\end{array} \right.
\end{array} \]
In matrix form
\[ f^{\dagger}_{p} = \Mat{t \circ (1 \times f_{1,1}) \circ t^{\ast}}{t \circ
(1 \times f_{1,2})}{(1 \times f_{2,1}) \circ t^{\ast}}{1 \times f_{2,2}}
\]
where
\[ t(i, (j, a)) = (p(i,j), a) . \]

\begin{proposition}
\begin{enumerate}
\item \, If $\sigma, \tau : \ofcourse A \rightarrow B$, $\sigma \Deq \tau$, and $p$, $q$
are pairing functions, then $\sigma^{\dagger}_{p} \Deq \tau^{\dagger}_{q}$.
\item \, $\sigma^{\dagger}_{p}$ satisfies the stack discipline.
\end{enumerate}
\end{proposition}
By virtue of this Proposition, we shall henceforth write $\sigma^{\dagger}$,
dropping explicit reference to the pairing function.

\begin{proposition}
For all $\sigma : \ofcourse A \rightarrow B$, $\tau : \ofcourse B \rightarrow C$:
\[ \begin{array}{clcl}
({\bf m1}) & \sigma^{\dagger} ; \tau^{\dagger} & \Deq & (\sigma^{\dagger} ;
\tau )^{\dagger} \\
({\bf m2}) & {\tt der}_{A}^{\dagger} ; \sigma & \Deq & \sigma \\
({\bf m3}) & \sigma^{\dagger} ; {\tt der}_{B} & \Deq & \sigma .
\end{array} \]
\end{proposition}
As an immediate consequence of this Proposition and standard results
\cite{ManesE:algt}:
\begin{proposition}
$(\ofcourse  , {\tt der}, ( \cdot )^{\dagger})$ is a comonad in ``Kleisli form''.
If we define, for $\sigma : A \rightarrow B$,
$\ofcourse  \sigma = ({\tt der}_{A} ; \sigma )^{\dagger} : \ofcourse A \rightarrow \ofcourse B$, and
$\delta_{A} : \ofcourse A \rightarrow \ofcourse \ofcourse A$ by $\delta_{A} = {\tt id}_{\ofcourse A}^{\dagger}$,
then $(\ofcourse  , {\tt der}, \delta )$ is a comonad in the standard sense.
\end{proposition}

\paragraph{Contraction and Weakening}
For each game $A$, we define ${\tt weak}_{A} : \ofcourse A \rightarrow I$ by
${\tt weak}_{A} = \{ \epsilon \}$.

A {\em tagging function} is an injective map
\[ c : \omega + \omega \rightarrowtail \omega . \]
Given such a map, the contraction strategy
${\tt con}_{A}^{c} : \ofcourse A \rightarrow \ofcourse A \tensor \ofcourse A$
is induced by the function
\[ \Mat{0}{(r \times 1) \circ \, {\tt inl}^{\ast} \vee (s \times 1)
\circ \, {\tt inr}^{\ast}}{{\tt inl} \, \circ (r^{\ast} \times 1)
\vee {\tt inr} \, \circ(s^{\ast} \times 1)} {0} \] where $r =
\omega \stackrel{\tt inl}{\longrightarrow} \omega + \omega
\stackrel{c}{\longrightarrow} \omega$, $s = \omega \stackrel{\tt
inr}{\longrightarrow} \omega + \omega
\stackrel{c}{\longrightarrow} \omega$.

Again, it is easily verified that ${\tt con}_{A}^{c} \Deq {\tt con}_{A}^{c'}$
for any tagging functions $c$, $c'$.

\begin{proposition}
${\tt con}_{A}$, ${\tt weak}_{A}$ are well-defined strategies which give
a cocommutative comonoid structure on $\ofcourse A$, {\em i.e.} the following diagrams
commute:
\[ \begin{diagram}
\ofcourse A & & \rTo^{[{\tt con}_{A}]} & & \ofcourse A \tensor \ofcourse A \\
\dTo^{[{\tt con}_{A}]}    & & & & \dTo_{[{\tt id}_{A} \tensor
{\tt con_{A}}]} \\
\ofcourse A \tensor \ofcourse A & \rTo_{[{\tt con}_{A} \tensor {\tt id}_{A}]} &
(\ofcourse A \tensor \ofcourse A ) \tensor \ofcourse A  & \rTo_{[{\tt assoc}_{A}]} & \ofcourse A \tensor (\ofcourse A \tensor \ofcourse A ) \\
\end{diagram} \]
\[ \begin{diagram}
\ofcourse A & \rTo^{[{\tt con}_{A}]} & \ofcourse A \tensor \ofcourse A \\
\dTo^{[{\tt id}_{A}]} & & \dTo_{[{\tt id}_{A} \tensor {\tt weak}_{A}]} \\
\ofcourse A & \lTo_{[{\tt unit}_{A}]} & \ofcourse A \tensor I \\
\end{diagram} \;\;\;\;\;\;\;\;
\begin{diagram}
\ofcourse A & \rTo^{[{\tt con}_{A}]} & \ofcourse A \tensor \ofcourse A \\
& \rdTo_{[{\tt con}_{A}]} & \dTo_{[{\tt symm}_{A,A}]} \\
& & \ofcourse A \tensor \ofcourse A \\
\end{diagram} \]
\end{proposition}

\subsection{The co-Kleisli category}
By Proposition~2.7.4, we can form the co-Kleisli category $\KG$, with:
\begin{description}
\item[Objects] The objects of $\Games$.
\item[Morphisms] $\KG (A, B) = \Games (\ofcourse A, B)$.
\item[Composition] If $\sigma : \ofcourse A \rightarrow B$ and $\tau : \ofcourse B \rightarrow C$
then composition in $\KG$ is given by:
\[ \sigma \fatsemi \tau = \sigma^{\dagger} ; \tau . \]
\item[Identities] The identity on $A$ in $\KG$ is
${\tt der}_{A} : \ofcourse A \rightarrow A$.
\end{description}

\paragraph{Exponential laws}
\begin{proposition}
\begin{enumerate}
\item There is a natural isomorphism
$e_{A,B} : \ofcourse (A \with B) \cong \ofcourse A \tensor \ofcourse B$.
\item \, $\ofcourse I = I$.
\end{enumerate}
\end{proposition}
\begin{proof}
\begin{enumerate}
\item We define $e_{A,B}:!(A \with B ) \linimpl !A \tensor !B $ as
  (the strategy induced by) the
map which sends ${\tt inl} (a,i)\in \, !A \tensor !B $ to
$({\tt inl}(a),i)\in\, !(A \with B ) $, $({\tt inl}(a),i)\in \, !(A \with B ) $ to
${\tt inl} (a,i)\in \, !A \tensor !B $ and similarly sends ${\tt inr}
(b,i)\in \, !A
\tensor !B$ to $({\tt inr}(b),i)\in \, !(A \with B ) $, $({\tt
  inr}(b),i)\in \,
!(A
\with B ) $ to  ${\tt inr} (b,i)\in \, !A \tensor !B $.

We define ${e_{A,B}^{-1}}:(!A \tensor !B) \linimpl !(A \with B )$ as
(the strategy induced by) the map
which sends ${\tt inl}(a,2i)\in \, !A \tensor !B $ to $({\tt
  inl}(a),i)\in \, !(A\with
B)$,  $({\tt inl}(a),i)\in \, !(A\with B)$ to ${\tt inl}(a,2i)\in \, !A \tensor !B
$ and $({\tt inr}(b),i)\in \, !(A \with B ) $ to  ${\tt inr} (b,2i+1)\in
\, !A \tensor
!B $, ${\tt inr} (b,2i +1)\in \, !A \tensor !B$ to $({\tt inr}(b),i)\in
\, !(A \with B
) $.

It is straightforward to check that $e_{A,B},{e_{A,B}^{-1}}$ are strategies.
Let's prove that $e_{A,B},{e_{A,B}^{-1}}$ define the required isomorphism.

\begin{itemize}
\item For $ e_{A,B}; {e_{A,B}^{-1}} :  (!A_2 \tensor !B_2) \linimpl (!A_1 \tensor
!B_1)$ (we have used different subscripts for different copies of
the same game) we have that ${\tt inl}(a,i)\in (!A_1\tensor !B_1)$ is sent to
${\tt inl}(a,2i)\in (!A_2\tensor B_2)$ and ${\tt inr}(b,j)\in (!A_1\tensor !B_1)$ is
sent to ${\tt inr}(b,2j+1)\in (!A_2\tensor B_2)$ . This strategy is
equivalent to the
identity. The automorphism which witnesses the equivalence is the map
which sends $i$ in $!A_1$ to $2i$ and $j$ in $!B_1$ to $2j+1$ (and is
the identity elsewhere).

\item For ${e_{A,B}^{-1}}; e_{A,B}$ the same map as above witnesses the
equivalence of ${e_{A,B}^{-1}}; e_{A,B}$ with the identity.
\end{itemize}
\item Immediate by definition.

\end{enumerate}
\end{proof}

\paragraph{Products in $\KG$}
\begin{proposition}
$I$ is terminal in $\KG$.
\end{proposition}
\begin{proof} For any game $A$ there is only one strategy in
$!A\linimpl I$, namely $\{\epsilon\}$. This is because $I$ has an
empty set of moves and for any opening move $a$ in $!A$ we have
$\lambda_{!A\linimpl
  I}(a)=P$ so that Opponent has no opening move in $!A\linimpl I$.$\;\;\;$
  \end{proof}

\begin{proposition}
$A \stackrel{\pi_{1}}{\longleftarrow} A \with B
\stackrel{\pi_{2}}{\longrightarrow} B$
is a product diagram in $\KG$, where
\[ \begin{array}{lcl}
\pi_1 & = & \ofcourse (A \with B) \stackrel{{\tt der}}{\longrightarrow}
A \with B \stackrel{{\tt fst}}{\longrightarrow} A \\
\pi_2 & = & \ofcourse (A \with B) \stackrel{{\tt der}}{\longrightarrow}
A \with B \stackrel{{\tt snd}}{\longrightarrow} B .
\end{array} \]

If $\sigma:!C\linimpl A,\tau:!C\linimpl B$ then their pairing
$\lang\sigma,\tau\rang:!C\linimpl A\& B$ is defined by
\[\begin{diagram}
\lang\sigma,\tau\rang=!C&\rTo^{\tt con}&!C\otimes
!C&\rTo^{\sigma^{\dag}\otimes\tau^{\dag}} &!A\otimes !B&\rTo^e&!(A\&
B)&\rTo^{\tt der}& A\&B.
\end{diagram}\]

\end{proposition}

In fact, we have:

\begin{proposition}
$\KG$ has countable products.
\end{proposition}

\paragraph{Cartesian closure}
We define $ A \Rightarrow B  \equiv \ofcourse A \linimpl B$.
\begin{proposition}
$\KG$ is cartesian closed.
\end{proposition}

\begin{proof} We already know that $\KG$ has finite products. Also,
we have the natural isomorphisms
\[ \begin{array}{lcl}
\KG (A \with B, C) & = & \Games (\ofcourse (A \with B), C) \\
& \cong & \Games (\ofcourse A \tensor \ofcourse B , C) \\
& \cong & \Games (\ofcourse A , \ofcourse B \linimpl C) \\
& = & \KG (A, B \Rightarrow C).
\end{array} \]
Thus $\KG$ is cartesian closed, with ``function spaces'' given by $\Rightarrow$.
\end{proof}

We shall write ${\cal I}=\KG$, since we think of this category as our
intensional model.

\subsection{Order-enrichment}
There is a natural ordering on strategies on a game $A$
given by set inclusion.
It is easily seen that (history-free) strategies are closed under directed
unions, and that $\{ \epsilon \}$ is the least element in this ordering.
However, morphisms in $\Games$ are actually partial equivalence classes
of strategies, and we must define an order on these partial equivalence
classes.

We define:

$$ [\sigma]\Sx_A [\tau] \IFF \sigma\Subeq \tau.$$

\begin{proposition}
$\Sx_{A}$ is a partial order over  $\hat{A}$. The least element
in this partial order is $[\{ \epsilon \} ]$.
\end{proposition}

We have not been able to determine whether $(\hat{A}, \sqsubseteq_{A})$ is
a cpo in general. However, a weaker property than cpo-enrichment suffices
to model PCF, namely {\em rationality}, and this property can be
verified for $\KG$.

A {\em pointed poset} is a partially ordered set with a least element.
A cartesian closed category ${\Bbb C}$ is {\em pointed-poset enriched}
(ppo-enriched) if:
\begin{itemize}
\item Every hom-set ${\Bbb C}(A, B)$ has a ppo structure
$({\Bbb C}(A, B), \sqsubseteq_{A,B}, \bot_{A,B})$.
\item Composition, pairing and currying are monotone.
\item Composition is {\em left-strict}: for all $f : A \rightarrow B$,
\[ \bot_{B,C} \circ f = \bot_{A, C} . \]
\end{itemize}
${\Bbb C}$ is {\em cpo-enriched} if it is ppo-enriched, and moreover
each poset
\[ ({\Bbb C}(A, B), \sqsubseteq_{A,B}) \]
is directed-complete,
and composition preserves directed suprema.
${\Bbb C}$ is {\em rational}  if it is ppo-enriched, and moreover
for all $f : A \times B \rightarrow B$:
\begin{itemize}
\item The chain $(f^{(k)} \mid k \in \omega )$ in ${\Bbb C}(A, B)$
defined inductively by
\[ f^{(0)} = \bot_{A, B}, \;\;\;\;\;\;\;\;
f^{(k+1)} = f \circ \langle \ident{A},
f^{(k)} \rangle \]
has a least upper bound, which we denote by $f^{\rat}$.
\item For all $g : C \rightarrow A$, $h : B \rightarrow D$,
\[ g \circ f^{\rat} \circ h  = \bigsqcup_{k \in \omega} g \circ f^{(k)} \circ
h . \]
\end{itemize}
Altough the standard definition of categorical model for PCF is based on
cpo-enriched categories, in fact rational categories suffice to interpret
PCF, as we will see in Section~2.10.

\paragraph{Strong completeness and continuity}
Let $A$ be a game, and $(\Lambda , \leqslant )$ a directed set.
A family $\{ [\sigma_{\lambda}] \mid \lambda \in \Lambda \}$ is said to be
{\em strongly directed} if there exist strategies $\sigma'_{\lambda}$ for each
$\lambda \in \Lambda$ such that $\sigma'_{\lambda} \in [\sigma_{\lambda}]$
and $\lambda \leqslant \mu \; \Rightarrow \; \sigma'_{\lambda} \subseteq
\sigma'_{\mu}$.

\begin{proposition}
A strongly directed family is $\sqsubseteq$-directed.
Every strongly directed family has a $\sqsubseteq$-least upper bound.
\end{proposition}

Now consider the constructions in $\Games$  we have introduced in previous
sections. They have all been given in terms of concrete operations
on strategies,
which have then been shown to be compatible with the partial preorder
relation $\Subeq$, and hence to give rise to well-defined operations
on morphisms of $\Games$.
Say that an $n$-ary concrete operation $\Phi$ on strategies is
{\em strongly continuous} if it is monotone with respect to $\Subeq$,
and monotone
and continuous with respect to subset inclusion and directed unions:
\[ \begin{array}{ll}
\bullet & \sigma_{1} \Subeq \tau_{1}, \ldots , \sigma_{n} \Subeq \tau_{n} \;
\Longrightarrow \; \Phi (\sigma_{1} , \ldots , \sigma_{n}) \Subeq
\Phi (\tau_{1}, \ldots , \tau_{n})  \\
\bullet & \Phi (\bigcup S_1 , \ldots , \bigcup S_n ) = \bigcup
\{ \Phi (\sigma_i , \ldots , \sigma_n ) \mid \sigma_i \in S_i , \;
i \in 1, \ldots , n \}
\end{array} \]
for directed $S_1 , \ldots , S_n$.
(Note that for $n=0$, these properties reduce to $\Phi \Deq \Phi$.)

\begin{proposition}
Composition, tensor product, currying and promotion
are strongly continuous.
\end{proposition}

\begin{proposition}\label{229}
$\KG$ is a rational cartesian closed category.
\end{proposition}

\subsection{The model of PCF}

PCF is an applied simply-typed $\lambda$-calculus; that is, the terms in PCF are
terms of the simply-typed $\lambda$-calculus built from a certain
stock of constants.
As such, they can be interpreted in any cartesian closed category once we have
fixed the interpretation of the ground types and the constants.
The constants of PCF fall into two groups: the ground and first-order constants
concerned with arithmetic manipulation and conditional branching;
and the recursion combinators
${\bf Y}_{T} : (T \Rightarrow T) \Rightarrow T$  for each type $T$.
These recursion combinators can be canonically interpreted in any rational
cartesian closed category ${\Bbb C}$.
Indeed, given any object $A$ in ${\Bbb C}$, we can define
$\Theta_{A} : {\bf 1} \times (A \Rightarrow A) \Rightarrow A \longrightarrow
(A \Rightarrow A) \Rightarrow A$
by
\[ \Theta_{A} = \llbracket F : (A \Rightarrow A) \Rightarrow A \vdash
\lambda f^{A \Rightarrow A}. f(F f) : (A \Rightarrow A) \Rightarrow A
\rrbracket . \]
Now define ${\bf Y}_{A} = \Theta_{A}^{\rat} : {\bf 1} \longrightarrow (A \Rightarrow
A) \Rightarrow A$.
Note that
\[ {\bf Y}_{A} = \bigsqcup_{k \in \omega} \Theta_{A}^{(k)} =
\bigsqcup_{k \in \omega} \llbracket {\bf Y}_{A}^{(k)} \rrbracket , \]
where
\[ {\bf Y}_{A}^{(0)} = \lambda f^{A \Rightarrow A}. \bot_{A} \;\;\;\;\;\;\;\;
{\bf Y}_{A}^{(k + 1)} = \lambda f^{A \Rightarrow A}. f({\bf Y}_{A}^{(k)} f) . \]
These terms ${\bf Y}_{A}^{(k)}$ are the standard ``syntactic approximants''
to ${\bf Y}_{A}$.

Thus, given a rational cartesian closed category ${\Bbb C}$, a model
$\Mod ({\Bbb C})$
of PCF can be defined by stipulating the following additional information:
\begin{itemize}
\item For each ground type of PCF, a corresponding object
of ${\Bbb C}$. This suffices to determine the interpretation of each
PCF type $T$ as an object in ${\Bbb C}$, using the cartesian closed structure
of ${\Bbb C}$.
(For simplicity, we shall work with the version of PCF with a single ground
type $N$.)
\item For each ground constant  and first-order function of PCF,
say of type $T$,
a morphism $x : {\bf 1} \rightarrow A$ in ${\Bbb C}$, where ${\bf 1}$
is the terminal object in ${\Bbb C}$, and $A$ is the object in ${\Bbb C}$
interpreting the type $T$. ($x$ is a ``point'' or ``global element'' of the
type $A$.)
\end{itemize}

We say that $\Mod ({\Bbb C})$ is a {\em standard model} if ${\Bbb
C}({\bf 1}, N ) \cong {\Bbb N}_{\bot}$, the flat cpo of the
natural numbers, and moreover the interpretation of the ground and
first-order arithmetic constants agrees with the standard one. We
cite an important result due to Berry
\cite{BerryG:modcom,BerryG:fulasl}.

\begin{theorem}[Computational Adequacy]
If $\Mod ({\Bbb C})$ is a standard model, then it is computationally
adequate; {\em i.e.} for all programs $M$ and ground constants $\underline{c}$,
\[ M \longrightarrow^{\ast} \underline{c} \;\; \Longleftrightarrow \;\;
\llbracket M \rrbracket = \llbracket \underline{c} \rrbracket \]
and hence the model is {\em sound}: for all terms $M, N : T$,
\[ \llbracket M \rrbracket \sqsubseteq \llbracket N \rrbracket \;\;
\Longrightarrow \;\; M \sqsubseteq^{\mbox{obs}} N . \]
\end{theorem}
(Berry stated his result for models based on cpo-enriched categories,
but only used rational closure.)

Thus to obtain a model $\Mod (\KG )$ it remains only to specify the
ground types and first-order constants.
The interpretation of $N$ as $\Nat$ has already been given at the
end of Section~2.1. It is readily seen that $\hat{\Nat} \cong {\Bbb N}_{\bot}$.

\paragraph{Ground constants}
For each natural number $n$, there is a strategy $\overline{n} : I \rightarrow
\Nat$, given by
\[ \overline{n} = \{ \epsilon , \ast \underline{n} \} . \]
Also, $\Omega_{\Nat} = [\{\epsilon \} ]$.

\paragraph{Arithmetic functions}
For each number-theoretic partial function $f : {\Bbb N} \Pfr {\Bbb N}$
there is a strategy
\[ \sigma^{f} = \{ \epsilon , \ast_{2}\ast_{1} \} \cup
\{ \ast_{2}\ast_{1}\underline{n}_{1}\underline{m}_{2} \mid
f(n) \Eqdef m \} . \]

\paragraph{Conditionals}
The strategy $\kappa$ interpreting ${\tt if0} : N \Rightarrow N \Rightarrow
N \Rightarrow N$ is defined as follows:
in response to the initial question, $\kappa$ interrogates its first
argument;
if the answer is 0, then it interrogates the second argument, and copies the
reply to the output; if the answer is any number greater than 0, it interrogates
the third argument, and copies the reply to the output.

\begin{proposition}
$\Mod (\KG )$ is a standard model of PCF.
\end{proposition}

\section{Intensional Full Abstraction}
\newcommand{\Intmod}{\cal M(I)}

\subsection{\pcfc}

In order to obtain our intensional full abstraction result, it turns out
that we need to consider an extension of PCF. This extension is quite
``tame'', and does not change the character of the language. It consists
of extending PCF with a family of first order constants
\[ {\tt case}_{k} : N \Rightarrow \underbrace{N \Rightarrow \cdots
  \Rightarrow N}_{k} \Rightarrow N \] for each $k \in \omega$.  The
functions that these constants are intended to denote are defined by:
\[ \begin{array}{lllll}
{\tt case}_{k} & \perp & n_0 \; n_1 \ldots\ n_{k-1} &=& \perp \\
{\tt case}_{k} & i     & n_0 \; n_1 \ldots\ n_{k-1} &=& n_i, \; 0 \leq
i < k \\
{\tt case}_{k} & i     & n_0 \; n_1 \ldots\ n_{k-1} &=& \perp, \; i
\geq k .
\end{array}
\]
The interpretation of ${\tt case}_{k}$ as a strategy is immediate:
this strategy responds to the initial question by interrogating its
first input; if the response is $i$, with $ 0 \leq i < k$, it
interrogates the $i+1$'th input and copies the answer to the output;
otherwise, it has no response.

To see how harmless this extension, which we call \pcfc, is, note that
each term in \pcfc\  is observationally equivalent to one in PCF.
Specifically,

\begin{tabbing}
case =  lx. \=ly.\= ...\= y. \kill
${\tt case}_{k} \equiv^{\tt obs} \lambda x^N.  \lambda y_0^N. \ldots
\lambda y_{k-1}^N.$ \\
\> ${\tt if0} \; x \; y_0 $ \\
\> \> (${\tt if0} \; ({\tt pred} \; x) \; y_1 $\\
\>\>\> \vdots \\
\>\>\> (${\tt if0} \; (\underbrace{{\tt pred \; pred \; pred}}_k \; x) \;
y_{k-1} \; \Omega ) \ldots ) .$\\
\end{tabbing}
The point is that our intensional model is sufficiently fine-grained
to distinguish between these observationally equivalent terms.
However, note that our results in Section~4 apply directly to PCF.

\subsection{Evaluation Trees}

We shall now describe a suitable analogue of B\"{o}hm
trees~\cite{BarendregtHP:lamcss}
for \pcfc.  These give an (infinitary) notion of normal forms for
\pcfc\ terms, and provide a bridge between syntax and semantics.

We use $\Gamma, \Delta $ to range over type environments $x_1: T_1,
\ldots, x_k: T_k$.  We define \finevalT{\Gamma}{T}, the finite
evaluation trees of type $T$ in context $\Gamma$, inductively as
follows:
\begin{itemize}

\item \twotrans{M \in \finevalT{\Gamma,x:T}{U}}{\lambda x^T. M \in
    \finevalT{\Gamma}{T \Rightarrow U}}

\item \twotrans{}{\Omega, {\tt n} \in \finevalT{\Gamma}{N}}

\item \twotrans{
\begin{array}{l}
\Gamma(x) = T_1 \Rightarrow \ldots T_k \Rightarrow N, \\
 P_i \in \finevalT{\Gamma}{T_i}, 1 \leq i \leq k, \\ Q_n \in
    \finevalT{\Gamma}{N}, n \in \omega, \\ \exists n \in \omega. \;
    \forall m \geq n. \; Q_n = \Omega
\end{array}
}{{\tt case}(xP_1 \ldots P_k, (Q_n
    \mid n \in \omega)) \in \finevalT{\Gamma}{N}}
\end{itemize}

We regard these evaluation trees as defined ``up to
$\alpha$--equivalence'' in the usual sense. Note that if we
identify each $${\tt case}(xP_1 \ldots P_k, (Q_n \mid n \in
\omega))$$ with $${\tt case}_{l}(xP_1 \ldots P_k, Q_0, \ldots,
Q_{l-1})$$ for the least $l$ such that $Q_n = \Omega$ for all $n
\geq l$, then every finite evaluation tree is a term in \pcfc.

We order \finevalT{\Gamma}{T} by the ``$\Omega$--match ordering'':
$M\lpo N$ if $N$ can be obtained from $M$ by replacing
occurrences of $\Omega$ by arbitrary finite evaluation trees.

\begin{proposition}\label{A}
$(\finevalT{\Gamma}{T}, \lpo)$ is a pointed poset with non-empty
meets.  Every principal ideal is a finite distributive lattice.
\end{proposition}

Now we define \evalT{\Gamma}{T}, the space of evaluation trees, to be
the ideal completion of \finevalT{\Gamma}{T}.  As an immediate
consequence of proposition~\ref{A}, we have

\begin{proposition}
\evalT{\Gamma}{T} is a dI-domain.  The compact elements are terms of
\pcfc.
\end{proposition}
Strictly speaking, the compact elements of \evalT{\Gamma}{T} are
principal ideals ${\downarrow}(M)$, where $M$ is a finite evaluation
tree, which can be identified with a term in \pcfc\ as explained
above.

\subsection{The Bang Lemma}

We now prove a key technical result. This will require an additional
hypothesis on games. Say that a game $A$ is {\em well-opened} if the
opening moves of $A$ can only appear in opening positions. That is,
for all $a\in M_A$ if $a\in P_A$ then
$$ sa\in P_A \Rightarrow s=\epsilon .$$
It is easy to see that $N$ and $I$ are well-opened, that if $A$
and $B$ are well-opened so is $A\with B$ and that if $B$ is
well-opened so is $A\Rightarrow B$.
Here and henceforth we blur the distinction between the type $N$ and
the game it denotes.
Thus the category of well-opened
games is cartesian closed, and generates the same PCF model
$\Intmod$.

Now let $A$ be well-opened and consider $s\in \even{P_{! A\linimpl
    !B}}$.
Using the switching condition, we see that $s$ can be written uniquely
as
$$s=\beta_1\cdots\beta_k$$
where each ``block'' $\beta_j$ has the form $(i_j,b_j)t_j$,
i.e. starts with a move in $!B$; every move in $!B$ occurring in
$\beta_j$ has the form $(i_j,b')$ for some $b'$, i.e. has the same
index as the opening move in $\beta_j$; if $\beta_i,\beta_j$ are two adjacent
blocks then $i\neq j$; and $|\beta_j|$ is even (so
each block starts with an O-move). We refer to  $i_j$ as the {\em
  block index} for $\beta_j$. For each such block index $i$ we define
$s_i$ to be the subsequence of $s$ obtained by deleting all blocks
with index $i'\neq i$.

Some further notation. For $s\in M^*_{!A\linimpl !B}$, we define
$${\tt FST}(s)=\{ i \mid \exists a. (i,a)\mbox{ occurs in } s \}$$
i.e. the set of all indices of moves in $!A$ occurring in $s$. Also,
we write $s \Rest A{,j}$ for the projection of $s$ to moves of the
form $(j,a)$, i.e. moves in $!A$ with index $j$; and similarly
$s\Rest B{,j}$.

\begin{lemma}\label{disblock}
 For all $\sigma:!A\linimpl !B$ with $A$ well-opened, $s\in\sigma$,
 and block indices $i,j$ occurring in $s$:
\ITEM{
\item[(i)] $s_i\in\sigma$,
\item[(ii)] $i\neq j$ implies ${\tt FST}(s_i)\cap {\tt FST}(s_j)=\varempty$.
}
\end{lemma}
\begin{proof} By induction on $|s|$. The basis is trivial. For the
inductive step, write $s=\beta_1\dots\beta_k \beta_{k+1}$,
$t=\beta_1\dots\beta_k$, $umm'=\beta_{k+1}$. Let the index of
$\beta_{k+1}$ be $i$. We show firstly that $(tu)_im\in
P_{!A\linimpl
  !B}$. By the induction hypothesis, for all $j\in {\tt FST}((tu)_i),\; \;
(tu)_i\Rest A{,j}=tu\Rest A{,j}$, while obviously $(tu)_i\Rest
B{,i}=tu\Rest B{,i}$.
Also, $m$ is either a move in $!B$ with index $i$, or a move in
$!A$. In the latter case, by the switching condition the index of $m$
is in ${\tt FST}((tu)_i)$.
Hence the projection conditions are satisfied
by $(tu)_im$. Moreover $(tu)_im$ is well-formed by the Projection Lemma
2.4.4.
Thus $(tu)_im\in P_{!A \linimpl !B}$ as required.

By induction hypothesis, $(tu)_i\in\sigma$, and since
$\sigma=\sigma_f$ is a well-defined history-free strategy, with
$f(m)=m'$ since $tumm'\in\sigma$ we conclude that
$(tumm')_i=(tu)_imm'\in\sigma$. Moreover, for $j\neq i$,
$(tumm')_j=(tu)_j\in\sigma$ by induction hypothesis. This establishes
{\it (i)}.

Now note that, if $tu$ satisfies {\it (ii)}, so does $tum$ by the
switching condition. Suppose for a contradiction that $tumm'$ does
not satisfy {\it (ii)}. This means that $m'=(j,a)$, where $j\in
{\tt FST}((tu)_{i'})$ for some $i'\neq i$ and hence that $s\Rest
A{,j} = s'a$ where $s'\neq\epsilon$, so that $a$ is a non-opening
move in $A$. But we have just shown that $(tu)_i mm'\in\sigma
\subseteq P_{!A\linimpl !B}$ and hence that $(tu)_i mm' \Rest
A{,j}\in P_A$. By induction hypothesis $$ {\tt FST}((tu)_i)\cap
{\tt FST}((tu)_{i'}) =\varempty$$ and hence $(tu)_imm' \Rest
A{,j}=a$. Thus $a$ is both an opening and a non-opening move of
$A$, contradicting our hypothesis that $A$ is well opened.
\end{proof}

With the same notation as in lemma \ref{disblock}:

\begin{corollary}\label{cordisblock}
\begin{itemize}
\item[(i)] $\forall j\in {\tt FST}(s_i) \; s_i\Rest A{,j} =s\Rest A{,j}$.
\item[(ii)] $\forall j\not\in {\tt FST}(s_i) \; s_i\Rest A{,j}=\epsilon$.
\item[(iii)] $s_i\Rest B{,i}=s\Rest B{,i}$.
\item[(iv)] $j\neq i$ implies $s_i\Rest B{,j}=\epsilon$.
\end{itemize}
\end{corollary}

\begin{lemma}\label{lemma2}
Let $\sigma,\tau: !A\linimpl !B$ with $A$ well-opened. If
$\sigma;\tt{der}_B \approx \tau;\tt{der}_B$ then $\sigma\approx\tau$.
\end{lemma}
\begin{proof} We prove the contrapositive. Suppose
$\sigma\not\approx\tau$. Then w.l.o.g. we can assume that there
exist positions $sab, s'a'$ such that $sab\in \sigma$,
$s'\in\tau$, $sa\approx s'a'$, and either
$s'a'\not\in\rm{dom}(\tau)$ or $ s'a'b'\in\tau$ and
$sab\not\approx s'a'b'$. Let the block index of $a$ in $sa$ be
$i$, and of $a'$ in $s'a'$ be $i'$. Note that the block index of
$b$ in $sab$ must also be $i$.

By Lemma \ref{disblock}, $(sab)_i\in\sigma$ and $s'_{i'}\in\tau$. We
claim that $(sa)_i\approx(s'a')_{i'}$. Indeed, if $s=\beta_1\dots
\beta_k$, $s'=\beta'_1\dots\beta'_{k'}$, then by definition of
$\approx_{!A\linimpl !B}$ we must have $k=k'$ and the permutation
$\pi=[\pi_A,\pi_B]$ witnessing $sa\approx s'a'$ must map the block
index of each $\beta_j$ to that of $\beta'_j$, so that in particular
$sa\Rest B{,i}\approx s'a'\Rest B{,i'}$. Moreover, $\pi_A$ must map
${\tt FST}((sa)_i)$ bijectively onto ${\tt FST}((s'a')_{i'})$. Using
Corollary \ref{cordisblock} for each $j\in {\tt FST}((sa)_i)$, $(sa)_i \Rest
  A{,j}= sa\Rest A{,j}\approx s'a'\Rest A {,\pi_A(j)}
  =(s'a')_{i'}\Rest A {,\pi_A(j)}$.

Now let $tcd$ be defined by replacing each $(i,m)\in !B$ in $s_i
ab$ by $m$; and $t'c$ be defined by replacing each $(i',m')\in !B$
in $s'_ia'$ by $m'$. Then $tcd\in\sigma;{\rm{der}_B}^i,\;
t'\in\tau;{\rm{der}_B}^{i'}$ and $tc\approx t'c'$. We wish to
conclude that $tcd,t'c'$ witness the non equivalence
$\sigma;\tt{der}_B\not\approx \tau;\tt{der}_B$. Suppose for a
contradiction that for some $d'$,
$t'c'd'\in\tau;{\tt{der}_B}^{i'}$ and $tcd\approx t'c'd'$. This
would imply that for some $b'$, $s'_{i' }a'b'\in\tau$ and
$s_iab\approx s'_{i'}a'b'$. Since $s'a'\in P_{!A\linimpl !B}$ and
$\tau$ is a well-defined history-free strategy, this implies that
$s'a'b'\in\tau$. Using Lemma \ref{disblock} and Corollary
\ref{cordisblock} as above, $sab\approx s'a'b'$. This yields the
required contradiction with our assumptions on $sab,s'a'.\;\;$
\end{proof}

\begin{proposition}
[The Bang Lemma]

For all $\sigma:{!A\linimpl !B}$ with $A$ well opened,
$$\sigma\approx (\sigma;\tt{der}_B)^{\dag}. $$
\end{proposition}

\begin{proof} By the right identity law (Prop. 2.11 \textbf{(m3)}),
 $\sigma;\tt{der}_B\approx (\sigma;\tt{der}_B)^{\dag};\tt{der}_B$.
By Lemma \ref{lemma2}, this implies that
$\sigma\approx(\sigma;\tt{der}_B)^{\dag}. \;\;$
\end{proof}

\subsection{The Decomposition Lemma}

In this section we prove the key lemma for our definability
result. We begin with some notational conventions. We will work
mostly in the cartesian closed category  $\M(\KG)$. We write
arrows in this category as $\sigma : A\Rightarrow B$ and
composition e.g. of $\sigma: A\Rightarrow B$ and $\tau :
B\Rightarrow C$ as $\tau\circ\sigma$. We will continue to write
composition in the Linear Category $\cal G$ in diagram order
denoted by $;$ . We write $${\tt Ap} : (A\Rightarrow B)\with
A\Rightarrow B$$ for the application in the cartesian closed
category, and ``linear'' application in $\cal G$ as $${\tt LAPP }
: (A\linimpl B)\tensor A\rightarrow B$$ All games considered in
this section are assumed to be well-opened. If $s\in
M^*_{D\Rightarrow B}$, we write $${\tt FST}(s)=\{ i\;\mid \;
\exists d\;.(i,d)\mbox{ occurs in } s\}$$ i.e. the set of all
indices of moves in $!D$ occurring in $s$.

Now we define a strategy
$$\chi : N\with N^{\omega} \Rightarrow N$$
corresponding to the {\bf case} construct. It will actually be most
convenient to firstly define the affine version
$$\chi_a : N_1\tensor N_2^{\omega}\linimpl N_0$$
where we have tagged the occurrences of $N$ for ease of
identification;
$$\chi_a={\tt Pref}\{*_0*_1 n_1 *_{2,n} m_{2,n} m_0 \;\mid \;
n,m\in\omega\}$$
i.e. $\chi_a$ responds to the initial question by interrogating its
first input; if it gets the response $n$ it interrogates the $n$'th
component of its second input, and copies the response as its answer
to the initial question.

Now we define

$$\chi =\; !(N\with N^{\omega})
\stackrel{e_{N,N^{\omega}}}{\longrightarrow}
!N\tensor
!N^{\omega}\stackrel{\rm{der}_N\tensor\rm{der}_{N^{\omega}}}{\longrightarrow}
N\tensor N^{\omega}\stackrel{\chi_a}{\longrightarrow}N$$
We will now fix some notation for use in the next few lemmas. Let
$$\sigma : C\with (A\Rightarrow N_2) \Rightarrow N_1$$
be a strategy where we have tagged the two occurrences of $N$ for
ease of identification. We assume that $\sigma$'s response to the
initial question $*_1$ in $N_1$ is to interrogate its second input,
i.e. to ask the initial question $*_2$ in $N_2$. Thus any non-empty
position in $\sigma$ must have the form $*_1 *_2 s$. Moreover by the
stack discipline any {\em complete} position in $\sigma$, i.e. one
containing an answer to the initial question $*_1$, must have the form
$$ *_1 *_2 \; s\;  n_2 \; t \; n_1$$
where $n_2$ is the answer corresponding to the question $*_2$ (this is
the sole---albeit crucial---point at which the stack condition is used
in the definability proof). Thus a general description of non-empty
positions in $\sigma$ is that they have the form
$$ *_1 *_2 s \; n_2 \;t$$
where $n_2$ is the answer corresponding to $*_2$, or
$$*_1 *_2 s$$
where $*_2$ is not answered in $s$.

\begin{lemma}\label{lem1}
For all $ *_1 *_2 s \;n_2 \; t\in\sigma$
\begin{itemize}
\item[(i)] $*_1 *_2 n_2 t\in\sigma$
\item[(ii)] ${\tt FST}(s)\cap {\tt FST}(t)=\varempty$.
\end{itemize}
\end{lemma}
\begin{proof} By induction on $|t|$, which must be odd. (The proof follows very similar
lines to that of Lemma~3.1 in the previous section). The basis is
when $t = m$, and $f(n_2 ) = m$, where $\sigma = \sigma_f$. Then
\textit{(i)} follows because $\sigma$ is a well-defined history-free
strategy,
and \textit{(ii)} holds because otherwise $m = (j, d)$ where $d$ is
both
a starting move, using $*_1 *_2 n_2 m \in \sigma$, and a non-starting
move, using $*_1 *_2 s n_2  t\in\sigma$, contradicting
well-openedness.
If $t=umm'$, then we firstly show that $$*_1 *_2  n_2 um
\in P_{C\with (A\Rightarrow N)\Rightarrow N}$$ By the induction
hypothesis and the switching conditions, for all $j\in {\tt
FST}(um)$ $$ *_1 *_2 n_2 um\Rest {C\with(A\Rightarrow N)}{,j} =
*_1 *_2 s n_2 um\Rest {C\with(A\Rightarrow N)}{,j}$$ so
$*_1*_2n_2um$ satisfies the projection conditions because $*_1*_2
sn_2 um$ does. Also, $*_2 s n_2$ is balanced so by the Parity
Lemma 2.4.3 $*_1\;t$ is well formed, and hence $*_1 *_2 n_2 um$ is
well formed. Thus $$*_1 *_2  n_2 um \in P_{C\with (A\Rightarrow
N)\Rightarrow N}$$ Now since $\sigma=\sigma_f$ is a well-defined
history-free strategy with $f(m)=m'$, and $*_1*_2 n_2 u\in\sigma$
by induction hypothesis, we must have $*_1*_2 n_2 umm'\in\sigma$,
establishing {\it (i)}.

For {\it (ii)} suppose for a contradiction that $m'=(j,d)$ for $j\in
{\tt FST}(s)$. Then $*_1*_2 s n_2 t\Rest {C\with (A\Rightarrow
  N)}{,j}=s'd\in P_{C\with(A\Rightarrow N)}$, where
$s'\neq\epsilon$. On the other hand, by induction hypothesis
$*_1*_2 n_2 umm'\Rest {C\with(A\Rightarrow N)}{,j}=d$, and by {\it
(i)}, $d\in P_{C\with(A\Rightarrow N)}$. This contradicts our
assumption that games are well-opened.
\end{proof}

Now we define
$$\sigma'=\{ *_1*_2 s\;  n_2 n_1\;\mid \; *_1 *_2 s \; n_2\in \sigma \}
\cup \{ *_1 *_2 s \;\mid \;*_1*_2s\in\sigma,\; *_2\mbox{ not answered in
  }\sigma\}$$
and for all $n\in\omega$
$$\tau_n=\{ *_1 t \;\mid \; *_1*_2 n_2 \; t \in \sigma\}$$
\begin{lemma}\label{lem2}
$\sigma':C\with(A\Rightarrow N)\Rightarrow N$ and
$\tau_n:C\with(A\Rightarrow N)\Rightarrow N$ ($n\in\omega$) are valid
strategies.
\end{lemma}
\begin{proof} The fact that each $\tau_n$ is a set of valid
positions follows from Lemma \ref{lem1}. That $\sigma',\tau_n$ are
history-free and satisfy the partial equivalence relation follows
directly from their definitions and the fact that $\sigma$ is a
valid strategy.
\end{proof}

\begin{lemma}\label{lem3}
$\sigma\approx{\tt con}_C;\sigma'\tensor \lang\tau_n\ \mid \  n\in\omega\rang ;\chi_a$.
\end{lemma}
\begin{proof} Unpacking the definition of the RHS $\tau={\tt
con}_C;\sigma'\tensor \lang\tau_n\ \mid \  n\in\omega\rang
;\chi_a$ we see that the second and third moves of $\chi_a$
synchronize and cancel out with the first and last moves of
$\sigma'$ respectively, and the fourth and fifth moves of $\chi_a$
cancel out with the first and last moves of the appropriate
$\tau_n$. Thus positions in $\tau$ have the form $$ *_1 *_2 s' n_2
t' \mbox{ or } *_1 *_2 s'$$ where $*_1 *_2 s n_2 t,\  *_1 *_2 s$
are positions in $\sigma$, and $s',t'$ are bijectively reindexed
versions of $s$ and $t$, with the property that ${\tt FST}(s')\cap
{\tt FST}(t')=\varempty$. However, by Lemma \ref{lem1} we know
that ${\tt FST}(s)\cap {\tt FST}(t)=\varempty$, and hence $$ *_1
*_2 s' n_2 t' \approx  *_1 *_2 s n_2 t $$ and $\sigma\approx\tau$
as required.
\end{proof}

\begin{lemma}\label{lem4}
$\sigma\approx\chi\circ \lang\sigma',\lang\tau_n\ \mid \  n\in\omega\rang \rang $
\end{lemma}
\begin{proof}
\[ \begin{array}[]{llc}
\chi\circ\lang\sigma',\lang\tau_n\ \mid \  n\in\omega\rang \rang  & ={\mbox{definition}}& \\
({\tt con};(\sigma'^{\dag}\tensor \lang\tau_n\ \mid \
n\in\omega\rang^{\dag};e;{\tt der})^{\dag};e^{-1};{\tt der}\tensor{\tt der};\chi_a
&\approx {\mbox{ Bang Lemma}}& \\
{\tt con};(\sigma'^{\dag}\tensor \lang\tau_n\ \mid \
n\in\omega\rang^{\dag});e;e^{-1};{\tt der}\tensor{\tt der};\chi_a &\approx&
\\
{\tt con};(\sigma'^{\dag}\tensor \lang\tau_n\ \mid \
n\in\omega\rang^{\dag});{\tt der}\tensor{\tt der} ;\chi_a &\approx
&\\ {\tt con};(\sigma'^{\dag};{\tt der}\tensor \lang\tau_n\ \mid \
n\in\omega\rang^{\dag};{\tt der});\chi_a &\approx &\\ {\tt
con};(\sigma' \tensor \lang\tau_n\ \mid \  n\in\omega\rang
);\chi_a &\approx {\mbox{Lemma \ref{lem3}}} &\\ \sigma. &&
\end{array}
\]
\end{proof}

We continue with our decomposition, and define
$$\sigma''=\{ s \ \mid \  *_1*_2s\in\sigma,\  *_2\mbox{ not answered in
  } s\}$$
\begin{lemma}\label{lem5}
$\sigma'':C\with (A\Rightarrow N) \Rightarrow !A$ is a well-defined
strategy, and
$$\sigma'\approx {\tt con}_{C\with(A\Rightarrow N)}; \pi_2\tensor
\sigma''; {\tt LAPP } .\ \;\; (\dag)$$
\end{lemma}
\begin{proof}We must firstly explain how moves in $\sigma''$ can be
interpreted as being of type $C\with (A\Rightarrow N)\Rightarrow
!A$. Let the index in $!(C\with (A\Rightarrow N))$ of the response
by $\sigma$ to the initial question $*_1$ be $i_0$. Then we regard
all moves in $s\in\sigma''$ with index $i_0$ as moves in the
target $!A$ , and all moves with index $i\neq i_0$ as moves in the
source $!(C\with (A\Rightarrow N))$. The projection conditions and
stack discipline are easily seen to hold for $s$ with respect to
this type. The fact that $\sigma''$ is history-free and satisfies
the partial equivalence relation follows directly from its
definition and the fact that $\sigma'$ is a valid strategy.

Now write $\tau$ for the RHS of $(\dag)$. We diagram $\tau$, tagging
occurrences of the types for ease of reference.

\[ \begin{diagram}
 & !(C_0\with(!A_0\linimpl N_0)) & \\
& \dTo^{{\tt con}} & \\
& !(C_1\with(!A_1\linimpl N_1))\tensor  !(C_2\with(!A_2\linimpl N_2))
&\\
&\dTo^{\pi_2\tensor\sigma''} & \\
&(!A_3 \linimpl N_3)\tensor \;!A_4 & \\
& \dTo^{{\tt LAPP }} &\\
& N_5 &
\end{diagram} \]

From the definitions ${\tt LAPP }$ plays copy-cat strategies
between $N_3$ and $N_5$ and $!A_3$ and $!A_4$; $\pi_2$ plays a
copy-cat strategy between $!A_3\linimpl N_3$ and a single index
$i_0$ in $!(C_1\with (!A_1 \linimpl N_1))$;  {\tt con} splits
$!(C_0\with(!A_0\linimpl N_0))$ into two disjoint address spaces
$!(C_0\with(!A_0\linimpl N_0))_L$ and $!(C_0\with(!A_0\linimpl
N_0))_R$ and  plays copy-cat strategies between $!(C_0\with
(!A_0\linimpl N_0))_L$ and $!(C_2\with(!A_2\linimpl N_2))$ and
between $!(C_0\with(!A_0\linimpl N_0))_R$ and
$!(C_2\with(!A_2\linimpl N_2))$. Thus we see that the opening move
in $N_5$ is copied to $(i_0,N_0)_L$ via $N_3$ and $(i_0,N_1)$, and
any response in $(i_0,N_0)_L$ is copied back to $N_5$. Similarly,
O's moves $(i_0,!A_0)_L$ are copied to $!A_4$ via $(i_0,!A_1)$ and
$!A_3$; and P's responses in $!A_4$ following $\sigma''$ are
copied back to $(i_0,!A_0)_L$. Finally, O's moves in $!(C_0\with
(!A_0\linimpl N_0))_R$ are copied to $!(C_2\with (!A_2\linimpl
N_2))$, and P's responses following $\sigma''$ are copied back to
$!(C_0\with (!A_0\linimpl N_0))_R$.

As regards sequencing, the initial move $*_5$ is copied
immediately as $*_{i_0,L}$. Opponent may now either immediately
reply with $n_{i_0,L}$, which will be copied back as $n_5$,
completing the play; or move in $(i_0,!A_0)_L$--- the only other
option by the switching condition. Play then proceeds following
$\sigma''$ transposed to $$\underline{\sigma}'': \;!(C_0 \with (\;
!A_0\linimpl N_0))_R \rightarrow (i_0,\; !A_0)_L,$$ until Opponent
replies with some $n_{i_0,L}$ to $*_{i_0,L}$. Thus positions in
$\tau$ have the form $$ *_5 *_{i_0,L}\; s'\; n_{i_0,L}\; n_5
\mbox{  or  } *_5 \; *_{i_0,L} \; s'$$ where $s'$ is a bijectively
reindexed version of $s\in\sigma''$, with $s\approx s'$. Clearly
$\sigma''\approx  \underline{\sigma}''$, and hence
$\sigma'\approx\tau$.

\end{proof}

We now prove a useful general lemma.
\begin{lemma}\label{lem6}
For all strategies
$\gamma:C\Rightarrow (A\Rightarrow B), \delta: C\Rightarrow A$
$${\tt Ap}\circ\lang\gamma,\delta\rang \approx {\tt con}_C;(\gamma\tensor\delta^{\dag});{\tt LAPP }.$$
\end{lemma}
\begin{proof}
\[
\begin{array}{llc}
{\tt Ap}\circ\lang\gamma,\delta\rang &= {\mbox{definition}}& \\
({\tt con}_C;\gamma^{\dag}\tensor\delta^{\dag};e;{\tt
der}_{(A\Rightarrow
 B)\tensor A})^{\dag};e^{-1};{\tt der}_{A\Rightarrow B} \tensor {\tt id}_A;{\tt LAPP }
&\approx {\mbox{Bang Lemma}}& \\ {\tt
con}_C;\gamma^{\dag}\tensor\delta^{\dag};e;e^{-1}; {\tt
der}_{A\Rightarrow B}\tensor {\tt id}_A;{\tt LAPP } & \approx&\\
{\tt con}_C;\gamma\tensor\delta^{\dag};{\tt LAPP }. & &
\end{array}
\]
\end{proof}

Now consider a game
$$(A_1\with\dots\with A_k)\Rightarrow N$$
where $$A_i=(B_{i,1}\with\dots\with B_{i,l_i})\Rightarrow N ,\;\;\;
1\leq i\leq k.$$
Let $\tilde{A} = A_1\with\dots\with A_k$, $\tilde{B_i} = B_{i,1}\with
\dots\with B_{i,l_i}, \; 1\leq i\leq k$.

We define
$\bot_{\tilde{A}}:\tilde{A}\Rightarrow N$ by
$\bot_{\tilde{A}}=\{\epsilon\}$
and ${\bf K}_{\tilde{A}}n:\tilde{A} \Rightarrow N \;\; (n\in\omega)$
by ${\bf K}_{\tilde{A}}n=\{\epsilon,*n\}$.Thus $\bot_{\tilde{A}}$ is
the completely undefined strategy of type
$\tilde{A}\Rightarrow N$ while ${\bf K}_{\tilde{A}}n$ is the constant
strategy which responds immediately to the initial question in $N$
with the answer $n$.

Finally, if $1\leq i\leq k$, and for each $1\leq j\leq l_i$
$$\sigma_j:\tilde{A}\Rightarrow B_{i,j}$$
and for each $n\in\omega$
$$\tau_n:\tilde{A}\Rightarrow N$$
we define
$$\check{{\bf C_i}}(\sigma_1,\dots,\sigma_{l_i},(\tau_n\;\mid \;
n\in\omega)):\tilde{A}\Rightarrow N$$
by
$$\check{{\bf C_i}}(\sigma_1,\dots,\sigma_{l_i},(\tau_n\;\mid \; n\in\omega))
=\chi\circ\lang {\tt Ap}\circ\lang\pi_i,\lang\sigma_1,\dots,\sigma_{l_i}\rang \rang ,\lang\tau_n\;\mid \;
n\in\omega\rang \rang. $$
\begin{lemma}[The Decomposition Lemma (uncurried version)]\label{lem7}

Let $\sigma:(A_1\with\dots\with A_n)\Rightarrow N$ be any strategy,
where
$$A_i=(B_{i,1}\with\dots\with B_{i,l_i})\Rightarrow N ,\;\;\;
1\leq i\leq k.$$
Then exactly one of the following three cases applies:
\begin{itemize}
\item[(i)] $\sigma=\bot_{\tilde{A}}$.
\item[(ii)] $\sigma={\bf K}_{\tilde{A}}n$ for some $n\in\omega$.
\item[(iii)] $\sigma\approx
  \check{{\bf C_i}}(\sigma_1,\dots,\sigma_{l_i},(\tau_n\;\mid \; n\in\omega))$

where $1\leq i\leq k$, $\sigma_j:\tilde{A}\Rightarrow B_{i,j},\; 1\leq
j\leq l_i$, $\tau_n:\tilde{A}\Rightarrow N,\; n\in\omega$ .
\end{itemize}
\end{lemma}

\begin{proof} By cases on $\sigma$'s response to the initial
question. If it has no response, we are in case {\it (i)}. If its
response is an immediate answer $n$ for some $n\in\omega$, we are
in case {\it (ii)}. Otherwise, $\sigma$ must respond with the
initial question in the $i$'th argument, for some $1\leq i \leq
k$. In this case, write $C=A_1\with\dots\with A_{i-1}\with
A_{i+1}\with\dots\with A_k$. We have the natural isomorphism
$$\alpha : \; !(C\with A_i)\cong \; !\tilde{A} :\alpha^{-1}$$ so
we can apply Lemma \ref{lem4} to conclude that
$$\alpha;\sigma\approx \chi\circ \lang \sigma',\lang \tau_n\;\mid
\; n\in\omega\rang \rang $$ By Lemma \ref{lem5}
$$\sigma'\approx{\tt con};\pi_2\tensor \sigma'';{\tt LAPP }$$
where $\sigma'':C\with A_i\Rightarrow !\tilde{B_i}$. By the Bang
Lemma, $$\sigma''\approx(\sigma'';{\tt der})^{\dag}.$$ Moreover
$$\sigma'';{\tt der}_B:C\with A_i\Rightarrow (B_{i,1}\with
\dots\with B_{i,l_i})$$ so by the universal property of the
product, $$\sigma'';{\tt der}_B\approx\lang
\sigma_1,\dots,\sigma_{l_i}\rang $$ where $\sigma_j:C\with
A_i\Rightarrow B_{i,j},\;\; 1\leq j\leq l_i$.

Thus $\sigma'\approx{\tt con};\pi_2\tensor
\lang \sigma_1,\dots,\sigma_{l_i}\rang^{\dag};{\tt LAPP } $ and by Lemma \ref{lem6},
$$\sigma'\approx {\tt Ap }\circ\lang \pi_2,\lang \sigma_1,\dots,\sigma_{l_i}\rang \rang $$
Thus
\[
\begin{array}{lcl}
\sigma &\approx & \alpha^{-1};\alpha;\sigma \\
 &\approx&
 \alpha^{-1};(\chi\circ\lang {\tt Ap }\circ\lang \pi_2,\lang \sigma_1,\dots,\sigma_{l_i}\rang \rang ,\lang \tau_n\;\mid \; n\in\omega\rang \rang  \\
&\approx & \chi\circ \lang {\tt Ap }\circ\lang \pi_i,\lang
\alpha^{-1};\sigma_1,\dots,\alpha^{-1};\sigma_{l_i}\rang \rang ,
\lang \alpha^{-1};\tau_n\;\mid \; n\in\omega\rang \rang  \\ &=&
\check{{\bf
C}_i}(\alpha^{-1};\sigma_1,\dots,\alpha^{-1};\sigma_{l_i},
(\alpha^{-1};\tau_n\;\mid \;n\in\omega)) .
\end{array}
\]
\end{proof}

The Decomposition Lemma in its uncurried version is not sufficiently
general for our purposes. Suppose now that we have a game
$$(A_1\with\dots\with A_k)\Rightarrow N$$
where
$$A_i=B_{i,1}\Rightarrow \dots B_{i,l_i}\Rightarrow N, \; \;\; (1\leq
i\leq l_i).$$

If for some $1\leq i\leq k$ and each $1\leq j\leq l_i$ we have
$$\sigma_j:\tilde{A}\Rightarrow B_{i,j}$$
and for each $n\in\omega$
$$\tau_n:\tilde{A}\Rightarrow N$$
then we define
$${\bf C}_i(\sigma_1,\dots,\sigma_{l_i},(\tau_n\;\mid \;
n\in\omega)):\tilde{A}\Rightarrow N$$
by
$${\bf C}_i(\sigma_1,\dots,\sigma_{l_i},(\tau_n\;\mid \; n\in\omega))=
\chi\circ\lang {\tt Ap }\circ\lang \dots
{\tt Ap }\circ\lang \pi_i,\sigma_1\rang ,\dots,\sigma_{l_i}\rang ,\lang \tau_n\;\mid \;
n\in\omega\rang \rang .$$
To relate ${\bf C}_i$ and $\check{{\bf C}_i}$, consider the canonical isomorphisms
$$\alpha_i: B_{i,1}\Rightarrow\dots B_{i,l_i}\Rightarrow N\cong
(B_{i,1}\with\dots\with B_{i,l_i})\Rightarrow N:\alpha_i^{-1} \;\;
(1\leq i\leq k)$$
Let $\tilde{\alpha}=!(\alpha_1\with\dots\with \alpha_k)$ so
$$\tilde{\alpha}:!(A_1\with\dots\with A_k)\cong !(A_1^u\with\dots\with
A_k^u)$$
where $A_i^u=(B_{i,1}\with\dots\with B_{i,l_i})\Rightarrow N$ is the
uncurried version of $A_i$.
Then
$${\bf C}_i(\sigma_1,\dots,\sigma_{l_i},(\tau_n\;\mid \; n\in\omega))\approx
\tilde{\alpha};\check{{\bf C}_i}(\tilde{\alpha};\sigma_1,\dots,\tilde{\alpha};
\sigma_{l_i},(\tilde{\alpha};\tau_n\;\mid \; n\in\omega)) \;\;\;\;(1)$$
In terms of $\lambda-$calculus, this just boils down to the familiar
equations
$${\tt Curry}(f)xy=f(x,y)$$
$${\tt Uncurry}(g)(x,y)=gxy$$
To see the relationship between the combinators $\bot, {\bf K}n$ and
$\bf C$ and
the syntax of PCF, we use the combinators to write the semantics of
finite evaluation trees.

\newcommand{\sem}[1]{{\cal S}( #1)}
Given $P\in {\bf FET}(\Gamma,T)$ where $\Gamma=x_1:T_1,\dots,x_k:T_k$, we
will define
$$\sem{ \Gamma\vdash P:T
  }:(\sem{T_1}\with\dots\with\sem{T_k})\Rightarrow \sem{T}$$

\begin{itemize}
\item $\sem{\Gamma\vdash\lambda x^T.P:T\Rightarrow
    U}=\Lambda(\sem{\Gamma,x:T\vdash P:U})$
\item $\sem{\Gamma\vdash\Omega:N}=\bot_{\tilde{T}}$
\item $\sem{\Gamma\vdash n:N}= {\bf K}_{\tilde{T}}n$
\item $\sem{\Gamma\vdash{\tt case}(x_i P_1\dots P_{l_i},(Q_n\;\mid \;
    n\in\omega)):N}={\bf C}_i(\sigma_1,\dots,\sigma_{l_i},(\tau_n\;\mid \;
  n\in\omega))$
where
$$T_i=U_{i,1}\Rightarrow\dots U_{i,l_i}\Rightarrow N,$$
$$\sigma_j=\sem{\Gamma\vdash P_j:U_{i,j}},\;\;1\leq j\leq l_i,$$
$$\tau_n=\sem{\Gamma\vdash Q_n:N}, \;\; n\in\omega.$$
\end{itemize}
We can now prove the general form of the Decomposition Lemma:
\begin{proposition}[Decomposition Lemma]\label{lem8}
Let $\sigma:(A_1\with\dots\with
A_p)\Rightarrow(A_{p+1}\Rightarrow\dots A_q\Rightarrow N)$ be any
strategy, where
$$A_i=B_{i,1}\Rightarrow\dots B_{i,l_i}\Rightarrow N,\;\; 1\leq i\leq
q$$
We write $\tilde{C}=A_1,\dots,A_p,\; \tilde{D}=A_{p+1},\dots,A_q$.
(Notation : if $\tau:\tilde{C},\tilde{D}\Rightarrow N$, then
$\Lambda_{\tilde{D}}(\tau):\tilde{C}\Rightarrow(A_{p+1}\Rightarrow\cdots
\Rightarrow A_q\Rightarrow N)$.)

Then exactly one of the following three cases applies.
\begin{itemize}
\item[(i)] $\sigma=\Lambda_{\tilde{D}}(\bot_{\tilde{C},\tilde{D}}).$
\item[(ii)] $\sigma=\Lambda_{\tilde{D}}({\bf K}_{\tilde{C},\tilde{D}} n)$
  for some $n\in\omega$.
\item[(iii)]
  $\sigma=\Lambda_{\tilde{D}}({\bf C}_i(\sigma_1,\dots,\sigma_{l_i},
(\tau_n\mid n\in\omega))),$ where $1\leq i\leq q$, and
\[\begin{array}{lllc}
\sigma_j:\tilde{C},\tilde{D} & \Rightarrow & B_{i,j}, & 1\leq j\leq
l_i, \\
\tau_n:\tilde{C},\tilde{D} & \Rightarrow & N, & n\in\omega.
\end{array}\]
\end{itemize}

\end{proposition}
\begin{proof} Let $\alpha_i:A_i\cong A_i^u:\alpha^{-1}$ be the
canonical isomorphism between $A_i$ and its uncurried version
$$A_i^u=(B_{i,1}\with\dots\with B_{i,l_i})\Rightarrow N$$ for each
$1\leq i\leq q$.

Let
$$\tilde{\alpha}=!(\alpha_1\with\dots\with\alpha_p\with\alpha_{p+1}\with\dots\with\alpha_q).$$
Note that
\[\begin{array}{rclr}
\bot_{\tilde{C},\tilde{D}} &=& \tilde{\alpha};\bot_{\tilde{C}^u,\tilde{D}^u}
& (2) \\
{\bf K}_{\tilde{C},\tilde{D}}n& =& \tilde{\alpha};{\bf
  K}_{\tilde{C}^u,\tilde{D}^u}n
& (3).
\end{array}\]

We can apply Lemma \ref{lem7} to
$\check{\sigma}=\tilde{\alpha}^{-1};\Lambda^{-1}_{\tilde{D}}(\sigma):
\tilde{C}^u,\tilde{D}^u\Rightarrow N$. The result now follows from
equations (1)--(3) since $$\sigma\approx
\Lambda_{\tilde{D}}(\tilde{\alpha};\check{\sigma}). \;\; $$
\end{proof}

With the same notations as in the Decomposition Lemma:

\begin{lemma}[Unicity of Decomposition]\label{lem9}
\begin{itemize}
\item[(i)] If $\sigma\approx\bot_{\tilde{C},\tilde{D}}$ then
  $\sigma=\bot_{\tilde{C},\tilde{D}}.$
\item[(ii)] If $\sigma\approx {\bf K}_{\tilde{C},\tilde{D}}n$ then
  $\sigma={\bf K}_{\tilde{C},\tilde{D}}n.$
\item[(iii)] If
$  {\bf C}_i(\sigma_1,\dots,\sigma_{l_i},(\tau_n\;\mid \; n\in\omega))\Subeq
  {\bf C}_i(\sigma'_1,\dots,\sigma'_{l_i},(\tau'_n\;\mid \; n\in\omega)) $ then
\[\begin{array}{ll}
\sigma_j\Subeq \sigma'_j, &  1\leq j\leq l_i ,\\
\tau_n\Subeq \tau'_n, &  n\in\omega.
\end{array}\]
\end{itemize}
\end{lemma}
\begin{proof} {\it (i)} and {\it (ii)} are trivial.

For {\it (iii)} write $\sigma={\bf C}_i(\sigma_1,\dots,\sigma_{l_i},(\tau_n\;\mid \;
n\in\omega))$ and
$\tau={\bf C}_i(\sigma'_1,\dots,\sigma'_{l_i},(\tau'_n\;\mid \; n\in\omega))$.

Suppose firstly that $s\in\tau_n$. Then $*_1*_2n_2 s\in\sigma$, so
since $\sigma\Subeq\tau$, for some $t$, $*_1*_2n_2 t\in\tau$ and
$*_1*_2n_2 s\approx*_1*_2n_2 t$. This implies that $t\in\tau'_n$ and
$s\approx t$. We conclude that $\tau_n\Subeq\tau'_n$.

Now suppose that $s\in\sigma_j$. Then $*_1*_2s'\in\sigma$ where
$s'$ is a reindexed version of $s$ with $s\approx s'$. Since
$\sigma\Subeq\tau$, there exists $t'$ such that $*_1*_2t'\in\tau $
and $*_1*_2s'\approx *_1*_2 t'$. This implies that there exists
$t\in\sigma'_j$ with $s\approx t$. We conclude that
$\sigma_j\Subeq\sigma'_j.\;\; \; $
\end{proof}

\subsection{Approximation Lemmas}
The Decomposition Lemma provides for one step of decomposition of an
arbitrary strategy into a form matching that of the semantic clauses
for evaluation trees. However, infinite strategies will not admit a
well-founded inductive decomposition process. Instead, we must appeal
to notions of continuity and approximation, in the spirit of Domain
Theory \cite{AbramskyS:domt}.

We define a PCF {\em type-in-context} (\cite{CroleRL:catt}) to be a type of
the form
$$(T_1\with\dots\with T_p)\Rightarrow U$$
where $T_1,\dots,T_p,U$ are PCF types. Given such a type-in-context
$T$, we will write ${\tt Str}(T)$ for the set of strategies on the game
$\sem{T}$.

The Unicity of Decomposition Lemma says that decompositions are unique up
to partial equivalence. Referring to the Decomposition Lemma,
Prop. \ref{lem8}, note that the proof of the decomposition
$$\sigma\approx {\bf C}_i(\sigma_1,\dots,\sigma_{l_i},(\tau_n\;\mid \;
n\in\omega))$$
involved defining specific strategies
$\sigma_1,\dots,\sigma_{l_i},(\tau_n\;\mid \; n\in\omega)$ from the given
$\sigma$. If we also fix specific pairing and tagging functions and
dereliction indices in the definition of promotion,
dereliction, contraction etc.( and hence in the $\Intmod$
operations of composition, pairing, currying etc.), we obtain an
operation $\Phi$ on strategies such that
\[ \Phi(\sigma)= \left\{ \begin{array}{ll}
            1 & \mbox{ in case {\it (i)} } \\
            (2,n) & \mbox{ in case {\it (ii)} } \\
           (3,\sigma_1,\dots,\sigma_{l_i},(\tau_n\;\mid \; n\in\omega))
           &\mbox{ in case {\it (iii)}}
        \end{array}
\right. \]
according to the case of the Decomposition Lemma which applies to
$\sigma$. We shall use $\Phi$ to define a family of
functions
$$p_k:{\tt Str}(T)\rightarrow {\tt Str}(T) \;\; (k\in\omega)$$
inductively as follows:
\begin{itemize}
\item $p_0(\sigma)=\Lambda_{\tilde{U}}(\bot_{\tilde{T},\tilde{U}})$
\item \[ p_{k+1}(\sigma)= \left\{ \begin{array}{ll}
\Lambda_{\tilde{U}}(\bot_{\tilde{T},\tilde{U}}), &
\Phi(\sigma)=1 \\

\Lambda_{\tilde{U}}({\bf K}_{\tilde{T},\tilde{U}}n) , &
\Phi(\sigma)=(2,n) \\

\Lambda_{\tilde{U}}({\bf C}_i(p_k(\sigma_1),\dots,p_k(\sigma_{l_i}),(\tau'_n\;\mid \;
n\in\omega))), & \Phi(\sigma) =\sigma_0

\end{array}
\right. \]
where
$$\sigma_0=
(3,\sigma_1,\dots,\sigma_{l_i},(\tau_n\mid n\in\omega) ) $$
and
\[ \tau'_n= \left\{ \begin{array}{ll}
p_k(\tau_n), & 0\leq n\leq k \\
\Lambda_{\tilde{U}}(\bot_{\tilde{T},\tilde{U}}), & n > k .
\end{array}
\right. \]

\end{itemize}
The principal properties of these functions are collected in the
following Lemma.

\begin{lemma}[Approximation Lemma for Strategies]\label{lemm1}
For all $k\in\omega$:
\begin{itemize}
\item[(i)] $\sigma\subseteq \tau$ implies $p_k(\sigma)\subseteq p_k(\tau)$
\item[(ii)] If $\sigma_0\subseteq \sigma_1\subseteq\dots $ is an increasing
  sequence,
$$ p_k(\bigcup_{l\in\omega}\sigma_l)=\bigcup_{l\in\omega}
p_k(\sigma_l)$$
\item[(iii)] $\sigma\Subeq \tau$ implies $p_k(\sigma)\Subeq p_k(\tau)$
\item[(iv)] $p_k(\sigma)\Subeq\sigma$
\item[(v)] $\forall s\in\sigma.\; |s|\leq 2k \Rightarrow \exists t\in
  p_k(\sigma).\; s\approx t$
\item[(vi)] $p_k(\sigma)\subseteq p_{k+1}(\sigma)$
\item[(vii)] $\bigcup_{l\in\omega} p_l(\sigma)\approx\sigma$
\item[(viii)] $p_k(p_k(\sigma))\approx p_k(\sigma)$
\end{itemize}
\end{lemma}
\begin{proof} Firstly, consider the operation $\Phi(\sigma)$. In
case {\it (iii)}, where
$$\Phi(\sigma)=(3,\sigma_1,\dots,\sigma_{l_i},(\tau_n\;\mid \;
n\in\omega))$$ $\Phi(\sigma)$ is obtained by firstly defining
$\sigma'$ and the $\tau_n$ from $\sigma$, then $\sigma''$ from
$\sigma'$, and finally $$\sigma_j=(\sigma'';{\tt der})^{\dag};
\pi_j.$$ Note that $\sigma';\sigma''$ and the $\tau_n$ are defined
{\em
  locally}, i.e. by operations on positions applied pointwise to
$\sigma$ and $\sigma'$ respectively. Together with the
$\subseteq-$monotonicity and continuity of Promotion, Dereliction,
Contraction etc. (Proposition 2.9.4) this implies {\it (i)} and
{\it (ii)}. Now note that ${\bf C}_i$ is $\subseteq-$ and
$\Subeq_{A}$ monotonic by Proposition 2.9.3. A straightforward
induction using $\Subeq-$monotonicity and $\subseteq-$monotonicity
of ${\bf C}_i$ respectively and the Unicity of Decomposition Lemma
yelds {\it (iii)}. Similarly routine inductions using
$\Subeq-$monotonicity and $\subseteq-$monotonicity of ${\bf C}_i$
respectively prove {\it (iv)} and {\it (vi)}.

We prove {\it (v)} by induction on $k$. The basis is trivial as are cases
{\it (i)} and {\it (ii)} of the Decomposition Lemma at the inductive step. Suppose
we are in case {\it (iii)}, with
$$\sigma\approx {\bf C}_i(\sigma_1,\dots,\sigma_{l_i},(\tau_n\;\mid \;
n\in\omega))$$
Consider firstly $s\in\sigma$ where $s=*_1*_2s'$ with $*_2$ not
answered in $s'$. Then $s'\in\sigma''$ where $\sigma''$ is derived
from $\sigma'$ and $\sigma'$ from $\sigma$ as in the proof of the
Decomposition Lemma. Since
$\lang \sigma_1,\dots,\sigma_{l_i}\rang^{\dag}\approx\sigma''$, $s'$ can be
decomposed into subsequences $s_{j,1},\dots,s_{j,p_j}$ with
$s'_{j,q}\approx s_{j,q}\in\sigma_j$, $1\leq j\leq l_i$, $1\leq q\leq
p_j$.

Since $|s_{j,q}| < |s|$, we can apply the induction hypothesis to
conclude that $s_{j,q}\approx u_{j,q}\in p_k(\sigma_j)$, and hence
that there is $*_1*_2 u\in p_{k+1}(\sigma)$ with $s\approx *_1*_2
u$. The case where $s=*_1*_2 s' n_2 t$ is similar.

To prove {\it (vii)}, note firstly that the union $\bigcup_{l\in\omega}
p_l(\sigma)$ is well-defined by {\it (vi)}. Now $\bigcup_{l\in\omega}
p_l(\sigma)\Subeq\sigma$ follows from {\it (iv)}, while $\sigma\Subeq\bigcup_{l\in\omega}
p_l(\sigma)$ follows from {\it (v)}.

Finally {\it (viii)} can be proved by induction   on $k$ and {\it
  (iii)} using the Unicity of Decomposition Lemma.
\end{proof}

We now turn to evaluation trees. Let
$\Gamma=x_1:T_1,\dots,x_k:T_k$. We define a family of functions
$$q_k:{\bf ET}(\Gamma,U)\rightarrow {\bf ET}(\Gamma,U) \;(k\in\omega)$$
inductively by

\[
\begin{array}{lcr}
q_0(P) & = & \lambda \tilde{x}^{\tilde{U}}.\Omega \\
q_{k+1}(\lambda \tilde{x}^{\tilde{U}}.\Omega) &=& \lambda
\tilde{x}^{\tilde{U}}.\Omega \\
q_{k+1}(\lambda \tilde{x}^{\tilde{U}}.n) &=& \lambda
\tilde{x}^{\tilde{U}}.n \\
q_{k+1}(\lambda \tilde{x}^{\tilde{U}}.{\tt case}(x_i P_1\dots
p_{l_i},(Q_n\;\mid \; n\in\omega))) & & \\
=\lambda \tilde{x}^{\tilde{U}}.{\tt case}(x_i q_k(P_1)\dots
q_k(P_{l_i}),(Q'_n\;\mid \; n\in\omega)) & &
\end{array}
\]

where

\[
Q'_n= \left\{ \begin{array}{cc}
q_k(Q_n) , & 0\leq n\leq k\\
\lambda \tilde{x}^{\tilde{U}}.\Omega, & n> k
\end{array}
\right. \]

The following is then standard:
\begin{lemma}[Approximation Lemma for Evaluation Trees]\label{alet}
The $(q_k\;\mid \; k\in\omega)$ form as increasing sequence of
continuous functions with $\bigsqcup_{k\in\omega} q_k={\tt
id}_{{\bf ET}(\Gamma,U)}$. Each $q_k$ is idempotent and has finite
image.
\end{lemma}

\subsection{Main Results}
We are now equipped to address the relationship between strategies
and evaluation trees directly. Let $\Gamma=x_1:T_1,\dots,x_k:T_k$.
We define a map $$\varsigma:{\bf FET}(\Gamma,U)\rightarrow {\tt
Str}(\tilde{T}\Rightarrow U)$$ this map is a concrete version of
the semantic map defined in section 2.4. That is, we fix choices
of pairing functions etc. as in the definition of $\Phi$ in 2.5,
and define $\varsigma(\Gamma\vdash P:U)$ as a specific
representative of the partial equivalence class $\sem{\Gamma\vdash
P:U}$. Thus we will have $$\sem{\Gamma\vdash
P:U}=[\varsigma(\Gamma\vdash P:U)].$$ We were sloppy about this
distinction in 2.4; we give the definition of $\varsigma$
explicitly for emphasis:
\[
\begin{array}{lcr}
\varsigma(\Gamma\vdash \lambda x^T.P:T\Rightarrow U) & = &
\Lambda(\varsigma(\Gamma,x:T\vdash P : U)) \\
\varsigma(\Gamma\vdash \Omega:N) &=& \bot_{\tilde{T}} \\
\varsigma(\Gamma\vdash n:N) &=& {\bf K}_{\tilde{T}}n\\
\varsigma(\Gamma\vdash {\tt case}(x_i P_1\dots P_{l_i},(Q_n\;\mid \;
n\in\omega)) &= &
{\bf C}_i(\sigma_1,\dots,\sigma_{l_i},(\tau_n\;\mid \; n\in\omega))
\end{array} \]
where
\[\begin{array}{lcr}
T_i &=& B_{i,1}\Rightarrow\dots\Rightarrow B_{i,l_i}\Rightarrow N ,\\
\sigma_j &=& \varsigma(\Gamma\vdash P_j:B_{i,j}),\; 1\leq j\leq l_i ,\\
\tau_n &=& \varsigma(\Gamma\vdash Q_n:N) , \; n\in\omega .
\end{array} \]

\begin{lemma}\label{lemmm1}
If $P\sqsubseteq Q$ then $\varsigma(\Gamma\vdash
P:U)\subseteq\varsigma(\Gamma\vdash Q:U)$
\end{lemma}
\begin{proof} By induction on the construction of $P$, using
$\subseteq$--monotonicity  of ${\bf C}_i$.
\end{proof}

Let $\tilde{T}=T_1,\dots,T_l$, and ${\tt Con}(\tilde{T})$ be the set of
all $\tilde{T}-$contexts $x_1:T_1,\dots,x_p:T_p$.
For each $k\in\omega$, we define a map
$$\eta_k: {\tt Str}(\tilde{T}\Rightarrow U) \rightarrow
\Pi_{\Gamma\in{\tt Con}(\tilde{T})} {\bf FET}(\Gamma,U)$$
inductively by:

\[
\begin{array}{lcr}
\eta_0(\sigma)\Gamma &=& \lambda \tilde{y}^{\tilde{U}}.\Omega \\

\end{array}
\]

\[\eta_{k+1}(\sigma)\Gamma=\left\{ \begin{array}{l}
 \lambda \tilde{y}^{\tilde{U}}.\Omega,
 \sigma=\Lambda_{\tilde{U}}(\bot_{\tilde{T},\tilde{U}}) \\
  \lambda \tilde{y}^{\tilde{U}}.n,
  \sigma=\Lambda_{\tilde{U}}({\bf K}_{\tilde{T},\tilde{U}} n) \\
 \lambda \tilde{y}^{\tilde{U}}.{\tt case}(z_i P_1\dots
 P_{l_i},(Q_n\;\mid \; n\in\omega)), \\
 \hspace{1in} \sigma\approx \Lambda_{\tilde{U}}({\bf C}_i(\sigma_1,\dots,\sigma_{l_i},(\tau_n\;\mid \;n\in\omega)))
\end{array}
\right. \]
where
\[\begin{array}{lcr}
\Gamma &=& x_1:T_1,\dots,x_p:T_p, \\
\Delta &=& y_1:U_1,\dots, y_q:U_q \\
\tilde{z}&=& x_1,\dots,x_p,y_1,\dots,y_q, \\
P_j &=& \eta_k(\sigma_j)\Gamma,\Delta, \;\; 1\leq j\leq l_i
\end{array} \]
and
\[ Q_n = \left\{ \begin{array}{cc}
\eta_k(\sigma_j)\Gamma,\Delta, & 0\leq n\leq k \\
\Omega & n> k
\end{array}
\right. \]

\begin{lemma}\label{lemmm2}
For all $k\in\omega$ :
\begin{itemize}
\item[(i)] $\sigma\Subeq\tau$ implies
  $\eta_k(\sigma)\Gamma\sqsubseteq\eta_k(\tau)\Gamma$ .
\item[(ii)] If $\sigma_0\subseteq\sigma_1\subseteq\dots$ is an
  increasing sequence,
$$
\eta_k(\bigcup_{l\in\omega}\sigma_l)\Gamma=\bigsqcup_{l\in\omega}\eta_k(\sigma_l)\gamma.$$
\item[(iii)] $\eta_k(\sigma)\Gamma\sqsubseteq\eta_{k+1}(\sigma)\Gamma$
\item[(iv)] $q_k(\eta_l(\sigma)\Gamma)=\eta_k(\sigma)\gamma,\;\; l\ge k$
\end{itemize}
\end{lemma}
\begin{proof} {\it (i)} is proved similarly to part {\it (iii)} of
the Approximation Lemma for strategies; {\it (ii)} is proved
similarly to part {\it (ii)}; and {\it (iii)} to part {\it (vi)};
{\it (iv)} is proved by a routine induction on $k$.
\end{proof}

\begin{lemma}\label{lemmm3}
For all $P\in{\bf FET}(\Gamma,U),
\;\sigma\in{\tt Str}(\tilde{T}\Rightarrow U),\; k\in\omega :$
\begin{itemize}
\item[(i)] $\eta_k(\varsigma(\Gamma\vdash P:U))\Gamma=q_k(P)$
\item[(ii)] $\varsigma(\Gamma\vdash(\eta_k(\sigma)\Gamma):U)\approx p_k(\sigma)$
\end{itemize}
\end{lemma}
\begin{proof} Both parts are proved by induction on $k$. The
induction bases are trivial as are cases {\it (i)} and {\it (ii)}
of the Decomposition Lemma at the inductive step, and the
corresponding cases on the construction of $P$

{\it (i)}

\[\begin{array}{lcr}
\eta_{k+1}(\varsigma(\Gamma\vdash \lambda \tilde{y}^{\tilde{U}}
. {\tt case}(z_i P_1\dots P_{l_i},(Q_n\;\mid \; n\in\omega )))) &=& \\
\lambda \tilde{y}^{\tilde{U}}.{\tt case}(z_i P'_1\dots
P'_{l_i},(Q'_n\;\mid \;n\in\omega)) &&
\end{array}
\]
where
\[
\begin{array}{llr}
P'_j &=& \eta_k(\varsigma(\Gamma,\Delta \vdash
P_j:B_{i,j}))\Gamma,\Delta \\
  &= {\mbox{ind.hyp}}& q_k(P_j)
\end{array}
\]

\[ \begin{array}{lll}

Q'_n & = & {\left\{ \begin{array}{cc}
\eta_k(\varsigma(\Gamma,\Delta\vdash Q_n:N))\Gamma,\Delta,& 0\leq
n\leq k\\
\Omega & n>k
\end{array}
\right . } \\
 & = {\mbox{ind.hyp}} & {\left\{ \begin{array}{cc}
q_k(\Omega) & 0\leq n\leq k \\
\Omega & n> k
\end{array}
\right . } \\
\end{array}
\]

{\it (ii)}
\[
\begin{array}{lcr}
\varsigma(\Gamma\vdash
\eta_{k+1}({\bf C}_i(\sigma_1,\dots,\sigma_{l_i},(\tau_n\;\mid \;n\in\omega)))\Gamma:U
 &\approx& \\
\Lambda_{\tilde{U}}({\bf C}_i(\sigma'_1,\dots,\sigma'_{l_i},(\tau'_n\;\mid \;
n\in\omega))) &&
\end{array}\]
where
\[
\begin{array}{lll}
\sigma'_j& \approx&
\varsigma(\Gamma,\Delta\vdash(\eta_k(\sigma_j)\Gamma,\Delta):U)  \\
 &\approx_{\mbox{ind.hyp}} & p_k(\sigma_j)
\end{array}
\]

\[\begin{array}{lll}
\tau'_n & \approx& \left\{ \begin{array}{cc}
 \varsigma(\Gamma,\Delta\vdash (\eta_k(\tau_n)\Gamma,\Delta):N ) &
 0\leq n\leq k \\
 \bot_{\tilde{T},\tilde{U}} & n>k
             \end{array} \right.  \\
  &\approx_{\mbox{ind.hyp}}& \left\{ \begin{array}{cc}
p_k(\tau_n) & 0\leq n\leq k \\
\bot_{\tilde{T},\tilde{U}} & n>k
    \end{array} \right.
\end{array} \]
\end{proof}

Now we define functions

$${\cal S} :{\tt ET}(\Gamma,U)\rightarrow {\tt Str}(\tilde{T},U)$$
$${\cal E} : {\tt Str} (\tilde{T},U)\rightarrow {\tt ET}(\Gamma,U)$$
by:
$${\cal S}(P) =\bigcup_{k\in\omega} \varsigma(\Gamma\vdash q_k(P):U)$$
$${\cal E}(\sigma) =\bigsqcup_{k\in\omega}\eta_k(\sigma)\Gamma$$
By Lemma \ref{lemmm1} and the Approximation Lemma for evaluation trees,
$(\varsigma(\Gamma\vdash q_k(P):U)\;\mid \; k\in\omega)$ is an
$\subseteq-$increasing sequence of strategies, so ${\cal S}$ is
well-defined. Similarly, by Lemma \ref{lemmm2} $\cal E$ is well-defined.

We now prove the key result on definability.
\begin{theorem}[Isomorphism Theorem]\label{isotheo}

\begin{itemize}
\item[(i)] For all $P\in{\tt ET}(\Gamma,U)$
$$ {\cal E}\circ{\cal S}(P)=P$$
\item[(ii)] For all $\sigma\in{\tt Str}(\tilde{T}\Rightarrow U)$,
$${\cal S}\circ {\cal E}(\sigma)\approx \sigma$$
\item[(iii)] Let $T=\tilde{T}\Rightarrow U$. Then there is an
  order-isomorphism
$${\cal S}_{\approx}:{\tt ET}(\Gamma,U)\simeq \sem{\hat{T}}:{\cal
  E}_{\approx}$$
where ${\cal S}_{\approx}(P)=[{\cal S}(P)]$ (i.e. the partial
equivalence class of ${\cal S}(P)$), and ${\cal E}_{\approx}([\sigma
])={\cal E}(\sigma)$.
\end{itemize}
\end{theorem}
\begin{proof}

{\it (i)}

\[ \begin{array}{lll}
{\cal E}\circ {\cal S} (P) && \\
 &= {\mbox{definition}}&
 \bigsqcup_{k\in\omega}\eta_k(\bigcup_{l\in\omega}\varsigma(\Gamma
 \vdash q_l(P):U))\Gamma \\
&= {\mbox{Lemma \ref{lemmm2}{\it (ii)}}}&
\bigsqcup_{k\in\omega}\bigsqcup_{l\in\omega}
\eta_k(\varsigma(\Gamma\vdash q_l(P):U))\Gamma \\
&=&\bigsqcup_{n\in\omega} \eta_n(\varsigma(\Gamma\vdash
q_n(P):U))\Gamma \\
&= {\mbox{Lemma \ref{lemmm3}}}& \bigsqcup_{n\in\omega} q_n\circ q_n(P)
\\
&= {\mbox{Lemma \ref{alet}}}& P.

\end{array} \]

{\it (ii)}

\[\begin{array}{lll}
{\cal S}\circ{\cal E} (\sigma) &&\\
 &=& \bigcup_{k\in\omega} \varsigma(\Gamma\vdash
 q_k(\bigsqcup_{l\in\omega}\eta_l(\sigma)\Gamma):U) \\
&= {\mbox{continuity of }q_k} &
\bigcup_{k\in\omega}\varsigma(\Gamma\vdash \bigsqcup_{l\in\omega}
q_k(\eta_l(\sigma)\Gamma): U) \\
&= {\mbox{Lemma \ref{lemmm2}(iv)}}& \bigcup_{k\in\omega}
\varsigma(\Gamma\vdash(\eta_k(\sigma)\Gamma):U) \\
&\approx {\mbox{Lemma \ref{lemmm3}}} & \bigcup_{k\in\omega} p_k(\sigma)\\
&\approx {\mbox{Lemma \ref{lemm1}}} & \sigma.
\end{array} \]

{\it (iii)} Firstly ${\cal E}_{\approx}$ is well-defined and monotone by
Lemma \ref{lemmm2}{\it (i)}. Also, ${\cal S}_{\approx}$ is monotone by
Lemma \ref{lemmm1}. By {\it (i)} and {\it (ii)}, ${\cal E}_{\approx}={\cal
  S}_{\approx}^{-1}. \;\;$
\end{proof}

As an immediate corollary of the Isomorphism Theorem and Proposition
3.2.2:
\begin{proposition}
For each PCF type $T$, $\sem{\hat{T}}$ is a dI-domain. Hence $\Intmod$
is an algebraic cpo-based model.
\end{proposition}
Thus although a priori we only knew that $\Intmod$ was a rational
model, by virtue of the Isomorphism theorem we know that the carrier
at each PCF type is an algebraic cpo. Hence the notion of {\em
  intensional full abstraction} makes sense for $\Intmod$. Recall
from the introduction that a model is intensionally fully abstract for a
language $\cal L$ if every compact element of the model is denoted by
a term of $\cal L$.

We can now prove the culminating result of this section.

\begin{theorem}[Intensional Full Abstraction]

$\Intmod$ is intensionally fully abstract for $\pcfc$.
\end{theorem}
\begin{proof} Consider any PCF type $T$. By the Isomorphism Theorem,
the compact elements of $\sem T$ are the image under ${\cal
S}_{\approx}$ of the compact elements of ${\bf ET}(\Gamma_0,T)$
(where $\Gamma_0$ is the empty context). But the compact elements
of ${\bf ET}(\Gamma_0,T)$ are just the finite evaluation trees
${\bf FET}(\Gamma_0,T)$ and the restriction of ${\cal
S}_{\approx}$ to ${\bf FET}(\Gamma_0,T)$ is the semantic map $\sem
.$ on finite evaluation trees {\sl qua} terms of $\pcfc$.
\end{proof}

\section{Extensional Full Abstraction}
\newcommand{\converges}{{\downarrow}}
\newcommand{\Converges}{{\Downarrow}}
\newcommand{\diverges}{{\uparrow}}
\newcommand{\Siepinski}{{\Sigma}}
\subsection{The Intrinsic Preorder}

We define the {\em Sierpinski game} $\Siepinski$ to be the game
$${\Siepinski}=(\{
q,a\},\{(q,OQ),(a,PA)\},\{\epsilon,q,qa\},{\tt id}_{P_{\Siepinski}})$$
with one initial question $q$, and one possible response $a$. Note
that $\hat{\Siepinski}$  is indeed the usual Sierpinski space. i.e. the
two-point lattice $\bot < \top$ with $\bot=\{\epsilon\},
\top=\{\epsilon,qa\}$.

Now for any game $A$ we define the {\em intrinsic preorder} $\Ip A$ on
${\tt Str}(A)$ by:
$$ x\Ip A y \Longleftrightarrow \forall\alpha :A\rightarrow {\Siepinski}.\; x;\alpha \Subeq y;\alpha$$
Note that if we write $x\converges \equiv x=\top$ and $x\diverges\equiv
x=\bot$, then:
$$x\Ip A y \Longleftrightarrow \forall\alpha :A\rightarrow {\Siepinski}.\; x;\alpha\converges \supset y;\alpha\converges$$
It is trivially verified that $\Ip A$ is a preorder.
\begin{lemma}[Point Decomposition Lemma]\label{PDL}

\begin{itemize}
\item[(i)] $\forall x\in{\tt Str}(! A).\; x\approx (x;{\tt
    der}_A)^{\dag}= !(x;{\tt der}_A)$
\item[(ii)] $\forall x\in{\tt Str}(A\with B).\; x\approx
  \lang x;{\tt fst},x;{\tt snd}\rang $
\item[(iii)]$\forall x\in{\tt Str}(A\otimes B).\; \exists
  y\in{\tt Str}(A),z\in{\tt Str}(B).\; x\approx y\otimes z$
\end{itemize}
\end{lemma}
\begin{proof} Firstly we must explain the notation. We think of a
strategy $\sigma$ in $A$ indifferently as having the type $\sigma:
I\rightarrow A$. Now since $!I=I$, we can regard
$!\sigma:!I\rightarrow !A$ as in ${\tt Str}(!A)$. Similarly, since
$I\otimes I=I$, we can regard $\sigma\otimes\tau$ as in ${\tt
Str}(A\otimes B)$, where $\sigma\in{\tt Str}(A),\tau\in{\tt
Str}(B)$. Finally, using $!I=I$ again we can form
$\lang\sigma,\tau\rang \in{\tt Str}(A\with B)$ where
$\sigma\in{\tt Str}(A),\tau\in{\tt Str}(B)$.

Next we note that {\it (i)} is a special case of the Bang Lemma, while {\it (ii)}
follows from the universal property of the product.

Finally, we prove {\it (iii)}. Given $x\in{\tt Str}(A\otimes B)$, write
$x=\sigma_f$, where $f:M_A^O+M_B^O \rightharpoonup M_A^P+M_B^P$. By the
switching condition, we can decompose $f$ as $f=g+h$, where
$g:M_A^O\rightharpoonup M_A^P$, and $h:M_B^O\rightharpoonup M_B^P$. Now
define $y=\sigma_g, z=\sigma_h$. It is clear that $y$ and $z$ are
well-defined strategies, and
$$x=\sigma_f=\sigma_{g+h}\approx\sigma_g\otimes\sigma_h=y\otimes z.$$
\end{proof}

Now we characterise the intrinsic preorder on the Linear types. The
general theme is that ``intrinsic = pointwise''. This is analogous to
results in Synthetic Domain Theory and PER models, although the proofs
are quite different, and remarkably enough no additional hypotheses
are required.
\begin{lemma}[Extensionality for Tensor]\label{EfT}

For all $x\otimes y, x'\otimes y'\in{\tt Str}(A\otimes B)$
$$x\otimes y\Ip{A\otimes B} x'\otimes y' \Longleftrightarrow x\Ip A x'
\wedge y\Ip B y'$$
\end{lemma}
\begin{proof} $(\Rightarrow)$. If $x\otimes y\Ip{A\otimes B}
x'\otimes y'$ and $x;\alpha\converges$, then $x\otimes
y;\beta\converges$ where

\[\begin{diagram}
\beta=A\otimes B& \rTo^{{\tt id}_A\otimes\bot_{B,I}} &A\otimes
I&\rTo^{\sim}& A & \rTo^{\alpha}& {\Siepinski},
\end{diagram} \]
\noindent $\bot_{B,I}=\{\epsilon\}$. This implies that $x\otimes
y;\beta\converges$, and hence that $x';\alpha\converges$. This shows
that $x\Ip A x'$; the proof that $y\Ip B y'$ is similar.

$(\Leftarrow)$. Suppose that $x\Ip A x', y\Ip B y'$ and $x\otimes
y;\gamma\converges$ where $\gamma: A\otimes B\rightarrow{\cal
  \Siepinski}$. Then define $\alpha:A\rightarrow {\Siepinski}$ by:

\[\begin{diagram}
\alpha=A &\rTo^{\sim}& A\otimes I& \rTo^{{\tt id}_A\otimes y} &A\otimes
B &\rTo^{\gamma} & {\Siepinski}
\end{diagram} \]

Then $x;\alpha\approx x\otimes y;\gamma \converges$, so
$x';\alpha\approx x'\otimes y;\gamma\converges$ since $x\Ip A x'$.
This shows that $x\otimes y\Ip{A\otimes B} x'\otimes y$. A similar
argument shows that $x'\otimes y\Ip{A\otimes B} x'\otimes y'$, and
so $$x\otimes y\Ip{A\otimes B} x'\otimes y\Ip{A\otimes B}
x'\otimes y'.\;$$
\end{proof}

\begin{lemma}[Extensionality for Product]\label{EfP}

For all $\lang x,y\rang , \lang x',y'\rang \in{\tt Str}(A\with B)$
$$\lang x,y\rang \Ip{A\with B} \lang x',y'\rang  \Longleftrightarrow x\Ip A x' \wedge y\Ip
B y'$$
\end{lemma}
\begin{proof} By the definition of $A\with B$, any $\gamma:A\with
B\rightarrow {\Siepinski}$ must factor as
\[\begin{diagram}
\gamma=A\with B& \rTo^{{\tt fst}}& A & \rTo^{\alpha} &{\Siepinski}
\end{diagram}\]
or as
\[\begin{diagram}
\gamma=A\with B& \rTo^{{\tt snd}}& B & \rTo^{\beta} &{\Siepinski}
\end{diagram}\]
This shows the right-to-left implication. Conversely, given
$\alpha:A\rightarrow {\Siepinski}$ we can form
\[\begin{diagram}
A\with B& \rTo^{{\tt fst}}& A & \rTo^{\alpha} &{\Siepinski}
\end{diagram}\]
and similarly for $\beta:B\rightarrow{\Siepinski}$.
\end{proof}

\begin{lemma}[Linear Function Extensionality]\label{LFE}

For all $f,g\in{\tt Str}(A\linimpl B)$
$$f\Ip{A\linimpl B} g \Longleftrightarrow \forall x\in{\tt Str}(A),\;
x;f\Ip B x;g$$
\end{lemma}
\begin{proof} $(\Rightarrow)$ Suppose $f\Ip{A\linimpl B} g, x\in{\tt
Str}(A), \beta:B\rightarrow{\Siepinski}$ and
$x;f;\beta\converges$. Then we define $\gamma:(A\linimpl B)
\rightarrow {\Siepinski}$ by
\[\begin{diagram}
\gamma=(A\linimpl B)&\rTo^{\sim} & (A\linimpl B)\otimes
I&\rTo^{{\tt id}_{A\linimpl B} \otimes x} & (A\linimpl B)\otimes
A& \rTo^{{\tt LAPP}} & B &\rTo^{\beta} & {\Siepinski}
\end{diagram}\]
For all $h\in{\tt Str}(A\linimpl B), h;\gamma\approx x;h;\beta$, so
$x;g;\beta\approx g;\gamma\converges$ since $f\Ip{A\linimpl B} g$ and
$f;\gamma\converges$.

$(\Leftarrow)$ Suppose $f;\gamma\converges$ where
$\gamma:(A\linimpl B)\rightarrow {\Siepinski}$. From the switching
condition we know that $\gamma$ can respond to the initial move in
${\Siepinski}$ only in $B$ or ${\Siepinski}$; to a move in $B$
only in $B$ or ${\Siepinski}$ and to a move in $A$ only in $A$ or
${\Siepinski}$. Moreover, whenever Player is to move in $A$ the
number of moves played in $B$ is odd, hence there is an unanswered
question in $B$ which must have been asked more recently than the
opening question in ${\Siepinski}$. By the stack discipline
$\gamma$ can in fact only respond in $A$ to a move in $A$.  Thus
if $\gamma\in\sigma_f$ where $f:M_A^O+M_B^P+M_{\cal
  O}^O\rightharpoonup  M_A^P+M_B^O+M_{\Siepinski}^P$ we can decompose $f$
as $f=g+h$ where: $g:M_A^O\rightharpoonup M_A^P,
h:M_B^P+M_{\Siepinski}^O\rightharpoonup M_B^O+M_{\Siepinski}^P$.
If we now define $x=\sigma_g, \beta=\sigma_h$ then:
\begin{itemize}
\item[(i)] $x\in{\tt Str}(A)$.
\item[(ii)] $\beta:B\rightarrow{\Siepinski}$.
\item[(iii)] $\forall h\in{\tt Str}(A\linimpl B). h;\gamma\approx x;h;\beta$.
\end{itemize}
Now
\[\begin{array}{lll}
f;\gamma\converges & \supset & x;f;\beta\converges \\
 &\supset_{\mbox{ by assumption}} & x;g;\beta\converges \\
 &\supset & g;\gamma\converges
\end{array}\]
as required.
\end{proof}

This argument can be understood in terms of Classical Linear Logic. If
we think of $A\linimpl {\Siepinski}$ as ``approximately $A^{\bot}$'', then
$$(A\linimpl B)\linimpl {\Siepinski} \approx (A\linimpl B)^{\bot}=
A\otimes B^{\bot}\approx A\otimes(B\linimpl {\Siepinski}).$$
To prove our final extensionality result, we will need an auxiliary
lemma.

\begin{lemma}[Separation of head occurrence]\label{Soho}

For all $\sigma:!A\rightarrow {\Siepinski}$, for some $\sigma':!A\otimes
A\rightarrow{\Siepinski}$:
\[ \begin{diagram}
\sigma\approx !A & \rTo^{{\tt con}_A} & !A\otimes !A&
\rTo^{{\tt id}_{!A}\otimes {\tt der}_A} & !A\otimes A& \rTo^{\sigma'}&
{\Siepinski}
\end{diagram} \]
\end{lemma}
\begin{proof}If $\sigma=\bot_{!A,{\Siepinski}}$ or
$\sigma=K_{!A}\top$, the result is trivial. If $\sigma$ responds
to the initial question $q$ with a move $(i,a)$ in $!A$ we define
$\sigma'$ by interpreting the index $i$ as a separate tensorial
factor rather than an index in $!A$. The only non-trivial point is
to show that $\sigma'\approx\sigma'$. If $q(i,a)sm\approx
q(i,a)s'm'$ where $q(i,a)s,q(i,a)s'\in\sigma'$, then any
permutation $\pi$ witnessing the equivalence must satisfy
$\pi(i)=i$. Let the response of $\sigma'$ to $m$ be $(j_1,a_1)$
and to $m'$  $(j_2,a_2)$. Since $\sigma\approx\sigma$ we must have
$q(i,a)sm(j_1,a_1)\approx_{!A\linimpl {\Siepinski}}
q(i,a)s'm'(j_2,a_2)$, and hence either $j_1=j_2=i$ or $j_1\neq
j_2\neq i$. In either cases, $q(i,a)sm(j_1,a_1)\approx_{!A\otimes
A\linimpl {\Siepinski}} q(i,a)s'm'(j_2,a_2)$, as required.
\end{proof}

\begin{lemma}[Bang Extensionality]\label{BE}

For all $x,y\in{\tt Str}(A).$
$$x\Ip A y \Longleftrightarrow !x\Ip{!A} !y$$
\end{lemma}
\begin{proof} $(\Leftarrow)$ If $!x\Ip{!A} !y$ and
$x;\alpha\converges$ then $!x;({\tt der}_A;\alpha)\converges$, so
$!y;({\tt der}_A;\alpha)\converges$, and hence
$y;\alpha\converges$ as required.

$(\Rightarrow)$ If $!x;\alpha\converges$, define $| \alpha|$ to be
the number of indices in $!A$ occurring in $!x \|\alpha$. We show
that, for all $\alpha:!A\rightarrow {\Siepinski}$ such that
$!x;\alpha\converges, !y;\alpha\converges$, by induction on
$|\alpha|$.  For the basis, note that $!x;\alpha\converges$ and $|\alpha|=0$
implies that $\alpha=K_{!A}\top$. For the inductive step, let
$|\alpha|=k+1$. By Lemma \ref{Soho}, for some $\beta:!A\otimes
A\rightarrow {\Siepinski}$,
$\alpha\approx{\tt con}_A;{\tt id}_{!A}\otimes{\tt der}_A;\beta$. For all
$z\in{\tt Str}(A). \;
!z;{\tt con}_A;{\tt id}_{!A}\otimes{\tt der}_A\approx !z\otimes z$, so
$!z\otimes z;\beta\approx !z;\alpha$.

Now define
\[\begin{diagram}
\gamma=!A & \rTo^{\sim}& !A\otimes I& \rTo^{{\tt id}_{!A}\otimes x} &
!A\otimes A& \rTo^{\beta} & {\Siepinski}
\end{diagram}\]
For all $z\in{\tt Str}(A),\; !z;\gamma\approx !z\otimes x;\beta$. In
particular, $!x;\gamma\approx !x\otimes x;\beta\approx
!x;\alpha\converges$.
Since $|\alpha|>0$, there is a first index $i_0$ in $!A$ used by
$\alpha$. By the definition of $\gamma$, $!x\| \gamma$ is
$!x \| \alpha$ with all moves at index $i_0$ deleted. Hence
$|\gamma|<|\alpha|$, and by induction hypothesis
$!y;\gamma\converges$.

Define $\delta:A\rightarrow {\Siepinski}$ by
\[\begin{diagram}
\delta=A & \rTo^{\sim} & I\otimes A &\rTo^{!y\otimes{\tt id}_A} &
!A\otimes A&\rTo^{\beta} & {\Siepinski}.
\end{diagram} \]
Then for all $z\in{\tt Str}(A).\;  z;\delta\approx!y\otimes
z;\beta$. In particular, $x;\delta\approx !y\otimes x;\beta\approx
!y;\gamma\converges$. By the assumption that $x\Ip A y,
y;\delta\converges$. This implies that $!y;\alpha\approx !y\otimes
y;\beta\converges$, as required.
\end{proof}

\begin{lemma}[Intuitionistic Function Extensionality]
$$\sigma \Ip{A\Rightarrow B}\tau \Longleftrightarrow
\forall x: 1\Rightarrow A,\; \beta :B\Rightarrow {\Siepinski}.
\; \beta\circ\sigma\circ x\converges \supset \beta\circ\tau\circ x\converges.$$
\end{lemma}

\begin{proof}
\[\begin{array}{lll}
\sigma\Ip{A\Rightarrow B}\tau &\Longleftrightarrow &\forall z\in{\tt
  Str}(!A). z;\sigma\Ip B z;\tau \\
 & & \mbox{Linear Function Extensionality} \\
& \Longleftrightarrow & \forall x\in{\tt Str}(A).
x^{\dag};\sigma\Ip B x^{\dag};\tau \\ && \mbox{Bang Lemma, } !I =I
\\ &\Longleftrightarrow & \forall x\in{\tt Str}(A).\;
x^{\dag};\sigma^{\dag}\Ip{!B} x^{\dag};\tau^{\dag} \\ &&\mbox{Bang
Extensionality, } {\tt der}_I={\tt id}_I \\ &\Longleftrightarrow&
\forall x\in{\tt Str}(A),\beta :!B\rightarrow {\Siepinski}. \;
x^{\dag};\sigma^{\dag};\beta\converges\supset
x^{\dag};\tau^{\dag};\beta\converges \\ &\Longleftrightarrow&
\forall x: 1\Rightarrow A,\beta : B\Rightarrow {\Siepinski}.\;
\beta\circ\sigma\circ x\converges \supset\beta\circ \tau\circ
x\converges. \;\;
\end{array}\]
\end{proof}

\begin{lemma}[Congruence Lemma]
\begin{itemize}
\item[(i)] $\sigma\Ip {A\Rightarrow B}\sigma' \wedge
  \tau\Ip {B\Rightarrow C} \tau' \supset
  \tau\circ\sigma\Ip {A\Rightarrow C} \tau'\circ\sigma'$
\item[(ii)] $\sigma\Ip {C\Rightarrow A}\sigma' \wedge
  \tau\Ip {C\Rightarrow B} \tau'\supset \lang\sigma,\tau\rang
  \Ip {C\Rightarrow A\with B}\lang\sigma',\tau'\rang.$
\item[(iii)] $\sigma\Ip {A\with B\Rightarrow C}\tau\supset
\Lambda(\sigma)\Ip {A\Rightarrow(B\Rightarrow C)}\Lambda(\tau).$
\end{itemize}
\end{lemma}

\begin{proof} {\it (i)}
\[\begin{array}{llll}
\beta\circ\tau\circ\sigma\circ x\converges
&\supset&\beta\circ\tau'\circ\sigma\circ x\converges & \tau\Ip
{B\Rightarrow C} \tau' \\ & \supset&\beta\circ\tau'\circ\sigma'
\circ x\converges & \sigma\Ip {A\Rightarrow B}\sigma'.
\end{array}\]

{\it (ii)}
For all $x : 1\Rightarrow C,\; \lang\sigma,\tau\rang\circ
x\approx\lang\sigma\circ  x,\tau\circ x\rang: I\rightarrow A\with B$;
and similarly, $\lang\sigma',\tau'\rang\circ x\approx\lang\sigma'\circ
x,\tau'\circ x\rang.$ By {\it (i)}, $\sigma\circ x\Ip A\sigma'\circ x$
and $\tau\circ x\Ip B\tau'\circ x$. The result now follows by Product
Extensionality.

{\it (iii)} Identifying morphisms with points of arrow types,
\[\begin{array}{llll}
\gamma\circ\Lambda(\sigma)\circ x\circ y\converges &\supset&
\gamma\circ\sigma\circ\lang x,y\rang\converges & \\
&\supset & \gamma\circ\tau\circ\lang x,y\rang\converges &
\sigma\Ip {A\with B\Rightarrow C}\tau \\
 &\supset & \gamma\circ\Lambda(\tau)\circ x\circ y\converges. &
\end{array}\]
\end{proof}

Finally we consider the relationship between the intrinsic
and intensional preorders.

\begin{lemma}\label{II}
\begin{itemize}
\item[(i)] If $\sigma\Subeq_A \tau$, then $\sigma\Ip A \tau$.
\item[(ii)] If $\sigma_o\subseteq\sigma_1\subseteq\dots $ is an
  increasing sequence, and for all $n$, $\sigma_n\Ip A\tau_n$, then
  $\bigcup_{n\in\omega}\sigma_n \Ip A \tau_n.$
\end{itemize}
\end{lemma}
\begin{proof} {\it (i)} By $\Subeq-$monotonicity of composition
(Proposition 2.9.3) if $\sigma\Subeq_A\tau$ and
$\sigma;\alpha=\top$ then $\top=\sigma;\alpha\Subeq_A \tau;\alpha$
and hence $\tau;\alpha=\top$.

{\it (ii)} By $\subseteq-$continuity of composition (Proposition
2.9.3), similarly to {\it (i)}.
\end{proof}

By Lemma \ref{II}, $\sigma\approx\tau$ implies $\sigma\simeq\tau$ where
$\simeq$ is the equivalence induced by the preorder $\Ip{}$. Thus each
$\simeq-$equivalence class is a union of
$\approx-$classes. Henceforth, when we write $[\sigma]$ we shall mean
the $\simeq-$equivalence class of $\sigma$.

We can define the notion of {\em strong chain} of $\simeq-$equivalence
classes, just as we did for $\approx-$classes:
a sequence
$$ (\dag) \;\;\; [\sigma_0]\Ip{}[\sigma_1]\Ip{}\dots $$
such that there are $(\sigma'_n\;\mid\; n\in\omega)$ with
$\sigma'_n\in[\sigma_n]$ and $\sigma'_n\subseteq \sigma'_{n+1}$ for
all $n\in\omega$.

\begin{lemma}\label{Strong}
Every strong $\Ip{}-$chain has a $\Ip{}-$least upper bound.
\end{lemma}
\begin{proof} Given a strong chain (\dag), take
$\bigsqcup_{n\in\omega} [\sigma_n]=[\sigma']$ where
$\sigma'=\bigcup_{n\in\omega}\sigma'_n$. For all $n,\;
\sigma_n\simeq\sigma'_n\subseteq \sigma'$, so by Lemma
\ref{II}{\it(i)}, $[\sigma']$ is un upper bound for
$([\sigma_n]\mid n\in\omega)$.

Finally, if $[\tau]$ is another upper bound, then for all $n$,
$\sigma'_n\Ip{}\tau$; so by Lemma \ref{II}{\it (ii)},
$\sigma'\Ip{}\tau$.
\end{proof}

\subsection{The Extensional Category }

We begin with some general considerations on quotients of rational cartesian
closed categories. Let $\cal C$ be a rational CCC. A {\em precongruence}
on $\cal C$ is a family
$\Ip{} =\{ \Ip{A,B} \mid A,B\in{\tt Obj}({\cal C})\}$ of
relations $\Ip{A,B}\subseteq {\cal C}(A,B)\times {\cal C}(A,B)$ satisfying
the following properties:
\begin{itemize}
\item[(r1)] each $\Ip{A,B}$ is a preorder
\item[(r2)] $f\Ip{A,B} f'$ and $g\Ip{B,C} g'$ implies
$g\circ f\Ip{A,C} g'\circ f'$
\item[(r3)] $f\Ip{C,A} f'$ and $g\Ip{C,B} g'$ implies
$\lang f,g\rang \Ip{C,A\times B}\lang f',g'\rang $
\item[(r4)] $f\Ip{A\times B,C} g$ implies $\Lambda(f)\Ip{A,B\Rightarrow C}
\Lambda(g)$
\item[(r5)] $\sqsubseteq_{A,B}\;  \subseteq \; \Ip{A,B}$
\item[(r6)] for all $f:A\times B\rightarrow B, g:C\rightarrow A,
h:B\rightarrow D$:
$$(\forall n\in\omega.\; h\circ f^{(n)}\circ g\Ip{C,D} k ) \supset
h\circ f^{\nabla}\circ g\Ip{C,D} k.$$
\end{itemize}

Given such a precongruence, we define a new category ${\cal C}/{\Ip{}}$ as
follows. The objects are the same as those of $\cal C$;
$${\cal C}/{\Ip{}} (A,B)=
({\cal C}(A,B)/\simeq_{A,B},\leq_{A,B}).$$
That is, a morphism in $C/\Ip{(A,B)}$
is a ${\simeq_{A,B}-}$equivalence class $[f]$, where $\simeq_{A,B}$ is the
equivalence relation induced by $\Ip{A,B}$. The partial ordering is then
the induced one:
$$[f]\leq_{A,B} [g] \Longleftrightarrow f\Ip{A,B} g.$$
Note that by (r5), $[\bot_{A,B}]$ is the least element with respect to this
partial order. By (r2)--(r4), composition, pairing and currying are well-defined
on ${\simeq-}$equivalence classes by
\[ \begin{array}{lll}
[g]\circ[f]& = & [g\circ f], \\
\lang [f],[g]\rang  & =& [\lang f,g\rang ], \\
\Lambda([f]) & =& [\Lambda(f)]\; .
\end{array} \]
 It is then immediate by (r5)
and the fact that $\cal C$ is a rational (and hence in particular a
ppo-enriched) CCC that ${\cal C}/{\Ip{}}$ is a ppo-enriched CCC.
It remains to verify rationality for ${\cal C}/{\Ip{}}$. By (r2) and (r5),
for any $f:A\times B\rightarrow B, g:C\rightarrow A,h:B\rightarrow D,$
the sequence  $([h\circ f^{(n)}\circ g]\mid n\in\omega)$ is a $\leq_{C,D}$-chain.
By (r5) and (r6), $[h\circ f^{\nabla}\circ g]$ is the $\leq_{C,D}-$least
upper bound of this chain. In particular~, taking $g={\tt id}_A$ and
$h={\tt id}_B$, $[f^{\nabla}]$ is the least upper bound of
$([f^{(n)}]\mid n\in\omega)$.

We record this result, which is a variant of [ADJ76], as

\begin{lemma}[Rational Quotient]\label{5.1}
If $\Ip{}$ is a precongruence on a rational CCC $\cal C$, then ${\cal C}/{\Ip{}}$
is a rational CCC.
\end{lemma}

Now we define a family $\Ip{}=\{\Ip{A\Rightarrow B} \mid A,B \in{\tt Obj}(K_{!}(
{\cal G})) \}.$

\begin{lemma}\label{5.2}
$\Ip{}$ is a precongruence on $K_{!}({\cal G})$.
\end{lemma}
\begin{proof} The fact that $\Ip{A\rightarrow B}$ is a preorder has
already been noted. (r2)--(r4) are the pre-congruence Lemma 4.1.8.
(r5) is Lemma 4.1.9{\it (i)}. Finally, we verify (r6). Let
$\sigma:A\with B\Rightarrow B, \tau: C\Rightarrow
A,\theta:B\Rightarrow D$. As already explained, since $\Subeq\;
\subseteq\; \Ip{}$, we work directly with $\simeq-$classes of
strategies, rather that $\simeq-$classes of $\approx-$classes of
strategies. Now $(\theta\circ\sigma^{(n)}\circ\tau \mid
n\in\omega)$ is a $\subseteq-$chain (using
$\subseteq-$monotonicity of composition), and we can apply Lemma
4.1.9{\it (ii)} to yeld (r6).
\end{proof}

Now we define $\E=K_{!}({\cal G})/{\Ip{}}$.

\begin{proposition}\label{5.3}
$\E$ is a rational CCC. Moreover, $\E$ is well-pointed in the
order-enriched sense:
$$f\leq_{A,b}g\Leftrightarrow\forall x:1\rightarrow A. \;\; f\circ x\leq_{A,B}
g\circ x\; .$$
\end{proposition}
\begin{proof} $\E$ is a rational CCC by Lemma \ref{5.1} and
\ref{5.2}. It is well-pointed by Intuitionistic Function
Extensionality (Lemma 4.1.7).
\end{proof}

Now we define the PCF model $\M(\E)$ with the same interpretation
of $\Nat$ as in $\M(K_{!}({\cal G}))$. The ground and first-order
constants of PCF are interpreted by the $\simeq-$equivalence
classes of their interpretations in $\M(K_{!}({\cal G}))$.

\begin{proposition}\label{5.4}
$\M(\E)$ is an order-extensional standard model of PCF.
\end{proposition}
\begin{proof} $\M(\E)$ is an order-extensional model of PCF by
Proposition 4.2.3. It is standard because $\M(K_{!}({\cal G}))$
is, and $\Ip{\Nat}=\Sx_{\Nat}$ .
\end{proof}

\subsection{An alternative view of $\E$}
We now briefly sketch another way of looking
at $\E$, which brings out its extensional character more clearly,
although technically it is no more than a presentational variant of the
above description.
Given a game $A$, define
$$\D(A)=(\{ [x]_{\simeq}\mid x\in{\tt Str}(A)\},\leq_A)$$
Then $\D(A)$ is a pointed poset. Given $\sigma:A\Rightarrow B$, define
$\D(\sigma):\D(A)\rightarrow \D(B)$ as the (monotone) function defined by:
$$\D(\sigma)([x])=[\sigma\circ x]$$
Write $f:A\rightarrow_{\E} B$ if $f:\D(A)\rightarrow \D(B)$ is a monotone
function such that $f=\D(\sigma)$ for some $\sigma: A\Rightarrow B$. In this
case we say that $f$ is {\em sequentially realised} by $\sigma$, and write
$\sigma\Vdash f$.

Note that there are order-isomorphisms

\begin{itemize}
\item $\D(I)\cong 1$
\item $\D(A\with B)\cong \D(A)\times\D(B)$
\item $\D(A\Rightarrow B)\cong \D(A)\Rightarrow_{\E}\D(B)$
\end{itemize}

Here $\D(A)\times\D(B)$ is the cartesian product of the posets $\D(A),\D(B)$,
with the pointwise order; while $\D(A)\Rightarrow_{\E} \D(B)$ is the set of
all functions $f:A\rightarrow_{\E} B$, again with the pointwise order.

Now note that, with respect to the representations of $\D(A\with B)$ as a
cartesian product and $\D(A\Rightarrow B)$ as a ``function space'', the
interpretations of composition, pairing and projections, and currying and
application in $\E$ are the usual set-theoretic operations on functions
{\em in extenso}. That is,
\[
\begin{array}{lcl}
\D(\tau\circ\sigma) & = & \D(\tau)\circ\D(\sigma) \\
\D(\lang \tau,\sigma\rang ) & = & \lang \D(\sigma),\D(\tau)\rang  \\
\D(\pi_1)  & = & \pi_1 \\
\D(\pi_2)  & = & \pi_2 \\
\D(\Lambda(\sigma))   & =& \Lambda(\D(\sigma)) \\
\D({\tt Ap})   & =& {\tt Ap}
\end{array} \]
\noindent where the operations on the right hand sides are defined as
in the category
of sets (or any concrete category of domains).

Thus an equivalent definition of $\E$ is as follows:
\[
\begin{array}{ll}
{\tt Objects} & \mbox{ as in } K_{!}(\G) \\
{\tt Arrows } & f : A\rightarrow_{\E} B\\
{\tt Composition } & \mbox{ function composition}
\end{array}\]

The r\^ole of the intensional structure, that is of the use of the game $A$
to represent the abstract space $\D(A)$, is to cut down the function spaces
to the sequentially realisable functions. Specifically, note the use of $A$
and $B$ in the definition of $\D(A)\Rightarrow_{\E} \D(B)$.

\subsection{Full Abstraction}

We recall that a model $\M$ is fully abstract for a language $\cal L$ if,
for all types $T$ and closed terms $M,N:T$
$$ \M\llbracket M\rrbracket \sqsubseteq \M\llbracket N\rrbracket
\Leftrightarrow M\sqsubseteq_{\tt obs} N \;\;
(\dag)$$
where
\[
\begin{array}{lcl}
 M\sqsubseteq_{\tt obs} N &\Leftrightarrow & \forall \mbox{  program
   context } C[.] \\
& & C[M]\Converges n \supset C[N]\Converges n
\end{array}
\]
Here a program context $C[.]$ is one for which $C[P]$ is a closed term of type
$N$  for any closed term $P:T$; and $\Converges$ is the operational
convergence relation. The left---to---right implication in $(\dag)$ is known as
{\em soundness} and the converse as {\em completeness}.
It is standard that soundness is a consequence of computational
adequacy \cite{CurienPL:catcsa};
thus by Proposition 2.10.1, standard models are sound. Also,
full abstraction for closed terms is easily seen to imply the corresponding
statement $(\dag)$ for open terms.

\begin{theorem}\label{1}
$\M(\E)$ is fully abstract for PCF.
\end{theorem}
\begin{proof} Firstly, $\M(\E)$ is a standard model by
Proposition~\ref{5.4}, and hence sound. We shall prove the
contrapositive of completeness. Suppose $M,N$ are closed terms of
PCF of type $T=T_1\Rightarrow\dots T_k\Rightarrow \Nat$ and
$$\M(\E)\llbracket M\rrbracket \nleq_{\llbracket T\rrbracket}
\M(\E)\llbracket N\rrbracket.$$ Let $[\sigma]=\M(\E)\llbracket
M\rrbracket, [\tau]= \M(\E)\llbracket N \rrbracket$. By
Intuitionistic Function Extensionality, for some $x_1\in{\tt
Str}(\llbracket T_1\rrbracket),\dots,x_k\in{\tt Str}(\llbracket
T_k\rrbracket),$ $$ \beta:!N\rightarrow\Siep,
\beta\circ\sigma\circ x_1\circ\dots\circ x_k \converges\mbox{ and
}\beta\circ\tau\circ x_1\circ\dots\circ x_k \diverges.$$ By
$\Subeq-$monotonicity of composition, this implies that
$\sigma\circ x_1\circ\dots\circ x_k\not\Subeq_{\Nat}\tau\circ
x_1\circ\dots \circ x_k$, and hence that $\sigma\circ
x_1\circ\dots\circ x_k=n$ for some $n\in\omega$, and $\tau\circ
x_1\circ\dots\circ x_k\neq n$. By $\subseteq-$continuity of
composition and the properties of the projections $p_k$ given in
the Approximation Lemma 3.5.1, for some $m\in\omega$, $\sigma\circ
p_m(x_1)\circ\dots\circ p_m(x_k)=n$, while by
$\subseteq-$monotonicity of composition, $\tau\circ
p_m(x_1)\circ\dots\circ p_m(x_k)\neq n$. By Lemma 3.6.3, there are
finite evaluation trees,and hence PCFc terms $P_1,\dots,P_k$ such
that $\llbracket P_i\rrbracket =[p_m(x_i)]$, $1\leq i\leq k$. This
means that $\llbracket M P_1\dots P_k\rrbracket = n$, while
$\llbracket N P_1\dots P_k\rrbracket \neq n$. By computational
adequacy, this implies that $M P_1\dots P_k\Converges n$ and
$\neg(N P_1\dots P_k\Converges n)$. By Lemma 3.1.1, each PCFc term
is observationally congruent to a PCF term. Hence there is a PCF
context $C[.]=[.]Q_1\dots Q_k$, where $Q_i\cong_{\tt obs} P_i$,
$1\leq i\leq k$, such that $C[M]\Converges n$ and
$\neg(C[N]\Converges n)$. This implies that $M\not\sqsubseteq_{\tt
obs} N$, as required.
\end{proof}

As an instructive consequence of this proof, we have:

\begin{corollary}[Context Lemma]
For all closed $M,N:T_1\Rightarrow \dots T_k\Rightarrow \Nat$,
\[
\begin{array}{lcl}
M\sqsubseteq_{\tt obs} N & \Leftrightarrow & \forall \mbox{ closed }
P_1:T_1,\dots,P_k:T_k \\
& & MP_1\dots P_k\Converges n\supset NP_1\dots P_k\Converges n
\end{array}
\]
\end{corollary}
\begin{proof} The left-to-right implication is obvious, by
considering applicative contexts $[.] P_1\dots P_k$. The converse
follows from the proof of the Full Abstraction Theorem, since if
$M\not\sqsubseteq_{\tt obs} N$, then $\llbracket M\rrbracket \nleq
\llbracket N\rrbracket$ by soundness, and then by the argument for
completeness this can be translated back into an applicative
context separating $M$ and $N$.
\end{proof}

The point of reproving this well-known result is that a semantic proof
falls out of the Full Abstraction Theorem. By contrast, Milner had to prove
the Context Lemma directly, as a necessary preliminary to his syntactic
construction of the fully abstract model. Moreover, the direct syntactic
proof, particularly for the $\lambda-$calculus formulation of PCF
\cite{CurienPL:catcsa}, is quite subtle. This gives some immediate
evidence of substance in our ``semantic analysis''.

\section{Universality}

The definability result we have achieved so far refers only to compact
strategies. Our aim in this section is to characterize precisely which
strategies are (extensionally) definable in PCF, and in fact to
construct a fully abstract model in which {\em all} strategies are
definable.

\subsection{ Recursive Strategies}

We shall develop effective versions of $\G$ and $\E$. Our treatment will be
very sketchy, as the details are lengthy and tedious, but quite routine.
We refer to standard texts such as \cite{SoareRI:recesd} for background.

We say that a game $A$ is {\em effectively given} if there is a surjective map
$e_A:\omega \rightarrow M_A$ with respect to which $\lambda_A$ (with some
coding of $\{P,O,Q,A\}$) and the characteristic functions of $P_A$ and
$\approx_A$ (with some coding of finite sequences) are tracked by recursive
functions. A strategy $\sigma$ on $A$ is then said to be {\em recursive}
if $\sigma$ is a recursively enumerable subset of $P_A$
(strictly speaking, if the set of codes of positions in $\sigma$ is r.e.).

\begin{lemma}\label{uni.1}
$\sigma=\sigma_f$ is recursive iff $f$ is tracked by a partial recursive
function. There are recursive functions taking an index for $\sigma$ to
one for $f$, and vice versa.
\end{lemma}
\begin{proof} The predicate $f(a)\simeq b\Leftrightarrow \exists s.
sab\in\sigma$ is clearly r.e. in $\sigma$, hence $f$ has an r.e.
graph and is partial recursive

Conversely, given $f$ define  a predicate $G(s,n)$ by:
\[ \begin{array}{lcl}
G(s,0) & = & s=\epsilon ,\\
G(s,n+1) &=& \exists a,b,t. \; s=tab\wedge s\in P_A\wedge G(t,n)\wedge
f(a)\simeq b.
\end{array} \]

Clearly $G$ is r.e. and hence so is $$\sigma={\tt graph}(f)=\{ s
\mid \exists n. G(s,n)\}. $$ These constructions are defined via
simple syntactic transformations and yield effective operations on
indices. \end{proof}

If $A$ and $B$ are effectively given, one can verify that the effective
structure lifts to $A\tensor B$, $A\linimpl B, A\with B$ and $!A$. Also,
$I$ and $\Nat$ are evidently effectively given. The most interesting point
which arises in verifying these assertions is that $\approx_{!A}$ is
recursive. This requires the observation that, in checking $s\approx_{!A} t$,
it suffices to consider permutations $\pi\in S(\omega)$ of {\em bounded}
(finite) support, where the bound is easily computed from $s$ and $t$.

Similarly, one can check that all operations on strategies defined in
Section 2 effectivize. For example, it is easily seen that the definition of
$\sigma;\tau$ in terms of sets of positions is r.e. in $\sigma$ and $\tau$;
or, we can give an algorithm for computing $EX(f,g)$. This algorithm
simply consists of applying $f$ and $g$ alternately starting from
whichever is applicable to the input, until an ``externally visible'' output
appears. Note that it is {\sl not} the case in general that unions of
$\subseteq-$chains of recursive strategies are recursive. For example every
strategy of type $N\linimpl N$ is a union of an increasing chain of finite
and hence recursive strategies. However, given a recursive
$\sigma: A\with B\Rightarrow B$, $\sigma^{\nabla}=\bigcup_{n\in\omega}
\sigma^{(n)}$ is recursive, since it can be enumerated uniformly effectively
in $n$ (``r.e. unions of r.e. sets are r.e.'').

Thus we can define a category $\G_{\tt rec}$ with objects effectively  given
games, and morphisms (partial equivalence classes of ) recursive strategies.
Also, the interpretations of PCF constants in $\M(K_{!}(\G))$ are clearly
recursive strategies.

\begin{proposition}
\begin{itemize}
\item[(i)] $\G_{\tt rec}$ is a Linear category
\item[(ii)] $K_{!}(\G_{\tt rec})$ is a rational cartesian closed category
\item[(iii)] $\M(K_{!}(\G_{\tt rec }))$ is a standard model of PCF
\end{itemize}
\end{proposition}

We can now consider the extensional quotient $\E_{\tt
rec}=K_{!}(\G_{\tt rec}) /{\Ip{}}$ where $\Ip{}$ is defined just
as for $K_{!}(\G)$, but of course with respect to recursive tests,
i.e. recursive strategies $A\linimpl \Siepinski$. All the results
of section 4 go through with respect to recursive tests.

\begin{proposition}
$\E_{\tt rec }$ is a well-pointed rational CCC. $\M(\E_{\tt rec })$ is a
fully abstract model of PCF.
\end{proposition}
\begin{proof} The result does require a little care, since the
Isomorphism Theorem 3.6.4 is not valid for $\M(\E_{\tt rec})$.
However, the Isomorphism Theorem was not used in the proof of the
Full Abstraction Theorem 4.3.1, but rather the finitary version
Lemma 3.6.3, which is valid in $\M(\E_{\tt rec})$.
\end{proof}

It is worth remarking that a nice feature of our definition of model in
terms of rationality rather than cpo-enrichment is that the recursive
version $\E_{\tt rec}$ is again a model in exactly the same sense as $\E$. By
contrast, in the cpo-enriched setting one must either modify the definition
of model explicitly (by only requiring completeness with respect to
r.e. chains), or implicitly by working inside some recursive realizability
universe.

\subsection{ Universal Terms }

The fact that $\M(K_{!}(\G_{\tt rec}))$ and $\M(\E_{\tt rec})$ are models
shows that all PCF terms denote recursive strategies, as we would expect.
Our aim now is to prove a converse; every recursive strategy is, up to
extensional equivalence, the denotation of a PCF term, and hence every
functional in the extensional model $\M(\E_{\tt rec})$ is definable in PCF.

More precisely our aim is to define, for each PCF type $T$, a
``universal term'' $U_T:\Nat \Rightarrow T$, such that

$$\E\llbracket U_T \lceil\sigma\rceil\rrbracket = [\sigma]$$
for each recursive $\sigma$. These universal terms will
work by simulating the evaluation tree corresponding to $\sigma$.

Firstly, we recall some notations from recursion theory. We fix an acceptable
numbering of the partial recursive functions \cite{SoareRI:recesd}
such that $\phi_n$ is the $n$'th
partial recursive function and $W_n$ is the $n$'th r.e. set. We also fix
a recursive pairing function
$\lang -,-\rang :\omega\times\omega\rightarrow\omega$
and a recursive coding of finite sequences.

A recursive strategy $\sigma$ is regarded as being given by a code
(natural number) $\lceil \sigma\rceil$. By virtue of Lemma 5.1.1
we use such a code indifferently as determining $\sigma$ by
$$\sigma=\sigma_f \mbox{ where } f=\phi_{\lceil\sigma\rceil}$$ or
$$W_{\lceil\sigma\rceil}=\{\lceil s\rceil \mid s\in\sigma\}$$ The
following lemma is a recursive refinement of the Decomposition
Lemma, and assumes the notations of Section 3.4.

\newcommand{\casodue}[5] { \[ {#1} = \left\{  \begin{array}{ll}
                         {#2} & {{#3}} \\
                         {#4} & {{#5}}
                         \end{array}
                      \right. \] }

\newcommand{\casotre}[7] { \[ {#1} = \left\{  \begin{array}{ll}
                         {#2} & {{#3}} \\
                         {#4} & {{#5}} \\
                         {#6} & {{#7}}
                         \end{array}
                      \right. \] }

\begin{lemma}[Decomposition Lemma (Recursive Version)]
For each PCF type $T$ there are partial recursive functions
$$D_T,H_T:\omega\rightharpoonup\omega\mbox{ and } B_T:\omega\times\omega
\rightharpoonup\omega$$
such that, if $\sigma$ is a recursive strategy on $T$

\casotre{D_T\lceil\sigma\rceil}{\mbox{undefined,}}{\sigma=\bot_{\tilde{T}}}
{\lang 2,n\rang ,}{\sigma=K_{\tilde{T}}n}{\lang 3,i\rang ,}{R(\sigma)}
\casodue{H_T\lceil\sigma\rceil}{\lang \lceil\sigma_1\rceil,\dots,\lceil\sigma_{l_i}
\rceil\rang ,}{R(\sigma)}{\mbox{undefined},}{\mbox{otherwise}}
\casodue{B_T(\lceil\sigma\rceil,n)}{\lceil\tau_n\rceil,}{R(\sigma)}
{\mbox{undefined},}{\mbox{otherwise}}
where $R(\sigma)$ stands for
$$\Phi(\sigma)=(3,i,\sigma_1,\dots,\sigma_{l_i},
(\tau_n\mid n\in\omega)).$$
\end{lemma}
\begin{proof} $\D_T\lceil\sigma\rceil$ is computed by applying
$\phi_{\lceil\sigma \rceil}$ to the (code of) the initial
question. The extraction of $\tau_n$ from $\sigma$,
$\tau_n=\{*_1s\mid  *_1*_2ns\in\sigma \}$, is obviously r.e. in
$\sigma$, uniformly effectively in $n$. Hence we obtain an r.e.
predicate $s\in B_T(\lceil\sigma\rceil,n)$, and by an application
of the \textsf{S-m-n} theorem we obtain the index for
``$B_T\lceil\sigma\rceil n =\lceil\tau_n\rceil$''.

Similarly the extraction of $\sigma'$ from $\sigma$ is r.e. in
$\sigma$, and that of $\sigma''$ for $\sigma'$ is r.e. in
$\sigma'$; while $\sigma_1, \dots,\sigma_{l_i}$ are obtained from
$\sigma''$ by composition, dereliction and projection, which are
computable operations by Proposition 5.1.2. Hence applying the
\textsf{S-m-n} theorem again we obtain the codes for $\sigma_1,
\dots,\sigma_{l_i}$.
\end{proof}

Given a PCF type $T$, we define the subtypes of $T$ to be the PCF
types occurring as subformulas of $T$, e.g. $(N\Rightarrow N)$ and
$N$ are subtypes of $(N\Rightarrow N)\Rightarrow N$. Let
$S_1,\dots,S_q$ be a listing of all the (finitely many) subtypes
of $T$, where we write $$S_i=S_{i,1}\Rightarrow\dots
S_{i,l_i}\Rightarrow N$$ To aid the presentation, we will use an
abstract datatype ${\tt Ctxt}_T$ of ``$T$-contexts'', which we
will show later how to implement in PCF. We will make essential
use of the fact that, while contexts can grow to arbritary size in
the recursive unfolding of an evaluation tree of type $T$, the
types occurring in the context can only be subtypes of $T$.

${\tt Ctxt}_T$ comes with the following operations:
\begin{itemize}
\item ${\tt  emptycontext}_T:{\tt Ctxt}_T$
\item ${\tt get}_S: N\Rightarrow{\tt Ctxt}_T\Rightarrow S$ for each subtype
$S$ of $T$
\item ${\tt extend}_{S_i}:{\tt Ctxt}_T\Rightarrow S_{i,1}\Rightarrow\dots
S_{i,l_i}\Rightarrow{\tt Ctxt}_T$ for each subtype $S_i$ of $T$
\item ${\tt map}_T: N\Rightarrow{\tt Ctxt}_T\Rightarrow N$.
\end{itemize}

If $\Gamma=x_1:U_1,\dots,x_n:U_n$, then $\Gamma^i$ is the subsequence of all
entries of type $S_i$, $1\leq i\leq q$ and $\Gamma_j=x_j:U_j$

The idea is that, if $\Gamma$ is an ``abstract context'',
\begin{itemize}
\item ${\tt extend}_{S_i} \Gamma x_1^{S_{i,1}}\dots x_{l_i}^{S_{i,l_i}}=
\Gamma,x_1:S_{i,1},\dots,x_{l_i}:S_{i,l_i}$
\item ${\tt map}_T \;i \;\Gamma =\lang i_1,i_2\rang $ where $\Gamma_i=x:S_{i_1}=
\Gamma_{i_2}^{i_1}$
\item ${\tt get}_{S_i} j \Gamma = \Gamma_j^i$.
\end{itemize}

Now we use the standard fact that every partial recursive function
$\phi:\omega\rightharpoonup\omega$ can be represented by a closed PCF term
$M: N\Rightarrow N$ in the sense that, for all $n\in\omega$
$$M n\Downarrow m\Leftrightarrow \phi n\simeq m .$$
This obviously implies that partial recursive functions of two arguments can
be represented by closed terms of type $ N\Rightarrow N\Rightarrow N$.
Specifically, we fix terms ${\bf  D_T,H_T}:N \Rightarrow N$ and
${\bf B_T}: N\Rightarrow N\Rightarrow N$ which represent $D_T,H_T$ and $B_T$
respectively.

Now we define a family of functions $$F_S:{\tt Ctxt}_T\Rightarrow
N\Rightarrow S$$ for each subtype $S=U_1\Rightarrow\dots
U_k\Rightarrow N$ of $T$, by the following mutual recursion:

\[ \begin{array}{l}
F_S=  \lambda k^{N}.\lambda \Gamma^{{\rm Ctxt}}_T.\lambda x_1^{U_1}\dots
\lambda x_k^{U_k}\\
\;\; {\tt let }\lang k_1,k_2\rang = D_T k \;{\tt in } \\
\;\;\;\; {\tt if }\; k_1=2 \;\; {\tt then }\; k_2 \; {\tt else}\\
\;\;\;\;\;\;\;\;\; {\tt if }\; k_1=3 \; {\tt then }\\
 \;\;\;\;\;\;\;\;\;\;\; {\tt let }\;  \Delta ={\tt extend}_S \;\Gamma x_1\dots
x_k \;  {\tt in}\\
\;\;\;\;\;\;\;\;\;\;\;\; {\tt let } \; \lang i_1,i_2\rang  = {\tt map}_T \;k_2\;
\Delta \;  {\tt in} \\
\;\;\;\;\;\;\;\;\;\;\;\;\;{\tt case }\; i_1 \;\; {\tt of } \\
\;\;\;\;\;\;\;\;\; \;\;\;\;\; 1 :  \dots  \\
\;\;\;\;\;\;\;\;\;\;\;\;\; \; \vdots  \\
\;\;\;\;\;\;\;\;\;\;\;\;\;  \; i: \; {\tt let }\;\lang k_1,\dots,k_{l_i}\rang
=\; H_S k \;\;\; {\tt in} \\
\;\;\;\;\;\;\;\;\;\;\;\;\;\;\;\;\;\;\;  \;{\tt let }\;n= ({\tt get}_{S_i}
i_2\Delta)(F_{S_{i,1}}  k_1\Delta) \dots\;\;(F_{S_{i,l_i}}k_{l_i}\Delta) \; \\
\;\;\;\;\;\;\;\;\;\;\;\;\;\;\;\;\;\;\;\;\;\;\;\;\;\; \;\;\;\;\;\;\;\;
\; {\tt in} \;\;F_N(B_S k n) \Delta\\
\;\;\;\;\;\;\;\;\;\;\;\;\;\;\;i+1 : \; \dots \\
\;\;\;\;\;\;\;\;\;\;\;\;\; \; \vdots  \\
\;\;\;\;\;\;\;\;\;\;\;\;\;\;\; q :  \; \dots \\
\;\;\;\;\;\;\;\;\;\;\;\; {\tt otherwise}: \; \Omega\\
\;\;\;\;\;\;\;\;\;\;\; {\tt endcase} \\
\;\;\;\;\;\;\;\;\;\;{\tt else} \; \Omega
\end{array} \]

These functions have been defined using some ``syntactic sugar''.
Standard techniques can be used to transform these definitions
into PCF syntax. In particular Beki\u{c}'s rule
\cite{WinskelG:fspl} can be used to transform a finite system of
simultaneous recursion equations into iterated applications of the
${\bf Y}$ combinator. The universal term $U_T$ can then be defined
by $$U_T=F_T \;\;{\tt emptycontext}_T.$$

It remains to be shown how ${\tt Ctxt}_T$ can be implemented in PCF. To do
this, we assume two lower-level data-type abstractions, namely product
types $T\times U$ with pairing and projections, and list types ${\tt list}(T)$
for each PCF type $T$, with the usual operations:
\begin{itemize}
\item ${\tt empty}_T:{\tt list}(T)\Rightarrow N$
\item ${\tt cons}_T:T\Rightarrow{\tt list}(T)\Rightarrow{\tt list}(T)$
\item ${\tt hd}_T:{\tt list}(T)\Rightarrow T$
\item ${\tt tl}_T:{\tt list}(T)\Rightarrow{\tt list}(T)$
\item ${\tt nil}_T:{\tt list}_T$
\end{itemize}

We write $l\converges i$ for the $i$'th component of a list.

We represent an abstract context $\Gamma$ by the $q+1-$tuple
$(l_1,\dots,l_q,{\tt mlist})$ where $l_i:{\tt list}(S_i), 1\leq i\leq q$ and
${\tt mlist}:{\tt list}(N)$. The idea is that $l_i=\Gamma^i$, while
\[ {\tt mlist}\converges i= \lang i_1,i_2\rang ={\tt map}_T \, i \,
\Gamma . \]

It is straightforward to implement the operations on contexts in terms of this representation.

\begin{itemize}
\item ${\tt emptycontext}_T=([],\dots,[],[])$
\item ${\tt map}_T i(l_1,\dots,l_q,{\tt mlist})={\tt mlist}\converges i$
\item ${\tt get}_{S_i} j(l_1,\dots,l_q,{\tt mlist})=l_i\converges j$
\item ${\tt extend}_{S_i}(l_1,\dots,l_q,{\tt mlist})x_1\cdots
  x_{l_i}= L$

where
$$L={\tt extend1}_{S_{i,l_i}}(\cdots({\tt extend1}_{S_{i,2}}({\tt
  extend1}_{S_{i,1}} (l_1,\dots,l_q,{\tt
  mlist})x_1)x_2)\cdots)x_{l_i}$$
and ${\tt extend1}_{S_{i,j}}(l_1,\dots,l_q,{\tt mlist})x$ equals
$$(l_1,\dots,l_j++[x],\dots,l_q,{\tt mlist}++[\lang j,{\tt
  length}_{S_j}(l_j)+1\rang]) $$
where $-++-$ is list concatenation.
\end{itemize}
Finally, we show how to represent lists and products in PCF. We represent
lists by
$${\tt List}(T)=( N\Rightarrow T)\times N$$
where e.g.
\begin{itemize}
\item ${\tt cons}_T=$
\[\begin{array}{llll}
\lambda x^T.  \lambda l:{\tt List}(T) & & & \\
\;\;{\tt let} \;(f,n)=l & {\tt in}\; (g,n+1)  & &\\
& {\tt where} &  & \\
& g=\lambda i^{N}. &{\tt if}\; i=0 &{\tt then}\; x \\
& & & {\tt else} \;\;f(i-1)
\end{array}\]
\item ${\tt empty}_T(f,n)= n=0.$
\end{itemize}
A function taking an argument of product type
$$T\times U\Rightarrow V$$
can be replaced by its curried version
$$T\Rightarrow U\Rightarrow V$$
while a function returning a product type can be replaced by the two
component functions.

This completes our description of the universal term $U_T$.

\newcommand{\R}{{\; \cal R\; }}
For each PCF type $T$, we define a relation $M \R_T a$ between
closed PCF terms of type $T$ and strategies $a\in{\tt Str}(T)$ by
$$M\R_T a\iff \llbracket M\rrbracket \simeq a .$$
This is extended to sequences $\tilde{M}\R_{\tilde{T}}\tilde{a}$ in
the evident fashion.

We fix a type $T$ with subtypes $S_1,\dots, S_q$ as in the previous
discussion.

\begin{lemma}\label{522}
Let $\tilde{T}\Rightarrow S$ be a PCF type-in-context and
$\sigma\in{\tt Str}(\tilde{T}\Rightarrow S)$ a compact strategy, where
$\tilde{T},S$ are subtypes of $T$. Let $\Gamma$ be a closed expression
of type ${\tt Ctxt}_T$ (which we will regard as a sequence of
closed terms), and $\tilde{a}$ a sequence of strategies.
Then
$$\Gamma\R_{\tilde{T}}\tilde{a} \Rightarrow (F_S\lceil
\sigma\rceil\Gamma)\R_S(\sigma\tilde{a}). $$
\end{lemma}

\begin{proof} By induction on the height of the finite evaluation
tree corresponding to $\sigma$ under Theorem \ref{isotheo} , and
by cases on the Decomposition Lemma for $\sigma$. The cases for
$\sigma=\Lambda_{\tilde{S}}(\bot_{\tilde{T},\tilde{S}}) $ and
$\sigma= \Lambda_{\tilde{S}}( {\bf K}_{\tilde{T},\tilde{S}}n)$ are
clear.

Suppose $$\sigma\approx {\bf
  C}_i(\sigma_1,\dots,\sigma_{l_i},(\tau_n\mid n\in\omega)).$$
By Intuitionistic Function Extensionality Lemma, it suffices to show
that, for all closed $\tilde{M}$ and
strategies $\tilde{b}$ such that $\tilde{M}\R_{\tilde{S}}\tilde{b}$
$$F_S\lceil \sigma\rceil
\Gamma\tilde{M}\R_N\sigma\tilde{a}\tilde{b}.$$
Let $\Delta={\tt extend}_S\Gamma\tilde{M},\;
\tilde{c}=\tilde{a},\tilde{b}$. Then
$\Delta\R_{\tilde{T},\tilde{S}}\tilde{c}$,
so by induction hypothesis,
$$F_{S_{i,j}}\lceil\sigma_j\rceil\Delta\R_{S_{i,j}}\sigma_j\tilde{c},\;\;\;
1\leq j\leq l_i$$
Hence if we define
\[\begin{array}{lll}
M &=&
\Delta_i(F_{S_{i,1}}\lceil\sigma_1\rceil\Delta)\cdots(F_{S_{i,l_i}}\lceil
\sigma_{l_i}\rceil\Delta)\\
&=& \Delta^{i_1}_{i_2}(F_{S_{i,1}}\lceil\sigma_1\rceil\Delta)\cdots
(F_{S_{i,l_i}}\lceil\sigma_{l_i}\rceil\Delta)\\
\end{array}\]
where $\lang i_1,i_2\rang={\tt map} \, i \, \Delta$, then $M\R_N
c_i(\sigma_1\tilde{c}) \cdots(\sigma_{l_i}\tilde{c}).$ Thus if
$c_i(\sigma_1\tilde{c}) \cdots(\sigma_{l_i}\tilde{c})=\bot_N$, then
$\llbracket M\rrbracket =\bot_n$, while if $ c_i(\sigma_1\tilde{c})
\cdots(\sigma_{l_i}\tilde{c}) = n$ then $\llbracket M\rrbracket =n$.

In the former case, $$\llbracket
F_S\lceil\sigma\rceil\Gamma\tilde{M}\rrbracket\simeq\bot_N
\simeq\sigma\tilde{c}.$$
In the latter case,
\[\begin{array}{lll}
\llbracket F_S\lceil\sigma\rceil\Gamma\tilde{M}\rrbracket
&\simeq&\llbracket F_N(B\lceil\sigma\rceil n)\Delta\rrbracket \\
&\simeq&
\llbracket F_N\lceil\tau_n\rceil\Delta\rrbracket,
\end{array}\]
while $\sigma\tilde{c}\simeq\tau_n\tilde{c}$, and by induction
hypothesis $F_N\lceil\tau_n\rceil\Delta \R_N\tau_n\tilde{c}.\;\;$
\end{proof}

Now we define a family of relations $(\preceq_k\mid k\in\omega)$, where
$\preceq_k\subseteq\omega\times\omega$, inductively as follows:
\[\begin{array}{{cll}}
\preceq_0 &=& \omega\times\omega \\
n\preceq_{k+1} m& \iff& (D_n=\lang 2,p\rang\Rightarrow D_m=\lang
2,p\rang)\\
&\wedge & (D_n=\lang 3,i\rang\Rightarrow D_m=\lang 3,i\rang\\
&\wedge & [H n=\lang k_1,\dots,k_{l_i}\rang\Rightarrow \\
&&       H m =\lang k'_1,\dots,k'_{l_i}\rang\wedge \bigwedge_{j=1}^{l_i}
k_j\preceq_k k'_j ]\\
&\wedge &  \forall p: 0\preceq p\preceq k.\;\;  B n p \preceq_k B m p).
\end{array}\]

We can read $n\preceq_k m$ as: the stategy coded by $m$ simulates the
strategy coded by $n$ for all behaviours of size $\leq k$.

We write $$n\preceq m\iff \forall k\in\omega. n\preceq_k m.$$

\begin{lemma}\label{due}
For all PCF types $T$, $\sigma\in{\tt Str}(T), k\in\omega:$
\begin{itemize}
\item[(i)] $p_k(\sigma)\preceq\sigma$.
\item[(ii)] $\sigma\preceq_k p_k(\sigma)$
\end{itemize}
\end{lemma}

\begin{lemma}\label{tre}
With $S,\Gamma,\tilde{M}$ as in Lemma \ref{522}, and $\sigma$ any strategy
in ${\tt Str}(S)$:
$$\llbracket F_S\lceil\sigma\rceil\Gamma\tilde{M}\rrbracket
=n\iff\exists k\in\omega.\;\llbracket
F_S\lceil p_k(\sigma)\rceil\Gamma\tilde{M}\rrbracket =n$$
\end{lemma}

\begin{proof} $(\Leftarrow)$ By Lemma \ref{due}{\it (i)}.

$(\Rightarrow)$ By Lemma \ref{due}{\it (ii)}, using continuity,
and hence the fact that only finitely many calls to $D,H$ and $B$
are made in evaluating $F_S\lceil\sigma\rceil\Gamma\tilde{M}$.
(This can be made precise using Berry's Syntactic Approximation
Lemma for PCF \cite{BerryG:fulasl}).
\end{proof}

\begin{theorem}[Universality Theorem]
For all PCF types $T$ and recursive strategies $\sigma\in{\tt Str}(T)$ with
$n=\lceil\sigma\rceil$,
$$\M(K_{!}(\G))\llbracket U_T n\rrbracket\simeq_T\sigma.$$
Thus every functional in $\M(\E_{\tt rec})$ (equivalently, every functional
in $\M(\E)$ realised by a recursive strategy) is definable in PCF.
\end{theorem}

\begin{proof} For all closed $\tilde{M}:\tilde{T}$.
\[\begin{array}{lll}
\llbracket U_T\lceil\sigma\rceil\tilde{M}\rrbracket = n &\iff&\exists
k\in\omega.\; \llbracket U_T\lceil p_k(\sigma)\rceil
\tilde{M}\rrbracket= n\\
&& \mbox{by Lemma \ref{tre}} \\
&\iff& \exists k\in\omega.\; p_k(\sigma)\llbracket \tilde{M}\rrbracket
=n\\
&& \mbox{by Lemma \ref{522}} \\
&\iff&\sigma \llbracket\tilde{M}\rrbracket = n \\
&&\mbox{by the Approximation Lemma for strategies.}
\end{array}\]

\noindent By the Intuitionistic Function Extensionality Lemma this
shows that $\llbracket
U_T\lceil\sigma\rceil\rrbracket\simeq\sigma.\; $
\end{proof}

In the case of cpo-enriched models, an important result due to Milner
is that the fully-abstract order extensional model is unique up to isomorphism.
For rational models, the situation is not quite so rigid. For example, both
$\M(\E)$ and $\M(\E_{\tt rec})$ are fully abstract, but $\M(\E_{\tt rec})$ is
properly contained in $\M(\E)$. To see this, note that all monotonic functions
of type $N\Rightarrow N$ are sequentially realised and hence live in
$\M(\E)$, while only the recursive ones live in $\M(\E_{\tt rec})$.
We can, however, give a very satisfactory account of the canonicity of
$\M(\E_{\tt rec})$. We define a category ${\tt FAMOD(PCF)}$ with objects
the fully abstract (rational) models of PCF.
A homomorphism $F:\M({\cal C})\rightarrow {\cal M}({\cal D})$ is a functor
from the full cartesian closed sub category of $\cal C$ generated by
the interpretation
of $N$ in $\M({\cal C})$ to the corresponding subcategory of $\cal D$.
$F$ is additionally required to be a rational CCC functor, and to preserve the
interpretation of $N$ and of the PCF ground and first-order constants.

\begin{theorem}[Extensional Initiality Theorem]

$\M(\E_{\tt rec})$ is initial in ${\tt FAMOD(PCF)}$.
\end{theorem}
\newcommand{\N}{{\cal N}}
\begin{proof} Let $\N$ be any fully abstract model. By the
Universality Theorem, there is only one possible definition of
$F:\M(\E_{\tt rec})\rightarrow {\cal N}$, given by $$F(\M(\E_{\tt
rec})\llbracket M\rrbracket)={\cal N}\llbracket M\rrbracket$$ for
all closed terms $M$ of PCF. Since $\M(\E_{\tt rec})$ and $\cal N$
are both fully abstract,
\[\begin{array}{lrll}
& \M(\E_{\tt rec})\llbracket M\rrbracket & \leq& \M(\E_{\tt rec})\llbracket N
\rrbracket \\
\Leftrightarrow & M & \sqsubseteq_{\tt obs} & N \\
\Leftrightarrow & {\cal N}\llbracket M\rrbracket & \leq  &
{\cal N}\llbracket N\rrbracket
\end{array}\]
so this map is well-defined, and preserves and reflects order. It
is a homomorphism by the compositional definition of the semantic
function.
\end{proof}






\begin{thebibliography}{Cur92b}

\bibitem[AP90]{AP90}
M. Abadi and G. Plotkin.
\newblock A {\sc PER} model of polymorphism and recursive types.
\newblock  In {\em Proceedings, Fifth Annual IEEE Symposium on Logic in Computer
  Science}, IEEE Computer Society Press, 1990.

\bibitem[Abr94]{AbramskyS:prop}
Samson Abramsky.
\newblock Proofs as processes.
\newblock {\em Theoretical Computer Science}, 135:5--9, 1994.

\bibitem[AJ93]{AbramskyS:gamexp}
S.~Abramsky and R.~Jagadeesan.
\newblock Game semantics for exponentials.
\newblock Announcement on the types mailing list, 1993.

\bibitem[AJ94a]{AbramskyS:gamfcm}
Samson Abramsky and Radha Jagadeesan.
\newblock Games and full completeness for multiplicative linear logic.
\newblock {\em Journal of Symbolic Logic}, 59(2):543 -- 574, June 1994.
\newblock Also appeared as Technical Report 92/24 of the Department of
  Computing, Imperial College of Science, Technology and Medicine.

\bibitem[AJ94b]{AbramskyS:domt}
Samson Abramsky and Achim Jung.
\newblock Domain theory.
\newblock In S.~Abramsky, D.~Gabbay, and T.~S.~E. Maibaum, editors, {\em
  Handbook of Logic in Computer Science Volume 3}, pages 1--168. Oxford
  University Press, 1994.

\bibitem[AJM94]{AbramskyS:fulap}
Samson Abramsky, Radha Jagadeesan, and Pasquale Malacaria.
\newblock Full abstraction for {PCF} (extended abstract).
\newblock In Masami Hagiya and John~C. Mitchell, editors, {\em Theoretical
  Aspects of Computer Software. International Symposium TACS'94}, number 789 in
  Lecture Notes in Computer Science, pages 1--15, Sendai, Japan, April 1994.
  Springer-Verlag.

\bibitem[AM95]{McCuskerGA:gamfal}
Samson Abramsky and Guy McCusker.
\newblock Games and full abstraction for the lazy $\lambda$-calculus.
\newblock In {\em Proceedings, Tenth Annual IEEE Symposium on Logic in Computer
  Science}, pages 234--243. IEEE Computer Society Press, 1995.

\bibitem[AM99]{AM97}
Samson Abramsky and Guy McCusker.
\newblock Game Semantics.
\newblock In Ulrich Berger and Helmut Schwichtenberg, editors,
{\em Computational Logic}, pages 1--56, Springer-Verlag 1999.

\bibitem[Abr97]{Abr97}
Samson Abramsky.
\newblock Games in the Semantics of Programming Languages.
\newblock In {\em Proceedings of the 1997 Amsterdam Colloquium}.

\bibitem[Abr00]{Abr00}
Samson Abramsky.
\newblock Axioms for Definability and Full Completeness.
\newblock In G. Plotkin, C. Stirling and M. Tofte, editors,
{\em Proof, Language and Interaction: Essays in honour of
Robin Milner}, pages 55--75, MIT Press 2000.

\bibitem[Bar84]{BarendregtHP:lamcss}
Henk~P. Barendregt.
\newblock {\em The Lambda Calculus: Its Syntax and Semantics}.
\newblock North-Holland, revised edition, 1984.

\bibitem[Ber79]{BerryG:modcom}
G{\'e}rard Berry.
\newblock {\em Modeles completement adequats et stable de lambda-calculs}.
\newblock PhD thesis, Universite Paris VII, 1979.

\bibitem[BC82]{BerryG:seqacd}
G{\'e}rard Berry and P.-L. Curien.
\newblock Sequential algorithms on concrete data structures.
\newblock {\em Theoretical Computer Science}, 20:265--321, 1982.

\bibitem[BCL85]{BerryG:fulasl}
G{\'e}rard Berry, P.-L. Curien, and Jean-Jacques L\'evy.
\newblock Full abstraction for sequential languages: the state of the art.
\newblock In M.~Nivat and John Reynolds, editors, {\em Algebraic Semantics},
  pages 89--132. Cambridge University Press, 1985.

\bibitem[BE91]{EhrhardT:extemb}
A.~Bucciarelli and T.~Ehrhard.
\newblock Extensional embedding of a strongly stable model of {PCF}.
\newblock In {\em Proc. ICALP}, number 510 in LNCS, pages 35--46. Springer,
  1991.

\bibitem[CF92]{FelleisenM:obsseq}
R.~Cartwright and M.~Felleisen.
\newblock Observable sequentiality and full abstraction.
\newblock In {\em Proc. POPL}, pages 328--342. ACM {P}ress, 1992.

\bibitem[Cro94]{CroleRL:catt}
Roy Crole.
\newblock {\em Categories for Types}.
\newblock Cambridge University Press, 1994.

\bibitem[Cur92a]{CurienPL:obssac}
Pierre-Louis Curien.
\newblock Observable sequential algorithms on concrete data structures.
\newblock In {\em Proceedings, {S}eventh {A}nnual {IEEE} {S}ymposium on {L}ogic
  in {C}omputer {S}cience}, pages 432--443. IEEE {C}omputer {S}cience {P}ress,
  1992.

\bibitem[Cur92b]{CurienPL:seqful}
Pierre-Louis Curien.
\newblock Sequentiality and full abstraction.
\newblock In P.~T.~Johnstone M.~Fourman and A.~M. Pitts, editors, {\em
  Applications of Categories in Computer Science}, pages 66--94. Cambridge
  University Press, 1992.

\bibitem[Cur93]{CurienPL:catcsa}
Pierre-Louis Curien.
\newblock {\em Categorical Combinators, Sequential Algorithms and
  Fun\-ctio\-nal Programming}.
\newblock Progress in Theoretical Computer Science. Birkhauser, 1993.

\bibitem[DP90]{DaveyBA:Intlor}
B.~A. Davey and H.~A. Priestley.
\newblock {\em Introduction to Lattices and Order}.
\newblock Cambridge University Press, 1990.

\bibitem[Ehr]{EhrhardT:prosas}
Thomas Ehrhard.
\newblock Projecting sequential algorithms on the strongly stable functions.
\newblock {\em {A}nnals of {P}ure and {A}pplied {L}ogic}
77(3):201--244, 1996.

\bibitem[Gan93]{GandyRO:diabso}
R.~O. Gandy.
\newblock Dialogues, {B}lass games and sequentiality for objects of finite
  type.
\newblock Unpublished manuscript, 1993.

\bibitem[GM00]{GM00}
Dan R. Ghica and Guy McCusker.
\newblock Reasoning about Idealized Algol Using Regular Languages.
\newblock In {\em Proceedings of ICALP 2000}, Springer Lecture Notes in
Computer Science, 2000.

\bibitem[Gir88]{GirardJY:geoi2d}
Jean-Yves Girard.
\newblock Geometry of interaction 2: Deadlock-free algorithms.
\newblock In P.~Martin-L\"{o}f and G.~Mints, editors, {\em International
  Conference on Computer Logic, {COLOG} 88}, pages 76--93. Springer-Verlag,
  1988.
\newblock Lecture Notes in Computer Science 417.

\bibitem[Gir89a]{GirardJY:geoi1i}
Jean-Yves Girard.
\newblock Geometry of interaction 1: Interpretation of {S}ystem {F}.
\newblock In R.~Ferro et~al., editor, {\em Logic Colloquium 88}, pages
  221--260. North Holland, 1989.

\bibitem[Gir89b]{GirardJY:towgi}
Jean-Yves Girard.
\newblock Towards a geometry of interaction.
\newblock In J.~W. Gray and Andre Scedrov, editors, {\em Categories in Computer
  Science and Logic}, volume~92 of {\em Contemporary Mathematics}, pages
  69--108. American Mathematical Society, 1989.

\bibitem[Hoa85]{HoareCAR:comsp}
C.~A.~R. Hoare.
\newblock {\em Communicating Sequential Processes}.
\newblock Prentice Hall, 1985.

\bibitem[JS93]{JungA:stufam}
Achim Jung and Allen Stoughton.
\newblock Studying the fully abstract model of {PCF} within its continous
  function model.
\newblock In {\em Proc. Int. Conf. Typed Lambda Calculi and Applications},
  pages 230--245, Berlin, 1993. Springer-verlag.
\newblock Lecture Notes in Computer Science Vol. 664.

\bibitem[Loa96]{Loa96}
Ralph Loader.
\newblock Finitary PCF is undecidable.
\newblock {\em Theoretical Computer Science}, to appear, 2000.

\bibitem[Lor60]{LorenzenP:logua}
P.~Lorenzen.
\newblock Logik und agon.
\newblock In {\em Atti del Congresso Internazionale di Filosofia}, pages
  187--194, Firenze, 1960. Sansoni.

\bibitem[Lor61]{LorenzenP:eindk}
P.~Lorenzen.
\newblock Ein dialogisches {K}onstruktivit\"{a}tskriterium.
\newblock In {\em Infinitistic Methods}, pages 193--200, Warszawa, 1961. PWN.
\newblock Proceed. Symp. Foundations of Math.

\bibitem[Mal93]{MalacariaP:frogi}
Pasquale Malacaria.
\newblock Dalle macchine a ambienti alla geometria dell'interazione.
\newblock Unpublished manuscript, 1993.

\bibitem[MH99]{MH99}
Pasquale Malacaria and Chris Hankin.
Non-Deterministic Games and Program Analysis: an application to security.
In {\em Proceedings of LiCS `99}, pages 443--453, 1999.

\bibitem[Man76]{ManesE:algt}
E.~Manes.
\newblock {\em Algebraic Theories}, volume~26 of {\em Graduate Texts in
  Mathematics}.
\newblock Springer-Verlag, 1976.

\bibitem[Mil77]{MilnerR:fulamt}
Robin Milner.
\newblock Fully abstract models of typed lambda-calculi.
\newblock {\em Theoretical Computer Science}, 4:1--22, 1977.

\bibitem[Mul87]{MulmuleyK:fulase}
K.~Mulmuley.
\newblock {\em Full Abstraction and Semantic Equivalence}.
\newblock MIT Press, 1987.

\bibitem[Nic94]{NickauH:hersf}
H.~Nickau.
\newblock Hereditarily sequential functionals.
\newblock In {\em Proceedings of the Symposium on Logical Foundations of
  Computer Science: Logic at St.\ Petersburg}, Lecture notes in Computer
  Science. Springer, 1994.

\bibitem[OR95]{OHearnPW:krilrp}
Peter~W. O'Hearn and Jon~G. Riecke.
\newblock Kripke logical relations and {PCF}.
\newblock {\em Information and Computation}, 120(1):107--116, 1995.

\bibitem[Plo77]{PlotkinGD:lcfcpl}
Gordon Plotkin.
\newblock {LCF} considered as a programming language.
\newblock {\em Theoretical Computer Science}, 5:223--255, 1977.

\bibitem[Soa87]{SoareRI:recesd}
R.~I. Soare.
\newblock {\em Recursively Enumerable Sets and Degrees}.
\newblock Perspectives in Mathematical Logic. Springer-Verlag, Berlin, 1987.

\bibitem[Sto88]{StoughtonA:fulamp}
A.~Stoughton.
\newblock {\em Fully abstract models of programming languages}.
\newblock Pitman, 1988.

\bibitem[Win93]{WinskelG:fspl}
G.~Winskel.
\newblock {\em The Formal Semantics of Programming Languages}.
\newblock Foundations of Computing. The MIT Press, Cambridge, Massachusetts,
  1993.
\end{thebibliography}
\end{document}